\documentclass[journal, one column, 12pt]{IEEEtran}

\usepackage{ifpdf}
\usepackage{color}
\usepackage{subfigure}
\usepackage{subfloat}
\usepackage{url}
\usepackage{mathtools}
\usepackage{cite}

\ifCLASSINFOpdf
  \usepackage[pdftex]{graphicx}
\else
\fi
\usepackage{epstopdf}
\usepackage{epsfig}

\usepackage{amsthm}
\usepackage{array}

\usepackage[font=footnotesize]{subfig}

\usepackage[font=footnotesize]{subfig}
\newtheorem{definition}{Definition}
\newtheorem{obs}{Observation}
\newtheorem{theorem}{Theorem}
\newtheorem{lem}{Lemma}
\newtheorem{rmk}{Remark}
\newtheorem{cor}{Corollary}
\newtheorem{example}{Example}
\newtheorem{assum}{Assumption}
\begin{document}

\title{The value of Side Information in the Secondary Spectrum Markets}


\author{
Arnob~ Ghosh\IEEEauthorrefmark{1},
 Saswati~Sarkar\IEEEauthorrefmark{1}\thanks{\IEEEauthorrefmark{1}
 The authors are with Electrical and Systems Engineering Dept., University of Pennsylvania, USA; Their E-mail Ids are: arnob@seas.upenn.edu, swati@seas.upenn.edu},
 Randall~Berry\IEEEauthorrefmark{2}\thanks{\IEEEauthorrefmark{2} The author is with Electrical Engineering and Computer Science Department of Northwestern University, USA; The e-mail Id is:rberry@ece.northwestern.edu}
 }




\maketitle
\begin{abstract}
We consider a secondary spectrum market where primaries set prices for their unused channels. The pay-off of a primary then depends on the availability of channels for its competitors, which a primary might not have information about. We study a model where a primary can acquire this competitor's channel state information (C-CSI) at a cost.   We formulate a game between two primaries, where each primary decides whether to acquire the C-CSI or not and then selects its price based on that. We first characterize the Nash Equilibrium (NE) of this game  for a symmetric model where the C-CSI is perfect. We show that the payoff of a primary is independent of the C-CSI acquisition cost. We then generalize our analysis to allow for imperfect estimation and cases where the two primaries have different C-CSI costs or different channel availabilities. Our results show interestingly that   the payoff of a primary increases when there is estimation error. We also show that surprisingly, the expected payoff of a primary may decrease when the C-CSI acquisition cost  decreases or primaries have different availabilities.  
\end{abstract}
\section{Introduction}
\begin{table*}
\caption{\small Frequently used Notations}
\begin{tabular}{| p{3cm} | p{13cm} | }
\hline
Notation & Significance \\ \hline
$v$ & The highest price that a secondary is willing to pay for an available channel.\\ \hline
$q_i$  & Availability probability of primary $i$. \\ \hline
$q$ & In the basic model $q_1=q_2$. \\ \hline
$s_i$ & The C-CSI acquisition cost for primary $i$. \\ \hline
$s$ & In the basic model $s_1=s_2=s$.\\ \hline
$q_s$ & The C-CSI estimation accuracy. \\ \hline
$c$ & The transaction cost which the primary incurs only when it sells its channel. \\ \hline
\end{tabular}
\label{tab:notations}
\vspace{-0.6cm}
\end{table*}
Spectrum sharing where license holders (primaries) allow unlicensed users (secondaries) to use their channels can enhance the efficiency of the spectrum usage. However, secondary access will only proliferate when it is rendered profitable to the primaries. We investigate a secondary spectrum market where   there are {\em competing primaries} that want to lease their unused channels to  {\em secondaries} in lieu of financial remuneration. In our setting, primaries can be wireless service providers or any other  intermittent users of the spectrum (e.g. government agencies, TV broadcasters) or the infrastructure (e.g. access point owners) and secondaries can also be service providers or individual users. \begin{color}{red}We assume the market operates in fixed time intervals. At the start of each interval, the primaries announce prices for their channels if they are {\it available}\footnote{\begin{color}{red}The time-scale at which this market operates could range from seconds to hours depending on the underlying technology. The key assumption being that if a primary puts its channel up for sale, it commits to allowing the secondary to use that channel for the next interval.\end{color}}.\end{color} Each secondary seeks to buy an available channel with the lowest price.  

The availability of a channel varies randomly because of the usage statistic of a primary,  e.g if a primary need to use a channel to serve its own traffic, the channel will be {\it unavailable for sale}\footnote{Availability could also depend on channel fading or interference levels over the time-scale at which the market operates.}.  When its channel is available, the profit or payoff that a primary can obtain will depend in part on the availability of the channels of other primaries in the market.   

We consider a scenario where the primary can estimate the competitor's channel state information (C-CSI) by incurring a cost. This C-CSI provides an estimate of the competitor's channel state (CCS), which indicates if it is available for sale or not. The motivation behind considering such a scenario is the growing interest in incorporating spectrum measurements into various approaches for dynamic spectrum access. For example, a recent proposal is made by FCC to sense the occupancy of the 3.5 GHz band by environmental sensing capability (ESC) operators (e.g. Google, KeyBridge)\cite{escweb}. The cost to acquire the C-CSI can be incurred in several ways.  For example, i) the primary could need to devote resources throughout the location to sense the competitor's channel and estimate its traffic patterns; this could cost the primary in terms of power consumption or processing resources. ii) Alternatively, the primary could incentivize users to collaboratively crowdsource such measurements using their own wireless devices (see e.g. \cite{Fatemieh,Shi}). Here the cost to the primary is the payments to incentivize participation. iii) Finally, these estimates could be performed by a \lq{}third party'  that deploys a network of wireless spectrum sensors as in \cite{Petrioli}\footnote{Another example of such a 3rd party sensing network is the ESC operators (e.g. Google, KeyBridge) in the 3.5 GHz  as mentioned before.}
 and uses these to sell \lq\lq{}forecast" of the CCS to a primary.
Additionally, there may be errors in this estimation due to traffic variability,  noise and hidden terminal effects. Characterizing how the C-CSI costs and errors impact the competition between the primaries is the goal of this paper. 

 We now illustrate the challenges involved in analyzing this setting. A primary with an available channel needs to select whether to acquire the C-CSI  and the price for its channel.  However, while taking its own decision, a primary does not know whether its competitors decided to acquire the C-CSI or not.  The importance of C-CSI  is immense. For example,  if the competitors' channels are not available, a primary can sell its channel at the highest possible price due to a lack of competition. This suggests that a primary would want to acquire C-CSI. However, if the other primaries are available,  the primary may have to select lower price because of the competition.   This lower price may not be enough to cover the C-CSI acquisition cost.  Hence, it is not a priori clear whether primaries will acquire the C-CSI.

The inherent uncertainty in the competitors\rq{} decisions also complicates the pricing strategy of a primary.   We now illustrate this when there are two primaries. If one primary ($A$) knows that the channel of the other primary ($B$) is available, $A$'s pricing decision depends on if $B$ also knows $A$'s availability or not, as this will in turn impact the competition it faces. On the other hand, if primary $A$ does not know $B$'s channel state, then it must balance between selecting a lower price, which will increase its probability of selling should the competitor be present, and a higher price, which will give it more revenue should $B$ be absent. 

To address these issues, we focus on a market with two primaries (though some extensions to more primaries are addressed in Section VII).  The restriction to two primaries is mainly done to simplify our analysis:- however,  there are practical situations where this is a reasonable assumption\footnote{For example, if the primaries are wireless service providers, then in many places this market is dominated by only two such providers. Hence, wireless market with two players have been considered in \cite{duantmc,kimjsac,kiminfo,ren}.} We study  a non-cooperative game with the primaries as players.  When its channel is available, each primary decides i) whether to acquire  the C-CSI  or not, and  ii) a price.  If the primary acquires the C-CSI, it may select different prices depending its estimate of the CCS. When the primary does not acquire the C-CSI, it has to select a price irrespective of the CCS. We characterize the Nash equilibrium (NE) strategies. 

{\em Basic Model}: We first consider a {\em basic model} (Section~\ref{sec:basic_model}) where the acquired C-CSI is equal to CCS. The channel availability probability and the C-CSI acquisition costs are  the same for both the primaries. We introduce a $[T,p]$ class of strategies (Definition~\ref{defn:classtp}) and show that the NE strategy profile is of this form. In a $[T,p]$  strategy, a primary acquires the C-CSI,  when the cost is below $T$ with probability (w.p.) $p$, and does not acquire the C-CSI, otherwise. We allow the primaries to potentially randomize their prices given the C-CSI acquired. Using this characterization, several interesting properties of the NE are then shown. First, $p$ is increasing as the C-CSI cost decreases, but a primary never acquires C-CSI w.p.1. Second, $T$ increases in the uncertainty of the CCS. Third, we have the following counter-intuitive result: {\it the ability to acquire  C-CSI does not impact the expected payoff of a primary}. Finally, we show that the equilibrium pricing strategy is indeed to randomize and we characterize the resulting price distributions. 
 
   {\em Impact of the Estimation Error}:  We, subsequently, investigate the impact of  C-CSI  errors on the decision, payoff and the pricing  strategy of a primary (Section~\ref{sec:errorestimation}). Conventional wisdom might suggest that errors in estimating the CCS should decrease the payoff. However, this is not definite because errors might also make primaries less aggressive in lowering their prices to compete against competitors when they acquire C-CSI, which could lead to higher profits. Interestingly, we show that this is indeed the case, i.e., the primary's pay-off is higher with C-CSI errors.
    
 We show that there again exists a $[T,p]$ type NE strategy, where the threshold $T$ decreases as the estimation error  increases. Intuitively, increases in the error in estimating the CCS makes the  acquisition of the C-CSI less attractive for larger costs.   The probability $p$ again increases as the C-CSI cost decreases, but now the expected pay-off of a primary also increases. 
 
The NE pricing strategies are again randomized. In the basic model, when the primary accurately estimates that the competitor's channel is unavailable, it selects the highest possible price $v$ w.p. $1$. However, when there are errors in C-CSI, the competitor's channel may be available even when it estimates that it is not. Thus, a primary also selects a lower price.

{\em Impact of Unequal C-CSI acquisition costs}: We also investigate the setting where C-CSI  acquisition costs may differ across the primaries (Section~\ref{sec:unequalcost}). 
We show that $[T,p_i]$ strategies are NE for primary $i$. Each primary has the same threshold $T$ but different probabilities $p_i$, where $p_i$ is larger for the primary with the lower cost. The expected payoff of the primary with a higher cost  is the same as it would have obtained if there was no provision of acquiring the C-CSI. In contrast to the basic model, the expected payoff of the primary with a lower cost  is higher compared to the other primary. 

 The primaries again randomize their prices, where the primary with a lower cost chooses its price from a larger range when it acquires the C-CSI (and a smaller range when it does not). In contrast to the basic model, the primary with a higher acquisition cost also has a point mass at the highest price in the price distribution when it does not acquire the C-CSI, i.e., the primary with a higher acquisition cost selects higher prices with higher probabilities when it does not acquire the C-CSI. 

{\em Impact of heterogeneous availabilities}: We next consider the impact of different availability probabilities across the primaries  (Section~\ref{sec:unequalavail}).
  Again, we show that the NE strategy  is of the form $[T,p]$. In this case, the primary with a higher availability has a higher threshold and a higher probability of acquiring the C-CSI. The expected pay-off of the primary with a higher availability is greater, and interestingly, the expected pay-off of the other primary   {\em decreases} as the cost for acquiring the C-CSI decreases which negates the {\em conventional wisdom that the payoff of a primary should not decrease as the cost of acquiring the C-CSI decreases}.  We also  show that  the pricing strategy of each primary is randomized over a given interval. However, the primary which has a higher availability probability selects a price from a larger (smaller, resp.) interval when it acquires (does not acquire, resp.) the C-CSI.  



{\em Related Literature}: Price selection in oligopolies has been extensively investigated in economics 
(dating back to the classic work of Bertrand\cite{bertrand}) as well as in the wireless setting. For wireless applications,   we  divide the entire genre of works in two parts: i) Papers in which prices are set via an auction (e.g. \cite{sengupta,Xu}), and ii) Papers, such as ours, which model  price competition as a non co-operative game (\cite{Ileri, Mailespectrumsharing, Mailepricecompslotted, Xing, Niyatospeccrn, Niyatomultipleseller, kavurmacioglu,jia,yang,yitan,zhang,lin,duan,kim,Gaurav1,isit,Janssen,ciss, recentinforms, Osborne, Kreps}).  To the best of our knowledge, this paper is the first to include the option of acquiring the C-CSI in the strategy space of the players.   Additionally, compared to  the first category of papers,  our model is readily scalable and a central auctioneer is not required. Most of the papers in the second category considered that the primary is aware of the competitor's channel state, which is also the case in the classic Bertrand model \cite{bertrand}.  The exceptions are some recent papers \cite{Gaurav1,isit,arnob_ton,Janssen,recentinforms} which considered that the primaries can not acquire the C-CSI, so that each primary selects its  price not knowing the CCS\footnote{ \cite{recentinforms} considered that each player has private information such as capacity  which {\em is unknown to the} competitor. Thus, this setting is equivalent to the setting where a primary can not know  the channel availability of its competitor.}.  In contrast, we consider that  primaries have the option of acquiring the C-CSI.  A primary now needs to judiciously decide whether to acquire the C-CSI  and selects a price based on the result of this decision.   In our setting, a primary ($A$, say)  is also unaware whether the other primary has  acquired the CSI of $A$, while in \cite{Gaurav1,isit,arnob_ton,Janssen,recentinforms} the primary $A$ knows that its channel state is unknown to other primaries.   Naturally, these papers did not consider  the impact of the C-CSI acquisition costs, estimation error and different channel availability probabilities.

\section{System Model}
We consider a secondary spectrum market with two primaries (players) and one secondary\footnote{If there are more than one secondary, then the decision of the primary is trivial, it will always sell its channel, thus, it will select the highest possible price and will never acquire the CSI. Our model can also accommodate the setting where the number of secondaries is not known a priori.}. We first provide the basic system model in Section~\ref{subsec:basic_model} and  subsequently, we specify certain generalizations of the model in Section~\ref{sec:generalization}. Commonly used notations are given in Table~\ref{tab:notations}.

\subsection{Basic Model}\label{subsec:basic_model}

\begin{color}{red}We consider a model in which spectrum leases occur over a sequence of fixed time-slots and focus on one such time-slot\footnote{Of course, in practice primaries and secondaries may take a longer term view, requiring a dynamic game model, but we leave such considerations for future work.}. The duration of a time-slot could range from minutes to hours depending on the underlying technologies and other considerations (e.g., the overhead in running the market). Here the main consideration is that if a primary announces a price for its channel, it commits to allowing secondary usage of the channel for the next time-slot.  The availability of a primary's channel for sale will depend in part on the primary's own traffic. We define a channel to be in state $1$ if it is available and otherwise, it is in state $0$.
\footnote{Availability could also depend on the rate available to a secondary being large enough, which in turn might depend on estimates of fading and interference levels over the next time-slot.  These can also be viewed as part of the CCS, though in such cases a better model might be to allow for the state of the channel to vary with the rate as in \cite{isit,arnob_ton} (see also Sect. VII).}   Each primary's channel is available w.p. $q$, where $1>q>0$ and $q$ is common knowledge. If C-CSI is not acquired, a primary is unaware of the realized state of its competitor.\footnote{In other words, this is viewed a Bayesian game in which the type of a primary is its channel state, and $q$ is the belief that each primary has about the type of its competitor, where here all beliefs are consistent and reflect the true type distribution. }\end{color}

 If a primary's channel is available, the primary can sell it for secondary use during the next time slot. In this case, it  decides whether to acquire the  C-CSI before deciding the price for its available channel. For example, if the market opens at time $t$, then at time $t-\delta$, a primary decides whether to acquire the C-CSI or not. By acquiring the C-CSI, a primary obtains an estimate of the competitor's channel state (CCS) for the entire duration of the slot. The C-CSI estimation  is accurate and thus, a primary knows the exact CCS.
  The primary incurs a cost $s$ if it estimates the C-CSI.\footnote{Note that a primary (specially, if it is a wireless service provider) may already have 
some knowledge of its competitor (especially long-scale trends obtained via market research). However, in our scenario 
the primary needs to estimate the CCS  on a smaller time scale (i.e., the slot duration at which the market operates). This short time-scale CCS is not readily available to operators.  Hence, the primary needs to incur a separate cost to acquire the C-CSI.}

Each primary then decides its price and posts it to the secondary at the beginning of the slot.\footnote{Note that if the channel of a primary is unavailable, then the secondary will never buy the channel irrespective of its price}  If the channels of both the primaries are available for sale, then, the secondary will buy the lower priced channel. If the two available channels have the same price, then a secondary will choose either of them w.p. $1/2$.


\subsection{Generalization of the Model}\label{sec:generalization}
\subsubsection{Estimation Error}\label{sec:estimation_model} When a primary acquires the C-CSI, it estimates the CCS for the entire slot duration. Because of the channel fading, noise in the environment, the random variation of the usage pattern of the channel, that estimation may be erroneous.  In Section~\ref{sec:errorestimation}, we consider such a setting where the estimated CCS is accurate only with probability $q_s$. Specifically, if a primary acquires the C-CSI, then, it will estimate that the CCS is $1$ ($0$, resp.) w.p. $q_s$ if the original CCS is $1$ ($0$, resp.). Without loss of generality, we assume that\footnote{If $q_s=1/2$, then there is no point of estimating the CCS as the setting becomes equivalent to the setting where a primary does not know the channel state of its competitor.} $1/2<q_s\leq 1$.   Note that when $q_s=1$, there is no estimation error and a primary  accurately estimates the CCS, thus, the basic system model  is a special case of this model.

\subsubsection{Different Costs of Acquiring the C-CSI}\label{sec:differentcost_model}
  In Section~\ref{sec:unequalcost} we generalize the basic model to allow each primary $i$ to incur a different cost, $s_i$  for acquiring the C-CSI.

\subsubsection{Different Channel Availability Probabilities}\label{sec:differentavail_model}
 We generalize the basic model in Section~\ref{sec:unequalavail} by allowing each primary $i$ to have different availabilities $q_i$. 

\subsection{Payoff of a primary}
 If primary $i$ sets its price at $x$ and it decides to acquire the C-CSI, then, its payoff is
\begin{align}
\begin{cases}
x-c-s_i,\quad \text{if the primary is able to sell its channel,}\nonumber\\
-s_i,\quad \text{otherwise.}
\end{cases}
\end{align}
Note that when both the primaries incur the same cost to acquire the C-CSI, then we have $s_i=s$. 

When a primary does not acquire the C-CSI, then its payoff at price $x$ is
\begin{align}
\begin{cases}
x-c,\quad \text{if the primary is able to sell its channel,}\nonumber\\
0, \quad \text{otherwise.}
\end{cases}
\end{align}
\subsection{Strategy of a Primary}
If the channel of a primary is available\footnote{If the channel of the primary is unavailable, then its decision is immaterial.}, it will take a decision $D\in \{Y,N\}$ where $Y$ denotes incurring the cost $s$ to estimate the C-CSI and $N$ denotes not acquiring the C-CSI.  Primary $i$ also sets a price for its available channel. Note that the primaries' decisions are simultaneous so that no primary is aware of the decision of its competitor when making its own decision.  If a primary selects $Y$,  it selects a price using either a distribution  $F_1(\cdot)$ or $F_0(\cdot)$ depending on whether it estimates the CCS as $1$ or $0$, respectively.   If a primary selects $N$, then it does not acquire the C-CSI, so it only selects its price using a single distribution $F(\cdot)$. 
\begin{definition}
The strategy $S_{i}$ of primary $i=1,2$ is $\sigma(D,\mathbf{F})$ where $\mathbf{F}=(F_0,F_1)$ when $D=Y$, $\mathbf{F}=(F,F)$ when $D=N$, and $\sigma(D,\mathbf{F})$ is a probability mass function over the strategies $(D,\mathbf{F})$.

The strategy of the primary other than $i$ is denoted as $S_{-i}$.
\end{definition}
\begin{definition}
$E\{u_{i}(S_i,S_{-i})\} $ denotes the expected payoff of primary $i$ when its channel is available, it uses strategy $S_{i}$ and the other primary uses strategy\footnote{Note that we consider the expected payoff of a primary as the expected payoff conditioned on the channel of the primary being available. Naturally if the channel of the primary is unavailable, it will attain a payoff of $0$ and so its unconditional expected payoff is simply this quantity scaled by $q_i$.} $S_{-i}$. 
\end{definition}
\subsection{Solution Concept}
We consider a non-cooperative game where each primary only wants to maximize its own expected payoff. We use the (Bayesian) Nash Equilibrium as a solution concept.
\begin{definition}
 A \emph{Nash  equilibrium} (NE)  $(S_1,S_2)$ is a strategy profile such that no primary can improve its expected payoff by unilaterally deviating from its strategy\footnote{\begin{color}{red}In the language of Bayesian games,  the expectation here is with respect to  a player's belief about the availability of the competitor's channel (which as we have noted can be viewed as the competitor's type). The belief about the competitor's type (i.e. availability) also changes depending on the information the primary has. Given the consistent belief assumption, we simply refer to the resulting Bayesian Nash equilibrium as a NE in the following.\end{color} }  \cite{mwg}. Thus, 
\begin{align}
E\{u_{i}(S_{i},S_{-i})\}\geq E\{u_{i}(\tilde{S}_{i},S_{-i})\} \ \forall \ \tilde{S}_{i}.
\end{align}
A strategy profile is symmetric if $S_{i}=S_{j}$ for any pair of players $i$ and $j$. 
\end{definition}

\section{Results of the Basic Model}\label{sec:equalcost}
We, first, investigate the system model depicted in Section~\ref{subsec:basic_model}. Note that this setting is a special case of each of the more generalized settings depicted in Sections~\ref{sec:estimation_model},\ref{sec:differentcost_model}, and \ref{sec:differentavail_model}.

\subsection{Goals}
Acquiring the CSI of the competitor has potential advantages. For example, if a primary knows that the channel of its competitor is unavailable, then, the primary can select a  high price because of the lack of competition. However, a primary has to incur a cost to acquire the CSI. Thus, conventional wisdom suggests that as the cost of acquiring the CSI decreases, a primary should more frequently acquire the CSI and thereby gain a higher payoff in an NE. However, conventional wisdom is not definitive because of the following. The payoff of a primary ($1$, say) also inherently depends on the decision  of other primary ($2$,say). If the primary $2$ decides to acquire the CSI of primary $1$, then primary $2$ selects a lower price when the channel of primary $1$ is available, thus, in response\footnote{In an NE, each player selects a best response strategy in response of the strategy of the other player.}, the primary $1$ also selects a lower price in the NE which reduces its payoff. On the other hand, acquiring the CSI of the competitor is also not ruled out either.  This is because a primary may acquire the CSI of its competitor and take advantage of the extra information.  Thus, it is not apriori clear whether a primary will acquire the CSI of its competitor. It is also not clear even if a primary decides to acquire the CSI  at what values of $s$ it will do so.   We resolve all the above quandaries.

The inherent uncertainty in the competitor\rq{}s decision also complicates the pricing strategy of the primary.   If primary $1$ knows that the channel of primary $2$ is available, its pricing decision still depends on if primary $2$ also know that its  channel  is available; if not then the primary $2$ may randomize among multiple prices, enabling primary $1$ to charge a higher price. If primary $2$ knows that the channel of primary $1$ is available, primary $2$ selects a lower price, in response primary $1$ also selects a lower price. On the other hand,   if the primary does not know the channel state of its competitor, then it may have to randomize over prices from an interval which is not known apriori.  Thus, it is also not apriori clear how a primary will select its price. We also characterize NE pricing strategies.

\subsection{Results}
\subsubsection{A class of Strategy for selecting $Y$}
We first define a class of strategies for  selecting $Y$. 
\begin{definition}\label{defn:classtp}
A $[T,p]$ strategy is a strategy where a primary selects 
\begin{align}
\begin{cases}
Y, \quad \text{w.p. } p \quad \text{when } s<T\nonumber\\
N, \quad \text{w.p. } 1 \quad \text{when }s\geq T
\end{cases}
 \end{align} for $0<p\leq 1$. The probability $p$ may be a function of $s$.  
\end{definition}
We show that in the basic model as well as in different generalizations, the NE strategy  is a $[T,p]$ strategy where $p$ is a strictly decreasing function of $s$. We also characterize  $T$, and $p$ in different generalizations of the basic model.

It is intuitive that in an NE, a primary will choose $Y$ with a probability $p=f(s)$ where $f(\cdot)$ is a decreasing function. It is also intuitive that $f(s)=0$ when $s>(v-c)$ as the maximum expected payoff that a primary attains is $v-c$. We, however, show that $f(s)$ can be $0$ even for smaller values of $s$. We also show that $p$ never becomes $1$ for any positive value of $s$.  We fully characterize the function $f(\cdot)$ and the value of the threshold $T$ above which a primary does not select $Y$.
\subsubsection{Main Results}
Our main results are--
\begin{itemize}
\item Regardless of the cost $s$,  there is no NE where both the players have  full knowledge of each other\rq{}s channel states w.p. $1$ (Theorem~\ref{thm:yandy}). There is no NE where  one primary has the complete knowledge of  the channel state of its competitor, but the other does not  (Theorem~\ref{thm:yandn}). Thus,  a primary can only select $Y$, if the other primary randomizes between $Y$ and $N$. 
\item We show that the {\em unique} NE strategy is a $[T,p]$ strategy  where $T=q(v-c)(1-q)$ (Theorems~\ref{thm:nandn} and \ref{thm:mixedstrategy}). Note that $T$ increases when the uncertainty of the availability of the channel increases i.e. $q$ becomes closer to $1/2$. Intuitively, when either $q$ is large or small, the uncertainty of the competitor\rq{}s channel decreases, thus, a primary selects $N$ for higher values of $s$. We also characterize the value of $p$ as a function of $s$ and show that $p$ is a decreasing function of $s$. 
\item  The expected payoff that a primary attains in any NE strategy profile is $(v-c)(1-q)$ (Theorems~\ref{thm:nandn} and \ref{thm:mixedstrategy}). Thus, the expected payoff of a primary is independent of the value of $s$. \cite{Gaurav1} shows that when a primary can not acquire the CSI of the competitor, then its payoff is $(v-c)(1-q)$. Thus, the provision of acquiring the CSI of the competitor does not impact the expected payoff of a primary. 
\item Theorem~\ref{thm:nandn} shows that when each primary selects $N$, then each primary randomizes its price from the interval $[\tilde{p},v]$. Theorem~\ref{thm:mixedstrategy} shows that when a primary selects $Y$ ($N$, resp.) and the channel of its competitor is available, then the primary selects its price from the interval $[\tilde{p}_1,\tilde{p}_2]$ ($[\tilde{p}_2,v]$,resp.). Intuitively, as the uncertainty of the availability of the competitors increases, a primary selects a higher price.  We also show that $\tilde{p}<\tilde{p}_1$. Thus, a primary selects its price from a larger interval when it randomizes between $Y$ and $N$. 
\end{itemize}
We now describe the results in details. We first state some price distributions $\phi(\cdot)$ and $\psi(\cdot)$ which we use throughout.
\begin{align}
\phi(x)=& 0\quad \text{if } x<\tilde{p}\nonumber\\
& \dfrac{1}{q}\left(1-\dfrac{(v-c)(1-q)}{x-c}\right)\quad \text{if } \tilde{p}\leq x\leq v\nonumber\\
& 1\quad \text{if } x>v.\label{eq:phi}\\
\psi(x)=& 
0 \quad \text{if } x<\tilde{p}\nonumber\\
& (1-\dfrac{(v-c)(1-q)}{x-c})\quad \text{if } \tilde{p}\leq x<v\nonumber\\
& 1-q, \quad \text{if } x=v\nonumber\\
& 1 \quad \text{if } x>v. \label{eq:psi}\\
& \text{where } \tilde{p}=(v-c)(1-q)+c.\label{eq:tildep}
 \end{align}
\subsection{Does there exist an NE where both primaries select $Y$?}
\begin{theorem}\label{thm:yandy}
There is no Nash equilibrium where both the primaries choose $Y$ w.p. $1$.
\end{theorem}
\textit{Outline of the proof}:
Assume  both players choose $Y$, so that they know each other\rq{}s channel state. Thus, the competition becomes similar to  {\em Bertrand Competition}\cite{mwg}, i.e. if the channel of its competitor is unavailable, then the primary will set its price at the $v$, otherwise it will set its price at the lowest value $c$. Now, the probability with which the channel of a primary is available is $q$. Thus, the expected payoff of a player is 
\begin{align}\label{eq:payoff_strategy}
& (v-c-s)(1-q)+(c-c-s)q
 \end{align} 
 Now consider the following unilateral deviation for a primary: Primary $1$ selects $N$ and sets its price at $v$ w.p. $1$. The channel of primary $1$ will be bought when the channel of primary $2$ is not available for sale. Since primary $1$  decides not to incur the cost $s$, thus, its expected payoff is 
 \begin{align}
 (v-c)(1-q)
 \end{align}
 This is strictly higher than (\ref{eq:payoff_strategy}). Hence, the strategy profile can not be an NE.\qed
 
  The above theorem means that there will be at least one primary which will be unaware of its competitor\rq{}s channel state with a non-zero probability.

%
%

\subsection{Does there exist an NE where one selects $Y$ and the other selects $N$?}
 \begin{theorem}\label{thm:yandn}
 For positive $s>0$, there is no NE where a primary selects $Y$ w.p. $1$ and the other selects $N$ w.p. $1$. 
 \end{theorem}
First, we provide the intuition behind the result. The primary (say, $1$) which selects $Y$ tends to select lower prices with higher probability when it knows that the channel of the other primary is available. Thus, in response the primary $2$ (which selects $N$)  selects higher prices with higher probabilities in order to gain a high payoff in the event that the channel of primary $1$ is unavailable since it knows that its probability of selling is very low in the event that the channel of primary $1$ is available. The primary $1$ can then gain a higher payoff  by selecting $N$ and higher prices as it does not have to incur the cost $s$. Hence, the primary $1$  has an incentive to deviate from its own strategy. The detailed proof is given below.
 \begin{proof}
 Without loss of generality, assume that primary $1$ selects $Y$ and primary $2$ selects $N$. First, we discuss the pricing strategies of primaries $1$ and $2$ and calculate the expected payoff of primary $1$, subsequently, we show that primary $1$ has an incentive to deviate. 
 
When primary $1$ knows that the channel of primary $2$ is not available, then primary $1$ will be able to sell its channel at the highest possible price, thus, it will select $v$ w.p. $1$ and its payoff if $(v-c)-s$. The above event occurs w.p. $1-q$.

 Now, we consider the case when the channel of primary $2$ is available. While deciding its price, primary $2$ only knows that the channel of primary $1$ is available w.p. $q$.  However, while selecting its price primary $2$ knows that the primary $1$ will know the channel state of primary $2$ if the channel of primary $1$ is available. Hence, when primary $1$ knows that the channel of primary $2$ is available, then the pricing decision becomes equivalent to the setting where primary $1$ knows that the channel of primary $2$ is available w.p. $1$ and primary $2$ knows that the channel of primary $1$ is available w.p. $q$.   The NE pricing strategy in the last setting has been studied in \cite{Gaurav1} and using Theorem 2 in \cite{Gaurav1} we have
\begin{lem}\label{lm:yn}
 Primary $1$ must select its price according to $\phi(\cdot)$ (given in (\ref{eq:phi})) and primary $2$ must select its price according to $\psi(\cdot)$ (given in (\ref{eq:psi})). 
\end{lem}
By Lemma~\ref{lm:yn} when the channel of primary $2$ is available for sale, then expected payoff of primary $1$ at any $\tilde{p}\leq x<v$
 \begin{align}\label{eq:x1}
 (x-c)(1-\psi(x))-s=(v-c)(1-q)-s.
 \end{align}
 At $x<\tilde{p}$, the payoff of primary will be strictly less than the expression in (\ref{eq:x1}). On the other hand at $v$, primary $1$ will get strictly a lower payoff compared to the payoff at a price just below $v$ since $\psi(\cdot)$ has a jump at $v$. Hence, the maximum expected payoff to primary $1$ in this case is $(v-c)(1-q)-s$.
 
 Thus, the expected payoff of primary $1$ is 
 \begin{align}\label{eq:payoff}
 (v-c)(1-q)+q(v-c)(1-q)-s.
 \end{align}
 Now, we show that if primary $1$ selects $N$, then the primary can achieve strictly higher payoff.   For $x\in [\tilde{p}, v)$, the expected payoff of primary $1$ at $N$ is
 \begin{align}
(x-c)(1-q\psi(x))= (x-c)(1-q)+q(v-c)(1-q)
 \end{align}
 Thus, for every positive $s$ there exists a small enough $\epsilon>0$ such that at $x=(v-c-\epsilon)$, it will attain strictly higher payoff than (\ref{eq:payoff}). Hence, if primary $1$ selects $N$ and the price $v-\epsilon$ w.p. $1$ then primary $1$ attains a strictly higher payoff. The result follows.
 \end{proof}

\subsection{Does there exist an NE where both primaries select $N$?}
 \begin{theorem}\label{thm:nandn}
 Suppose that each primary selects the  strategy $(N,\phi)$ ($\phi(\cdot)$ is given in (\ref{eq:phi})).
 The above strategy profile is the unique NE when $s\geq q(v-c)(1-q)$.\\
 However, the above is not an NE when $s<q(v-c)(1-q)$.
 \end{theorem}
 We  provide an intuition behind the result. When $s$ is high, if a primary selects $Y$, then it has to incur high cost compared to the potential gain it will achieve, thus, no primary has any incentive to deviate. When $s$ is low, if a primary  deviates and selects $Y$, then it can gain higher payoff by taking advantage of the CSI of the other primary. Thus, the strategy profile fails to be an NE when $s$ is low.  We prove that the strategy profile is an NE in Theorem~\ref{thm:nandnerror}. We show the uniqueness in Appendix.

{\em Remark}: The result shows that when the cost $s$ is high, in an equilibrium both the primaries select $N$. It is obvious that if $s>(v-c)$, then a primary will never opt for $Y$. The above theorem shows that even if $s\geq (v-c)q(1-q)$, primaries will select $N$.  
\subsection{Does there exist an NE when $s$ is low?}
Note from Theorems~\ref{thm:yandy}, \ref{thm:yandn} and \ref{thm:nandn} that if $s$ is low, then there is no NE strategy where each primary selects either $Y$ or $N$ w.p. $1$. Thus, at least one primary must randomize between $Y$ and $N$ when $s$ is low.

Now, consider the following price distributions
\begin{align}
\psi_1(x)=\begin{cases}
0, \quad \text{if }x<\tilde{p}_1,\nonumber\\
\dfrac{1}{p}(1-\dfrac{\tilde{p}_1-c}{x-c})\quad \text{if } \tilde{p}_1\leq x\leq \tilde{p}_2,\nonumber\\
1, \quad \text{if } x>\tilde{p}_2.
\end{cases}
\end{align}
and
\begin{align}
\psi_2(x)=\begin{cases}
0, \quad \text{if }x<\tilde{p}_2\nonumber\\
\dfrac{1}{q(1-p)}(1-\dfrac{(v-c)(1-q)}{x-c}-qp)\quad \text{if } \tilde{p}_2\leq x\leq v\nonumber\\
1, \quad \text{if } x>v
\end{cases}
\end{align}
where $\tilde{p}_1$ and $\tilde{p}_2$ are
\begin{align}\label{eq:tildeps}
\tilde{p}_1=\dfrac{(v-c)(1-q)(1-p)}{1-qp}+c, \quad 
\tilde{p}_2=\dfrac{(v-c)(1-q)}{1-qp}+c
\end{align}
Note that both $\psi_1(\cdot)$ and $\psi_2(\cdot)$ are continuous.  In the following, we show that a strategy profile based on these distribution is a NE when $s$ is small enough.

\begin{theorem}\label{thm:mixedstrategy}
Consider the following strategy profile: Each primary selects $Y$ w.p. $p$ and $N$ w.p. $1-p$ where $p=\dfrac{q(v-c)(1-q)-s}{q(v-c)(1-q)-sq}$. When choosing $Y$, the primary selects its price according to $\psi_1(\cdot)$ when it knows that the channel state of the other primary is available, otherwise it selects $v$ w.p. $1$. When choosing  $N$, the primary selects price according to $\psi_2(\cdot)$. \\
The above strategy profile is the unique  NE if $s<q(v-c)(1-q)$. 
The expected payoff that a primary attains in the NE strategy profile is $(v-c)(1-q)$.
\end{theorem}
The proof of the strategy profile is an NE is similar to the Theorem~\ref{thm:errormixedstrategy} where we consider the CCS estimation may not be accurate. We prove the uniqueness in the Appendix. 
\begin{figure*}
\begin{minipage}{0.28\linewidth}
\begin{center}
\includegraphics[width=0.99\textwidth]{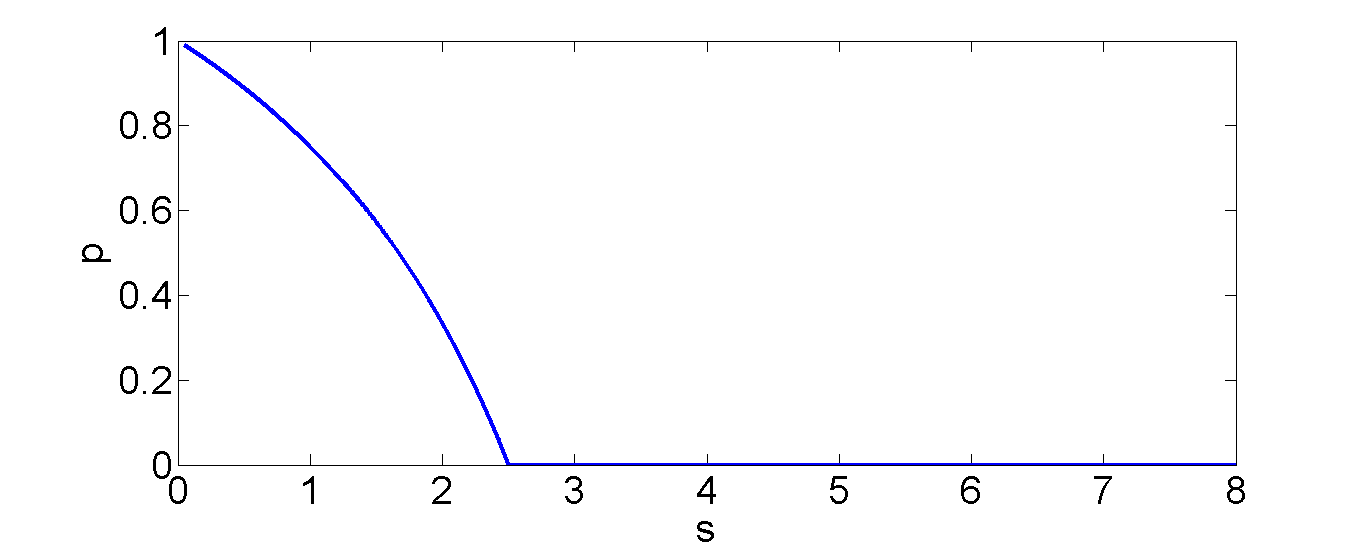}
\caption{\small Variation of $p$ as a function of $s$ in an example setting: $v=11, c=1, q=0.5$. When $s\geq 2.5$, $p=0$. $p\rightarrow 1$ as $s\rightarrow 0$.}
\label{fig:c1}
\vspace{-0.5cm}
\end{center}
\end{minipage}\hfill
\begin{minipage}{0.28\linewidth}
\begin{center}
\includegraphics[width=0.99\textwidth]{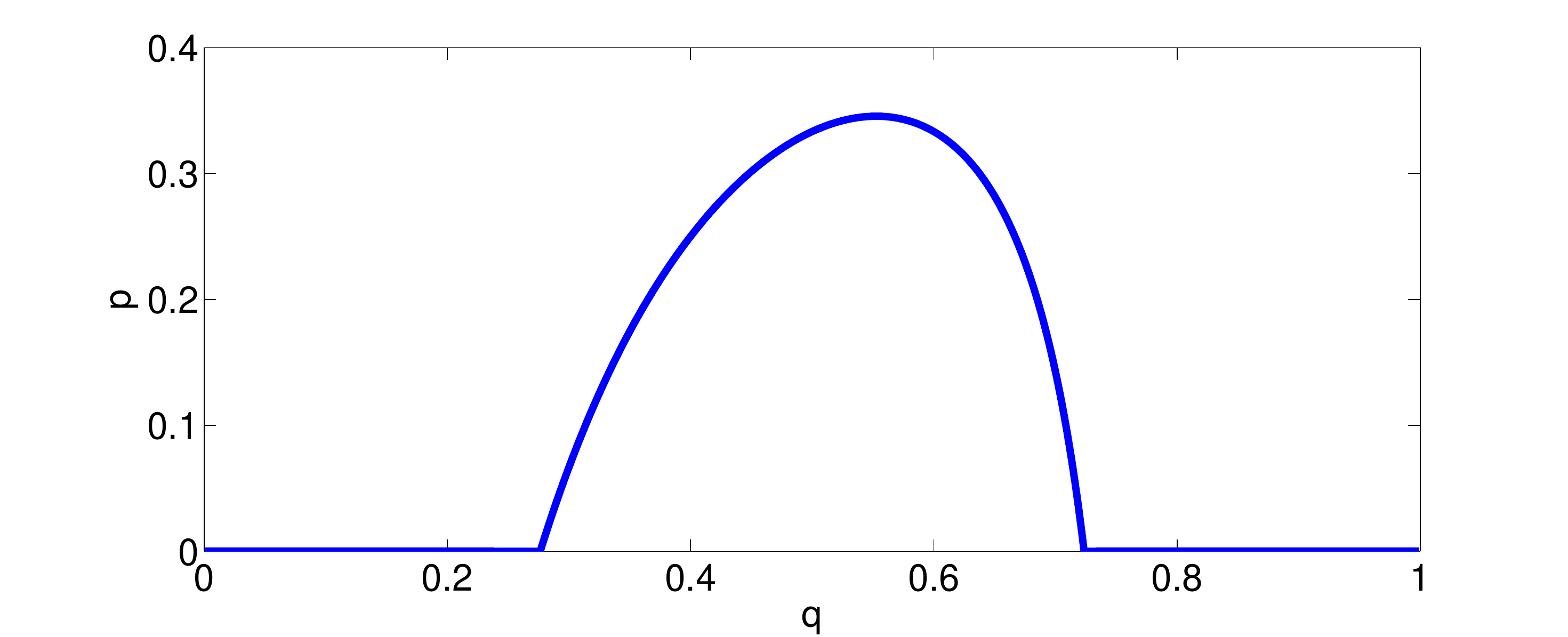}
\caption{\small Variation of $p$ as a function of $q$ in an example setting: $v=11, c=1, s=2$. 
When either $q\leq 0.28$ or $q\geq 0.72$, $p=0$. $p$ is maximized at $q^{*}=0.55$,  maximum value of $p$ is $0.35$. }
\label{fig:q}
\vspace{-0.5cm}
\end{center}
\end{minipage}\hfill
\begin{minipage}{0.4\linewidth}
\includegraphics[width=0.99\textwidth]{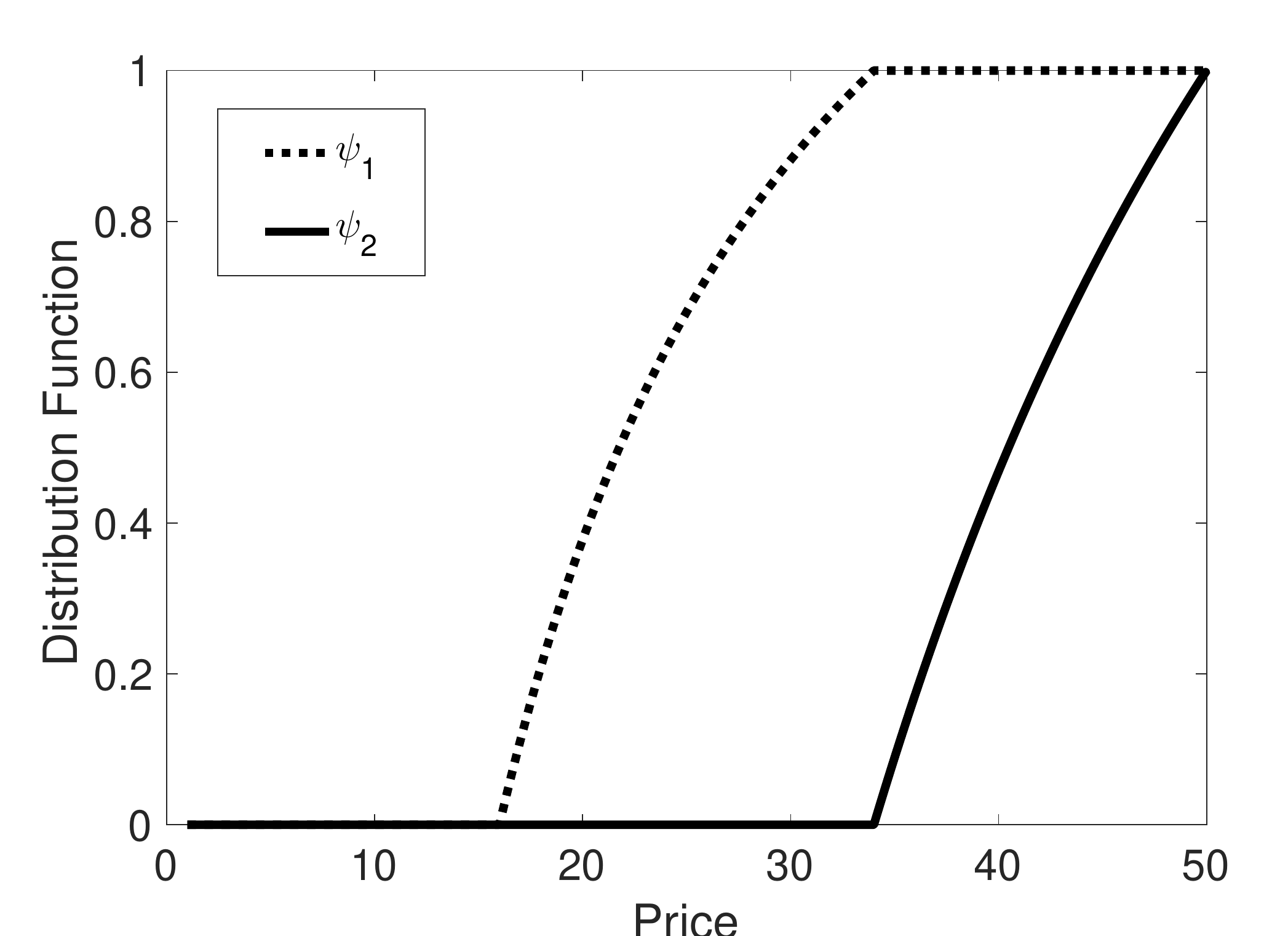}
\vspace{-0.1in}
\caption{\small The pricing strategies $\psi_1(\cdot)$ and $\psi_2(\cdot)$ in the scenario described in Example~\ref{ex:basic}. The support of $\psi_1(\cdot)$ and $\psi_2(\cdot)$ are respectively $[16,34]$ and $[34,50]$.}
\label{fig:ex_basic}
\vspace{-0.2in}
\end{minipage}
\end{figure*}

\begin{example}\label{ex:basic}
\begin{color}{red}
We illustrate the computation of the  NE streategy for an example where $v=50, c=0, s=8,$ and $q=0.5$. Since $s<q(v-c)(1-q)$,  the NE strategy is given by Theorem~\ref{thm:mixedstrategy} which gives $p=0.5294$. Using $p$,  the strategies $\psi_1(\cdot)$ and $\psi_2(\cdot)$ are readily  obtained and shown in  Fig.~\ref{fig:ex_basic}.
\end{color}
\end{example}

  {\em Discussion}: Note from the above theorem that when $s$ is low there exists an NE where {\em both the primaries randomize} between $Y$ and $N$. It is also easy to discern that as $s$ decreases, $p$ increases and as $s\rightarrow 0$, $p\rightarrow 1$ (Fig.~\ref{fig:c1}). Thus, when the cost of obtaining the competitor\rq{}s CSI decreases, then the primaries will be more likely to acquire that information. 

Note that $q(1-q)$ is the measure of uncertainty, if the uncertainty if higher (i.e. $q=1/2$), then the threshold is also higher. A primary never selects $Y$ if $s\geq (v-c)/4$. By differentiating, it is easy to discern that when $s<(v-c)/4$, then $p$ is maximized at $q^{*}=1-\sqrt{s/(v-c)}$ (Fig.~\ref{fig:q}). Since $s<(v-c)/4$,  $q^{*}>1/2$. Note also that $q^{*}$ decreases as $s$ increases. Intuitively, when $s$ increases,  primaries tend to select $Y$ only when the uncertainty of the availability of channel increases. 

The support set of $\psi_1(\cdot)$ is $[\tilde{p}_1,\tilde{p}_2]$ and $\psi_2(\cdot)$ is $[\tilde{p}_2,v]$.  Thus, under $Y$ a primary selects lower prices when the primary knows that the channel of its competitor is available compared to the setting where the primary is not aware of the channel state of its competitor.  This is because in the former case the uncertainty of the appearance of the competitor is reduced.

Since $p$ increases as $s$ decreases, thus, from (\ref{eq:tildeps}), $\tilde{p}_1$  increases as $s$ decreases. Thus, a primary selects its price from a larger interval when $s$ decreases. Also note that $\tilde{p}_2$ also increases as $s$ decreases. Thus, the support set of $\psi_1(\cdot)$ increases as $s$ decreases.

 Theorems~\ref{thm:nandn} and \ref{thm:mixedstrategy} imply that  the expected payoff of a primary is $(v-c)(1-q)$.  Note that when the primaries always know each other\rq{}s channel states,  the competition becomes equivalent to the Bertrand competition \cite{mwg} and the expected payoff is \footnote{It can also be obtained from (5).} $(v-c)(1-q)$ and   when the primaries are constrained to select only $N$, the expected payoff is again $(v-c)(1-q)$ \cite{Gaurav1,isit}. Hence, our result also builds the bridge between the two extremes. Specifically, it shows that the cost $s$ or  the availability of  the competitor\rq{}s CSI does not impact the expected payoff. 

\subsection{Welfare of the Secondaries}
\begin{figure}
\includegraphics[width=60mm,height=30mm]{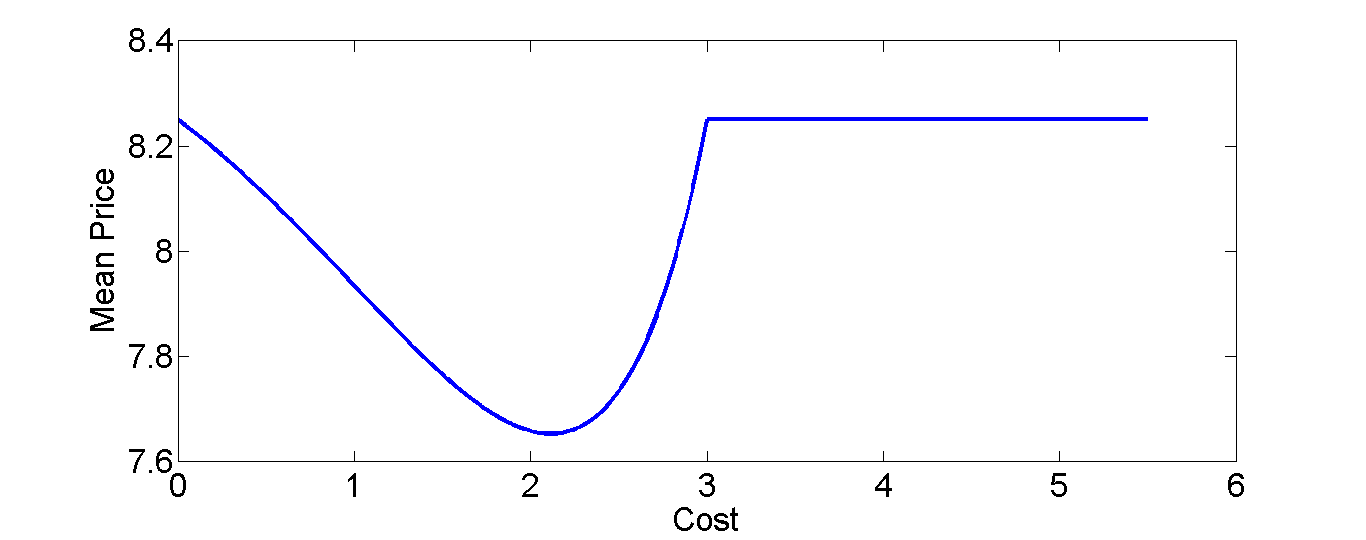}
\caption{Mean price paid by the secondary}
\label{fig:mean_price}
\vspace{-0.5cm}
\end{figure}
 Fig.~\ref{fig:mean_price} shows the variation of the expected price paid by the secondary. Initially, the expected price decreases as the C-CSI acquisition cost $s$ increases. The expected price reaches the minimum value, and then increases with the increase in $s$. When $s\geq q(v-c)(1-q)$ i.e. the primaries select $N$ w.p. $1$, the expected price is the same in the setting with $s=0$ i.e. when the primaries select $Y$ w.p. $1$. Fig.~\ref{fig:mean_price} shows that the expected price paid by the secondary is minimum at a positive cost; the minimum is not attained when $s=0$ which negates the conventional wisdom. Note that the expected payoffs of the primaries are independent of the cost $s$. Thus, the expected social welfare which is the sum of the expected payoffs of the primaries and the expected utility of the secondary (which is the negative of the price paid by the secondary) is in fact minimum at $s=0$. 

 \section{Impact of Error in the Estimation}\label{sec:errorestimation}
We, now, investigate the impact of the estimation error. Towards this end, we consider the system model specified in Section~\ref{sec:estimation_model}. Specifically, if a primary estimates the CSI of its competitor, the estimation is accurate only with probability $q_s$ ($1/2<q_s\leq 1$). 

\subsection{Goals}
The impact of error in the estimation on the decision and the payoff of each primary is not apriori clear. The conventional wisdom suggests that the error in the estimation should decrease the payoff. However, the conventional wisdom is not definitive because of the following. If there is an error  in estimating the channel state of the competitor, then, the primary $2$  selects a higher price even when it estimates that the channel state of the primary $1$  is $1$, thus, in response, the primary $1$ selects a higher price without reducing the winning probability, which may increase the payoff. It also remains to be seen  whether the expected payoff of a primary is independent of $s$ like in the basic model.  Even   if the selection of $Y$ belongs to the class $[T,p]$ (Recall Definition~\ref{defn:classtp}), the dependence of $T$ and $p$ on the estimation error is also not apriori clear.

The pricing strategy also depends on the estimation error.   For example, when the estimation error is $0$, then a primary selects a high price when the estimated channel state of the competitor is $0$ as the channel of the competitor is unavailable.  However, when there is an error in estimated channel state, the actual channel state may not be  $0$ even when the estimated channel state is $0$. Thus, a higher price may reduce the probability of winning and a lower price may reduce payoff in the event of a selling. Our goal is to characterize the pricing strategies of the primaries. 

\subsection{Main Results}\label{sec:mainresults_estimationerror}
We now summarize our main findings in this section here--
\begin{itemize}
\item We show that the NE strategy is a $[T,p]$ strategy (Definition~\ref{defn:classtp}) with $T=q(2q_s-1)(v-c)(1-q)$ (Theorems~\ref{thm:nandnerror}, \ref{thm:mixederror}). Note that in the basic model, we have also seen $[T,p]$ type strategy for selecting $Y$. However, due to the estimation error, the threshold is different compared to the basic model. The threshold decreases as $q_s$ decreases i.e. primaries select $N$ w.p. $1$ for larger values of $s$.  Intuitively,  the uncertainty regarding the channel state increases as the estimation error increases, thus,  the uncertainty of the channel state increases even when primary selects $Y$. A primary is more reluctant to select $Y$. Hence, primary selects $N$ w.p. $1$ for smaller values of $s$. We also characterize $p$ as a function of $s$ and show that $p$ decreases monotonically with $s$. 
\item The expected payoff of each primary is strictly higher than $(v-c)(1-q)$ when a primary randomizes between $Y$ and $N$ (i.e. $s<T$) and $q_s<1$. In the basic model, we have shown that the expected payoff of each primary is  $(v-c)(1-q)$ irrespective of the value of $s$.  Thus, the error in estimation increases the payoff of each primary which negates the conventional wisdom that expected payoff of a primary should increase as the estimation error decreases. The payoff of each primary also increases with the decrease in $s$ when $q_s<1$. Hence, in contrast to the basic model, the expected payoff of each primary depends on the value of $s$.
\item In NE pricing strategy:
\begin{itemize}
\item When a primary selects $Y$ and estimates that the channel state of its competitor is $1$, then it selects its price from the interval $[\tilde{p}_1,L_N]$.
\item When a primary selects $N$, then it selects its price from the interval $[L_N,L_0]$. 
\item When a primary selects $Y$ and estimates that the channel state of its competitor is $0$, then it selects its price from the interval $[L_0,v]$. 
\end{itemize}
If $q_s=1$, a primary always selects $v$ w.p. $1$ when the primary selects $Y$ and the channel state of the competitor is $0$ since the primary will always be able to sell its channel because of the unavailability of  its competitor. However, when $q_s<1$,  there is a potential error in the estimation, thus, a primary randomizes among prices from an interval $[L_0,v]$ even when the primary estimates that the channel state of its competitor is $0$. Also note that when a primary estimates that the channel state of its competitor is $1$ ($0$,resp.), then its competitor is more likely to be available (unavailable, resp.), hence, the primary selects lower (higher, resp.) prices compared to the setting where a primary selects $N$. 
\end{itemize}
\subsection{High $s$}
First, we state some results which we use throughout this section. Note that when a primary decides to estimate the CSI of its competitor, it estimates the channel state of its competitor is $1$ w.p. $qq_s+(1-q)(1-q_s)$ and the primary estimates the channel state of its competitor is $0$ w.p. $(1-q)q_s+q(1-q_s)$.  Note that when $q_s=1$, then the above probabilities becomes $q$ and $1-q$ respectively. If a primary estimates that its competitor\rq{}s channel state is $1$, then the actual channel state is $1$ w.p. 
\begin{align}\label{eq:1given1}
\dfrac{q_sq}{qq_s+(1-q)(1-q_s)}. 
\end{align}
Similarly, if a primary estimates that its competitor\rq{}s channel state is $0$, then the actual channel state of its competitor is $1$ w.p. 
\begin{align}\label{eq:0given0}
\dfrac{q(1-q_s)}{(1-q)q_s+q(1-q_s)}. 
\end{align}Note that when $q_s=1$, then both the above probabilities become $1$.

Our main result in this section shows that
\begin{theorem}\label{thm:nandnerror}
There exists a NE where each primary selects $N$ w.p. $1$ if $s\geq (v-c)(1-q)(2qq_s-q)$. In the NE pricing strategy, each primary selects its price according to $\phi(\cdot)$ (described in (\ref{eq:phi})). The expected payoff of each primary is $(v-c)(1-q)$.
\end{theorem}
\begin{proof}
We show that a primary does not have any profitable unilateral deviation when the other primary follows the strategy prescribed in the theorem. Towards this end,  we, first, show that under $N$ the maximum expected payoff of a primary is $(v-c)(1-q)$ (Step i).  It is attained when the primary follows the strategy $\phi(\cdot)$ (Step ii). Subsequently, we show that if the primary selects $Y$, then its expected payoff is at most $(v-c)(1-q)$ which will show that the primary does not have any profitable unilateral deviation (Step iii).  

Step i: At any price $x\in [\tilde{p},v]$ the expected payoff of a primary is 
\begin{align}\label{eq:payoff_n}
(x-c)(1-q\phi(x))=(v-c)(1-q) \quad (\text{from } (\ref{eq:phi}))
\end{align}
A price strictly less than $\tilde{p}$ will fetch a payoff strictly less than $(v-c)(1-q)$ (by (\ref{eq:tildep})). Thus, the maximum expected payoff of a primary under $N$ is $(v-c)(1-q)$.

Step ii: Note from (\ref{eq:payoff_n}) that a primary attains the maximum expected payoff when it selects its price from the interval $[\tilde{p},v]$.

Step ii:  Now, we show that if primary $1$ selects $Y$, it can not get a strictly higher payoff when $s\geq (v-c)(1-q)(2qq_s-q)$. Towards this end, we show that when a primary selects $Y$ and estimates that the channel state of the other primary is $1$, then it will attain a maximum expected payoff of $(v-c)(1-q)-s$ (Step ii.a). Subsequently, we show that if the primary selects estimates that the channel state of the competitor is $0$, then it will attain a maximum expected payoff of $(v-c)(1-q)\dfrac{q_s}{(1-q_s)q+q_s(1-q)}-s$ (Step ii.b.). Finally, we show that the expected payoff of the primary is at most $(v-c)(1-q)$ when it selects $Y$ (Step ii.c.). 

Step ii.a: Suppose that the primary $1$ selects $Y$ and estimates that the channel state of primary $2$ is $1$. Using (\ref{eq:1given1}) the expected payoff of primary $1$ at any price $x\in [\tilde{p},v]$ is 
\begin{align}
& (x-c)(1-\dfrac{qq_s}{qq_s+(1-q)(1-q_s)}\phi(x))-s\nonumber\\
& =(x-c)(1-\dfrac{q_s}{qq_s+(1-q)(1-q_s)}(1-\dfrac{\tilde{p}-c}{x-c}))-s\nonumber\\
& =(x-c)(1-\dfrac{q_s}{qq_s+(1-q)(1-q_s)})+(\tilde{p}-c)\dfrac{q_s}{qq_s+(1-q)(1-q_s)}-s
\end{align}
Note that $\tilde{p}-c=(v-c)(1-q)$. Since $q_s>qq_s+(1-q)(1-q_s)$ when $q_s>1/2$, thus, the above is maximized at $\tilde{p}$, hence, the maximum expected payoff that primary $1$ can attain when it estimates the channel state of its competitor is $1$  is 
\begin{align}\label{eq:nmax1}
& (v-c)(1-q)(1-\dfrac{q_s}{qq_s+(1-q)(1-q_s)})+(v-c)(1-q)\dfrac{q_s}{qq_s+(1-q)(1-q_s)}-s\nonumber\\
& =(v-c)(1-q)-s
\end{align}

Step ii.b:  Now, suppose that the primary $1$ estimates that the channel state of primary $2$ is $0$. Using (\ref{eq:0given0})  the   expected payoff of primary $1$ at any price $x\in [\tilde{p},v]$ in this case is
\begin{align}
& (x-c)(1-\dfrac{(1-q_s)q}{(1-q_s)q+q_s(1-q)}\phi(x))-s\nonumber\\
& =(x-c)(1-\dfrac{1-q_s}{(1-q_s)q+q_s(1-q)}(1-\dfrac{\tilde{p}-c}{x-c}))-s\nonumber\\
& =(x-c)(1-\dfrac{1-q_s}{(1-q_s)q+q_s(1-q)})+(\tilde{p}-c)\dfrac{1-q_s}{(1-q_s)q+q_s(1-q)}-s
\end{align}
The above is maximized at $x=v$. Hence, the maximum expected payoff that a primary can attain is
\begin{align}\label{eq:nmax2}
& (v-c)\dfrac{(1-q)(2q_s-1)}{(1-q_s)q+q_s(1-q)}+(v-c)(1-q)\dfrac{1-q_s}{(1-q_s)q+q_s(1-q)}-s\nonumber\\
& = (v-c)(1-q)\dfrac{q_s}{(1-q_s)q+q_s(1-q)}-s
\end{align}

Step ii.c: Note that a primary estimates that the channel state of primary $2$ is $0$ w.p. $(1-q_s)q+q_s(1-q)$ and the channel state of primary $2$ is $1$ w.p. $qq_s+(1-q)(1-q_s)$. The primary also incurs the cost of $s$ when it selects $Y$. Hence, from (\ref{eq:nmax1}) and (\ref{eq:nmax2}) the maximum expected payoff that primary $1$ can attain by selecting $Y$ is
\begin{align}
& (v-c)(1-q)q_s+(v-c)(1-q)(qq_s+(1-q)(1-q_s))-s\nonumber\\
& =(v-c)(1-q)(2qq_s-q+1)-s
\end{align}
However, since $s\geq (v-c)(1-q)(2qq_s-q)$, thus the maximum expected payoff that a primary can attain by selecting $Y$ is $(v-c)(1-q)$. Hence, a primary does not have any profitable unilateral deviation. 
\end{proof}
Note  that the threshold $(v-c)(1-q)(2qq_s-q)$ increases as $q_s$ increases.   Intuitively, as $q_s$ increases, the uncertainty regarding the channel state of the competitors reduces, thus, a primary tends to select $N$ for a smaller range of the values of $s$. 

The expected payoff of each primary is identical and equal to $(v-c)(1-q)$. Since both the players select $N$, thus, the expected payoff does not depend on $s$ in this case.

\subsection{Low $s$}
Now, we show that there exists a NE where each primary randomizes between $Y$ and $N$ when $s<(v-c)(1-q)(2qq_s-q)$.  Towards this end, we introduce some distribution functions parameterized by $p$. The significance of $p$ is shown later.
\begin{align}\psi_{Y,1}(x)=&
0, x<\tilde{p}_1\nonumber\\
& \alpha_{1,p}(1-\dfrac{\tilde{p}_1-c}{x-c}),  \tilde{p}_1\leq x\leq L_N\nonumber\\
& 1, x>L_N\label{eq:psi1yerror}
\end{align}
\begin{align}\label{eq:psinerror}
\psi_{N}(x)=& 
0,\quad x<L_N\nonumber\\
& \alpha_{N,p}(1-\dfrac{\tilde{p}_2-c}{x-c}-\beta_{N,p})\quad L_N\leq x\leq L_0\nonumber\\
& 1,\quad x>L_0
\end{align}
and, when $q_s<1$, then
\begin{align}\label{eq:psi0yerror}
\psi_{Y,0}(x)=&
0,\quad x<L_0\nonumber\\
& \alpha_{0,p}(1-\dfrac{\tilde{p}_3-c}{x-c}-\beta_{0,p})\quad L_0\leq x\leq v\nonumber\\
& 1, \quad x>v
\end{align}
if $q_s=1$, then 
\begin{align}\label{eq:psi_qs1}
\psi_{Y,0}(x)=H(x-v)
\end{align}
where $H(\cdot)$ is the heaviside step function or unit step function and
\begin{align}\label{eq:ln}
\tilde{p}_3-c& =(v-c)\dfrac{(1-q)q_s}{(1-q)q_s+q(1-q_s)},\quad
\tilde{p}_2-c=(v-c)(1-q)q_s\dfrac{1-(1-p)q-pqq_s}{pq(1-q_s)^2+q_s(1-q)}\nonumber\\
L_0-c& =(\tilde{p}_2-c)/(1-(1-p)q-pqq_s),\quad
L_N-c=(\tilde{p}_2-c)/(1-pqq_s)\nonumber\\
\tilde{p}_1-c& =(L_N-c)\dfrac{qq_s(1-pq_s)+(1-q)(1-q_s)}{qq_s+(1-q_s)(1-q)}.
\end{align}
and
\begin{align}
 \alpha_{1,p}& =\dfrac{qq_s+(1-q)(1-q_s)}{pqq_s^2},\quad
\alpha_{N,p}=\dfrac{1}{(1-p)q},\quad
\beta_{N,p}=pqq_s\nonumber\\
\alpha_{0,p}& =\dfrac{q(1-q_s)+q_s(1-q)}{pq(1-q_s)^2},\quad
\beta_{0,p}=\dfrac{pq(1-q_s)q_s+(1-p)q(1-q_s)}{q(1-q_s)+q_s(1-q)}.
\end{align}
It is easy to discern that all the above distribution functions are continuous when $q_s<1$. When $q_s=1$, then only $\psi_{Y,0}(\cdot)$ is discontinuous which has a jump of $1$ at $v$.  Also note that the structures of $\psi_{Y,1}(\cdot), \psi_{N}(\cdot)$ and $\psi_{Y,0}(\cdot)$ are similar (i.e. variation with $x$ is the same). However, their support sets, the scaling parameters (i.e. $\alpha_{1,p}, \alpha_{N,p}, \alpha_{0,p}$), and the constants (i.e $\beta_{0,p}$ , $\beta_{N,p}$ ) are different.

When $q_s=1$, the values of the parameters in (\ref{eq:ln}) are greatly simplified which are given by--
\begin{align}\label{eq:qs=1}
\tilde{p}_3-c& =v-c,\quad
\tilde{p}_2-c=(v-c)(1-q),\quad
L_0-c=v-c\nonumber\\
L_N-c& =\dfrac{\tilde{p}_2-c}{1-pq},\quad
\tilde{p}_1-c =(L_N-c)(1-p)
\end{align}
Thus, $L_0-c$ and $\tilde{p}_3-c$  are  the highest when $q_s=1$. 
Intuitively, when $q_s=1$, a primary knows that the channel state of its competitor is unavailable w.p. $1$ if the primary estimates that the channel state of the competitor is $0$. Thus, the primary selects $v$ w.p. $1$.  

Now, we are ready to state the main result of this section. 
\begin{theorem}\label{thm:mixederror}
Consider the following strategy profile: Each primary selects $Y$ w.p. $p$ and $N$ w.p. $1-p$ where $p$ satisfies the following equality
\begin{align}\label{eq:alternatep}
\tilde{p}_2-c=(\tilde{p}_1-c)(qq_s+(1-q)(1-q_s))+(\tilde{p}_3-c)(q(1-q_s)+q_s(1-q))-s
\end{align}
where $\tilde{p}_1, \tilde{p}_2$ and $\tilde{p}_3$ are given in (\ref{eq:ln}).
While selecting $Y$, each primary selects its price from $\psi_{Y,1}(\cdot)$ (given in (\ref{eq:psi1yerror})) if the estimated channel state of the other primary is $1$ and each primary selects its price from $\psi_{Y,0}(\cdot)$ (given in (\ref{eq:psi0yerror})) if the estimated channel state of the other primary is $0$. While selecting $N$, each primary selects its price using $\psi_{N}(\cdot)$ (given in (\ref{eq:psinerror})).

The above strategy profile is an NE when $s<(v-c)(1-q)(2qq_s-q)$. The above strategy profile is unique in the class of symmetric NE strategies. The expected payoff that each primary gets is $\tilde{p}_2-c$.
\end{theorem}
{\em Discussion}: Note that when $q_s=1$, we  know from Theorem~\ref{thm:mixedstrategy} that the strategy profile is unique one among {\em all} strategy profiles not only {\em symmetric} ones. There is no equilibrium where both the players select $Y$ w.p. $1$ even when $q_s<1$ (we have already shown the above for $q_s=1$ in Theorem~\ref{thm:yandy}). 

Now, we show that there exists a unique solution of (\ref{eq:alternatep}) in $p$ in the interval $0<p<1$ when $0<s<q(v-c)(1-q)(2q_s-1)$. 
\begin{obs}\label{obs:error}
There exists a unique solution in $p\in (0,1)$  of the equation (\ref{eq:alternatep}) when $0<s<(v-c)(1-q)(2qq_s-q)$. As $s$ decreases $p$ increases.
\end{obs}
\begin{proof}
First noe that (\ref{eq:alternatep}) can be written as
\begin{align}
\tilde{p}_2-c-(\tilde{p}_1-c)(qq_s+(1-q)(1-q_s))=(\tilde{p}_3-c)(q(1-q_s)+q_s(1-q))-s.\nonumber
mber
\end{align}
 Using (\ref{eq:ln}) we can rewrite the above as 
\begin{align}\label{eq:rep}
& (v-c)(1-q)q_s\dfrac{1-(1-p)q-pqq_s}{pq(1-q_s)^2+q_s(1-q)}\left(1-\dfrac{(qq_s(1-pq_s)+(1-q)(1-q_s))}{(1-pqq_s)}\right)\nonumber\\& =(v-c)(1-q)q_s-s
\end{align}
First, we show that the left hand of (\ref{eq:rep}) is strictly increasing in $p$ (Step i). Next, we show that when $p=0$, the left hand side of (\ref{eq:rep}) is less than the right hand side and when $p=1$, the left hand side of (\ref{eq:rep}) is greater than the right hand side (Step ii). Since the left hand side of (\ref{eq:rep}) is continuous in $p$ and strictly increasing in $p$,  there exists a unique solution $p\in (0,1)$  of (\ref{eq:rep}).  The last part easily follows as the right hand side of (\ref{eq:rep}) decreases with $s$,  the left hand side of (\ref{eq:rep}) is strictly increasing in $p$ and independent of $s$.  Now, we show steps i and ii. 

Step i: By differentiating the left hand side of (\ref{eq:rep}) we can show that
\begin{align}
1-\dfrac{qq_s(1-pq_s)+(1-q)(1-q_s)}{1-pqq_s}
\end{align}
is strictly increasing in $p$ when $q_s>1/2$. On the other hand, it is easy to discern that $(v-c)(1-q)q_s\dfrac{1-(1-p)q-pqq_s}{pq(1-q_s)^2+q_s(1-q)}$ is non-decreasing in $p$ when $q_s>1/2$. Thus, the left hand side of (\ref{eq:rep}) is strictly increasing in $p$.
 
Step ii: When $p=0$, then the value of left hand side of the equation (\ref{eq:rep}) is
\begin{align}
(v-c)(1-q)(1-qq_s-(1-q)(1-q_s))=(v-c)(1-q)(q_s+q-2qq_s)
\end{align}
Now, $(v-c)(1-q)q_s-(v-c)(1-q)(q_s+q-2qq_s)=(v-c)(1-q)(2qq_s-q)$. Since $s<(1-q)(v-c)(2qq_s-q)$, thus, the left hand side of (\ref{eq:rep}) is less than the right hand side. 

Now when $p=1$, then the left hand side of (\ref{eq:rep}) is 
\begin{align}
& (v-c)(1-q)q_s[\dfrac{1-qq_s}{q(1-q_s)^2+(1-q)q_s}-\dfrac{qq_s(1-q_s)+(1-q)(1-q_s)}{q(1-q_s)^2+(1-q)q_s}]\nonumber\\
& =(v-c)(1-q)q_s[\dfrac{qq_s^2+q+q_s-3qq_s}{q(1-q_s)^2+(1-q)q_s}] \nonumber\\
&=(v-c)(1-q)q_s[\dfrac{q(1-q_s)^2+q_s(1-q)}{q(1-q_s)^2+(1-q)q_s}]=(v-c)(1-q)q_s
\end{align}
Since $s>0$, thus, the left hand side of (\ref{eq:rep}) is greater than the right hand side.

Since the left hand side of (\ref{eq:rep}) is continuous function of $p$, thus, by intermediate value theorem there exists a solution in the interval $(0,1)$. 
\end{proof}
Next, we show that the expected payoff of a primary is a strictly greater than $(v-c)(1-q)$ when $q_s<1$ and the payoff increases with the decrease in $s$.
\begin{lem}\label{lm:result}
When $q_s<1$,  $\tilde{p}_2-c$ increases with the decrease in $s$ and $\tilde{p}_2-c$ is strictly greater than $(v-c)(1-q)$ when $s<(v-c)(1-q)(2qq_s-q)$.
\end{lem}
\begin{proof}
 Now, it is easy to discern that $\tilde{p}_2-c$ is strictly increasing in $p$ when $q_s<1$. Now, $p$ increases with the decrease in $s$ (by Observation~\ref{obs:error})  when $s<(v-c)(1-q)(2qq_s-q)$. Hence, $\tilde{p}_2-c$ is a strictly decreasing function in $s$ when $q_s<1$.

When $p=0$, then $\tilde{p}_2-c=(v-c)(1-q)$ (by (\ref{eq:ln})). Since $\tilde{p}_2-c$ is strictly increasing in $p$ when $q_s<1$ and $s<q(v-c)(1-q)(2q_s-1)$, hence $\tilde{p}_2-c>(v-c)(1-q)$.
\end{proof}
Note from Theorem~\ref{thm:mixederror} that the expected payoff attained by a primary under the NE is 
$\tilde{p}_2-c$. From (\ref{eq:qs=1}) note that $\tilde{p}_2-c=(v-c)(1-q)$ when $q_s=1$. Thus, the above lemma entails that the expected payoff of each primary increases when there is an error in estimating the channel state of the competitor. This contradicts the conventional wisdom that the payoff should increase with the decrease in error in the estimation. In Section~\ref{sec:mainresults_estimationerror} we have already explained the apparent reason behind this result. 

The above lemma entails that the expected payoff increases as $s$ decreases when $q_s<1$. Note that when $q_s=1$, the expected payoffs of primaries are independent of $s$ which we have already seen in the basic model (Section~\ref{sec:basic_model}). 

Note that when $q_s=1$, then $\psi_{Y,0}(\cdot)$ has a jump of size $1$ i.e. a primary will select $v$ w.p. $1$ as the primary will always be able to sell its channel. However, when $q_s<1$, then a primary selects its price using a continuous distribution from the interval $[L_0,v]$ where $L_0<v$. We have already explained the reason behind this in Section~\ref{sec:mainresults_estimationerror}. 

\subsubsection{Proof of Theorem~\ref{thm:mixederror}}
Before digging into the details of proof, we state few more results which we use throughout. Note from (\ref{eq:ln}) that 
\begin{align}\label{eq:l0}
(\tilde{p}_3-c)(q_s(1-q)+q(1-q_s))/[q_s(1-q)+pq(1-q_s)^2]
=L_0-c
\end{align}
Since $q_s>1/2$, thus, by cross multiplication, it is easy to see that
\begin{align}\label{eq:compareqs}
\dfrac{q_s((1-q_s)q+(1-q)q_s)}{(1-q_s)(qq_s+(1-q_s)(1-q))}>1
\end{align}

We show that primary $1$ can not gain higher profit by deviating from the strategy prescribed in Theorem~\ref{thm:mixederror} when primary $2$ follows the strategy prescribed in Theorem~\ref{thm:mixederror}. This will complete the proof. Toward this end, we first show that when primary $1$ selects $Y$ and it estimates that the channel state of its competitor is $1$, then it will attain a maximum expected payoff of $\tilde{p}_1-c-s$. The maximum expected payoff is attained  when it follows the strategy $\psi_{Y,1}(\cdot)$ (Step i). Subsequently, we show that under $Y$, when the primary estimates the channel state as $0$, then the maximum expected payoff that primary $1$ can attain is $\tilde{p}_3-c-s$ and it is attained when the primary follows the strategy $\psi_{Y,0}(\cdot)$ (Step ii). Subsequently, we show that the maximum expected payoff that primary $1$ can attain under $Y$ is $\tilde{p}_2-c$ and it is attained when primary $1$ follows the strategy (Step iii). Subsequently, we show that when primary $1$ selects $N$, then its maximum expected payoff is $\tilde{p}_2-c$ and it is attained when the primary follows the strategy $\psi_N(\cdot)$ (Step iv). Finally, we show that the maximum expected payoff of primary $1$ is $\tilde{p}_2-c$ and it is attained if primary $1$ follows the strategy profile (Step v).

Step i: Suppose that primary $1$ selects $Y$ and estimates that the channel state of primary $2$ is $1$. We show that the maximum expected payoff attained by the primary $1$ is $\tilde{p}_1-c-s$ and this is attained only when the primary selects its price from the interval $[\tilde{p}_1,L_N]$. Toward this end, we first show any price in the interval $[\tilde{p}_1,L_N]$ will fetch an expected payoff of $\tilde{p}_1-c-s$ (Step i.a.). Subsequently, we show that if primary $1$ selects a price from the interval $[L_N,L_0]$ and $[L_0,v]$  it will fetch an expected payoff of less than $\tilde{p}_1-c-s$ in Step i.b. and Step i.c. respectively. Note that at any price less than $\tilde{p}_1$ will fetch a strictly lower payoff compared to the price at $\tilde{p}_1$ as primary $2$ does not select any price lower than or equal to $\tilde{p}_1$. Thus, this will complete step i.

Step i.a: Here, we are considering the scenario where primary $1$ estimates that the channel state of primary $2$ is $1$. Under $Y$, when the primary $1$ estimates that the channel state of primary $2$ is $1$, then the probability that the channel state of primary $2$ is $1$ is
\begin{align}\label{eq:1suff}
\dfrac{q_sq}{qq_s+(1-q)(1-q_s)}
\end{align}
Suppose that primary $1$ selects a price $x$ in the interval $[\tilde{p}_1,L_N]$. When the channel state of primary $2$ is $1$ it will select a price less than or equal to $x$ only if it selects $Y$, it estimates  the channel state of primary $1$ as $1$ and selects a price less than or equal to $x$.The primary $2$ selects $Y$ w.p. $p$. Now, when the channel of primary $1$ is available, then primary $2$ estimates the channel state of primary $1$ as $1$ w.p. $q_s$ and selects a price less than or equal to $x$ w.p. $\psi_{Y,1}(x)$. The channel state of primary $2$ is $1$ with probability given in (\ref{eq:1suff}). Hence, the expected payoff of primary $1$ when its channel is available and selects a price $x$ in the interval $[\tilde{p}_1,L_N]$ is
\begin{align}
(x-c)(1-\dfrac{pq_sqq_s}{qq_s+(1-q)(1-q_s)}\psi_{Y,1}(x))-s=\tilde{p}_1-c-s\quad \text{from } (\ref{eq:psi1yerror}).
\end{align}

Step i.b.: Now, suppose that primary $1$ selects a price $x\in [L_N,L_0]$. Note that if the channel of primary $2$ is available, it will select a price less than or equal to $x$ if if one of the following occurs--i) primary $2$ selects $Y$ and estimates the channel state of primary $1$ is $1$, ii) primary $2$ selects $N$ and selects a price less than or equal to $x$. (i) occurs with probability $pq_s$ and (ii) occurs with probability $(1-p)\psi_N(x)$. Since primary $1$ estimates that the channel state of primary $2$ is $1$, thus, the probability that the true state of the channel of primary $2$ is indeed $1$ is given by (\ref{eq:1suff}).  Thus, the probability that the primary $2$ will select a price less than or equal to $x$ is
\begin{align}
\dfrac{pq_s^2q}{qq_s+(1-q)(1-q_s)}+\dfrac{(1-p)qq_s}{qq_s+(1-q)(1-q_s)}\psi_{N}(x)\nonumber
\end{align}
Thus, at $x$, the expected payoff of primary $1$ is--
\begin{align}\label{eq:higherln}
& (x-c)(1-\dfrac{pq_s^2q}{qq_s+(1-q)(1-q_s)}-\dfrac{(1-p)qq_s}{qq_s+(1-q)(1-q_s)}\psi_{N}(x))-s\nonumber\\
& =(x-c)(1-\dfrac{pq_s^2q}{qq_s+(1-q)(1-q_s)}-\dfrac{q_s}{qq_s+(1-q)(1-q_s)}(1-\dfrac{\tilde{p}_2-c}{x-c}-pqq_s))-s\quad \text{from } (\ref{eq:psinerror})\nonumber\\
& =(x-c)(1-\dfrac{q_s}{qq_s+(1-q)(1-q_s)})+(\tilde{p}_2-c)\dfrac{q_s}{qq_s+(1-q)(1-q_s)}-s
\end{align}
Since $q_s>1/2$, the co-efficient is negative. Thus, the above  is maximized at $x=L_N$. Using (\ref{eq:ln}) , the above expression is thus, upper bounded by
\begin{align}\label{eq:highln}
& (L_N-c)(1-\dfrac{q_s}{qq_s+(1-q)(1-q_s)})+(L_N-c)\dfrac{q_s(1-pqq_s)}{qq_s+(1-q)(1-q_s)}-s\nonumber\\
& =(L_N-c)\dfrac{qq_s(1-pq_s)+(1-q)(1-q_s)}{qq_s+(1-q)(1-q_s)}-s=\tilde{p}_1-c-s\quad \text{from }(\ref{eq:ln})
\end{align}

Step i.c:  From steps i.a. and i.b. we have already shown that when the maximum expected payoff of primary $1$ is $\tilde{p}_1-c-s$ at a price in the interval $[\tilde{p}_1,L_0]$. When $q_s=1$, $L_0=v$ (from (\ref{eq:qs=1})). Thus, it shows that when $q_s=1$, the maximum expected payoff of primary $1$ is indeed $\tilde{p}_1-c-s.$

Now, we consider the case where $q_s<1$ and primary $1$ selects price $x\in [L_0,v]$.  When the channel of primary $2$ is available, then  primary $2$ selects a price less than or equal to $x$ if one of the following occurs--i) it selects $Y$ and estimates that the channel state of primary $1$ is $1$, ii) primary $2$ selects $N$, iii) primary $2$ selects $Y$, estimates that the channel state of primary $1$ is $0$ and selects a price less than or equal to $x$. (i) occurs with probability $pq_s$ since the channel of primary $1$ is available. (ii) occurs with probability $1-p$.  (iii) occurs with probability $p(1-q_s)\psi_{Y,0}(x)$ (since the channel of primary $1$ is available). On the other hand the probability that the channel of primary $2$ is available is given by (\ref{eq:1suff}).  Hence, the probability that primary $2$ selects a price less than or equal to $x$ is 
\begin{align}
\dfrac{pqq_s^2+(1-p)qq_s+pqq_s(1-q_s)\psi_{Y,0}(x)}{qq_s+(1-q)(1-q_s)}\nonumber
 \end{align} Thus,  at any price $x\in [L_0,v]$, the expected payoff of primary $1$ is
\begin{align}
& (x-c)(1-\dfrac{pqq_s^2+(1-p)qq_s+pqq_s(1-q_s)\psi_{Y,0}(x)}{qq_s+(1-q)(1-q_s)})-s\nonumber\\
& =(x-c)(1-\dfrac{pqq_s^2+(1-p)qq_s}{qq_s+(1-q)(1-q_s)})-s-\nonumber\\
& (x-c)\dfrac{q_s[(1-q)q_s+(1-q_s)q]}{(1-q_s)[qq_s+(1-q)(1-q_s)]}(1-\dfrac{\tilde{p}_3-c}{x-c}-\dfrac{pq(1-q_s)q_s+(1-p)q(1-q_s)}{(1-q)q_s+(1-q_s)q})\quad \text{from } (\ref{eq:psi0yerror})\nonumber\\
&=(x-c)(1-\dfrac{[q(1-q_s)+q_s(1-q)]q_s}{(1-q_s)[qq_s+(1-q)(1-q_s)]})+(\tilde{p}_3-c)\dfrac{[q(1-q_s)+q_s(1-q)]q_s}{(1-q_s)[qq_s+(1-q)(1-q_s)]}-s
\end{align}
By (\ref{eq:compareqs}) the co-efficient of $(x-c)$ is negative, thus, the maximum of the above expression is attained at $x=L_0$. Thus, the expected payoff at $x$ is upper bounded by expected payoff at $L_0$. From (\ref{eq:highln}) (which also gives the expected payoff at $L_0$) we have already bounded the expected payoff at $L_0$ which is $\tilde{p}_1-c-s$.  
   Hence,  the maximum expected payoff that primary $1$ can attain in this case is $\tilde{p}_1-c-s$ and it is attained at any price in the interval $[\tilde{p}_1,L_N]$.
   
   Step ii: Suppose that primary $1$ estimates that the channel state of primary $2$ is $0$. When $q_s=1$, then the channel of primary $2$ is unavailable with probability $1$. Hence, primary $1$ will attain the highest possible payoff at $v$ and the payoff is $(v-c)-s=\tilde{p}_3-c-s$ (by (\ref{eq:qs=1})). Thus, we consider the case when $q_s<1$. We show that the maximum expected payoff attained by primary $1$ is $\tilde{p}_3-c-s$ and it is attained at any price in the interval $[L_0,v]$. Towards this end, we first show that any price from the interval $[L_0,v]$ will fetch an expected payoff of $\tilde{p}_3-c-s$ (Step ii.a.). Subsequently, we show that any price in the interval $[L_N,L_0]$ and $[\tilde{p}_1,L_N]$ will fetch an expected payoff of at most $\tilde{p}_3-c-s$ (Steps ii.b. and ii.c. resp.). Any price less than $\tilde{p}_1$ will fetch a payoff which is strictly less than the payoff at $\tilde{p}_1$, thus, this will complete Step ii. 
   
   Step ii.a: When primary $1$ estimates that the channel state of primary $2$ is $0$, then the probability that the channel state of primary $2$ is $1$ is
   \begin{align}\label{eq:prob_1giv0}
   \dfrac{q(1-q_s)}{q(1-q_s)+q_s(1-q)}
   \end{align}
   Suppose that primary $1$ selects a price in the interval $x\in [L_0,v]$. If the channel of primary $2$ is available, then, the primary $2$ will select a price less than or equal to $x$ if one of the following occurs--i) it selects $Y$ and estimates that the channel state of primary $1$ is $1$, ii) primary $2$ selects $N$, iii) primary $2$ selects $Y$, estimates that the channel state of primary $1$ is $0$ and selects a price less than or equal to $x$. (i) occurs with probability $pq_s$ since the channel of primary $1$ is available. (ii) occurs with probability $1-p$.  (iii) occurs with probability $p(1-q_s)\psi_{Y,0}(x)$ (since the channel of primary $1$ is available). On the other hand the probability that the channel of primary $2$ is available is given by (\ref{eq:prob_1giv0}) as primary $1$ estimates that the channel state of primary $2$ is $0$.  Hence, the probability that primary $2$ selects a price less than or equal to $x$ is 
   \begin{align}
   \dfrac{pq(1-q_s)q_s+(1-p)q(1-q_s)+pq(1-q_s)^2\psi_{Y,0}(x)}{q(1-q_s)+q_s(1-q)}.\nonumber
   \end{align}Hence, the expected payoff of primary $1$ at $x$  is
   \begin{align}
   & (x-c)(1-\dfrac{pq(1-q_s)q_s+(1-p)q(1-q_s)}{q(1-q_s)+q_s(1-q)}-\dfrac{pq(1-q_s)^2}{q(1-q_s)+q_s(1-q)}\psi_{Y,0}(x))-s\nonumber\\
   & =\tilde{p}_3-c-s\quad \text{from } (\ref{eq:psi0yerror})
   \end{align}
   Step ii.b.: Now, suppose primary $1$ selects a price $x$ in the interval $ [L_N,L_0]$. When the channel of primary $2$ is available, then primary $2$ selects a price less than or equal to $x$ if one of the following occurs--i) primary $2$ selects $Y$ and estimates the channel state of primary $1$ is $1$, ii) primary $2$ selects $N$ and selects a price less than or equal to $x$. (i) occurs with probability $pq_s$ and (ii) occurs with probability $(1-p)\psi_N(x)$. Given that the primary $1$ estimates that the channel state of primary $2$ is $0$, the probability that channel of primary $1$ is available is given by (\ref{eq:prob_1giv0}). Thus, the probability that primary $2$ selects a price less than or equal to $x$ is given by 
   \begin{align}
   \dfrac{(pq_s+(1-p)\psi_N(x))q(1-q_s)}{q(1-q_s)+q_s(1-q)}
   \end{align}
   Hence, the expected payoff of primary $1$ at $x$ is
   \begin{align}\label{eq:0estimate}
   & (x-c)(1-\dfrac{pq(1-q_s)q_s}{q(1-q_s)+q_s(1-q)}-\dfrac{(1-p)q(1-q_s)}{q(1-q_s)+q_s(1-q)}\psi_{N}(x))-s\nonumber\\
   & =(x-c)(1-\dfrac{pqq_s(1-q_s)}{q(1-q_s)+q_s(1-q)})-  (x-c)\dfrac{1-q_s}{q(1-q_s)+q_s(1-q)}(1-\dfrac{\tilde{p}_2-c}{x-c}-pqq_s)-s\quad \text{from } (\ref{eq:psinerror})\nonumber\\
   & =(x-c)(1-\dfrac{1-q_s}{q(1-q_s)+q_s(1-q)})+(\tilde{p}_2-c)\dfrac{1-q_s}{q(1-q_s)+q_s(1-q)}-s
   \end{align}
   By (\ref{eq:compareqs}) the above is maximized at $x=L_0$. Hence, the maximum possible expected payoff is
   \begin{align}\label{eq:bound0}
   & (L_0-c)(1-\dfrac{1-q_s}{q(1-q_s)+q_s(1-q)})+(L_0-c)\dfrac{(1-q_s)(1-(1-p)q-pqq_s)}{q(1-q_s)+q_s(1-q)}-s\nonumber\\
  & =(L_0-c)(1-\dfrac{[(1-p)q+pqq_s](1-q_s)}{q(1-q_s)+q_s(1-q)})-s\nonumber\\
  & =(L_0-c)\dfrac{pq(1-q_s)^2+q_s(1-q)}{q(1-q_s)+q_s(1-q)}-s =\tilde{p}_3-c-s\quad \text{from } (\ref{eq:ln})
   \end{align}
   Step ii.c.: Now, suppose that primary $1$ selects a price  $x$ in the interval $ [\tilde{p}_1,L_N]$. Now, if the channel of primary $2$ is available, then it selects a price less than or equal to $x$ if it selects $Y$, estimates that the channel state of primary $1$ is $1$ and selects a price less than or equal to $x$. The above occurs with probability $pq_s\psi_{Y,1}(x)$. The probability that the channel state of primary $2$ is $1$ given that the primary $1$ estimates that the channel state of primary $2$ is $0$ is given by (\ref{eq:prob_1giv0}). Hence at $x$ the expected payoff of primary $1$ is
   \begin{align}
   & (x-c)(1-\dfrac{pq(1-q_s)q_s\psi_{Y,1}(x)}{q(1-q_s)+q_s(1-q)})-s\nonumber\\
   & =(x-c)(1-\dfrac{(1-q_s)[qq_s+(1-q)(1-q_s)]}{q_s[q(1-q_s)+q_s(1-q)]}(1-\dfrac{\tilde{p}_1-c}{x-c}))-s\quad \text{from } (\ref{eq:psi1yerror})\nonumber\\
   & =(x-c)(1-\dfrac{(1-q_s)(qq_s+(1-q)(1-q_s))}{q_s((1-q_s)q+q_s(1-q))})+ (\tilde{p}_1-c)\dfrac{(1-q_s)(qq_s+(1-q)(1-q_s))}{q_s(q(1-q_s)+q_s(1-q))}-s\nonumber
   \end{align}
   By (\ref{eq:compareqs}) the above is maximized at $x=L_N$.  Thus, the expected payoff at any $x\in [\tilde{p}_1,L_N]$ is upper bounded by the expected payoff at $L_N$. Now, from (\ref{eq:bound0}) (which also gives the expected payoff at $L_N$) we have already shown that the expected payoff at $L_N$ is upper bounded by $\tilde{p}_3-c-s$. 
 Hence, the maximum expected payoff attained by primary $1$ in this case is $\tilde{p}_3-c-s$ and it is attained at price in the interval $[L_0,v]$. 
  
  Step iii: In Step (i), we have shown that under $Y$ if primary $1$ estimates that the channel state of primary $2$ is $1$, then the maximum expected payoff is $\tilde{p}_1-c-s$. In Step (ii), we have shown that if primary $1$ estimates that the channel state of primary $2$ is $0$, then the maximum expected payoff is $\tilde{p}_3-c-s$. Primary $1$ estimates that the channel state of primary $2$ is $1$ w.p. $qq_s+(1-q)(1-q_s)$ and the channel state of primary $2$ is $0$ w.p. $(1-q)q_s+q(1-q_s)$. Thus, under $Y$, the maximum expected payoff that primary $1$ can attain is
  \begin{align}
  (\tilde{p}_1-c-s)(qq_s+(1-q)(1-q_s))+(\tilde{p}_3-c-s)((1-q)q_s+q(1-q_s))=\tilde{p}_2-c\quad \text{from } (\ref{eq:alternatep}).
  \end{align}
  We have already shown that the payoff is achieved when primary $1$ follows the strategy prescribed in the theorem.
  
  Step iv: Now, we show that under $N$, the maximum expected payoff that primary $1$ can attain is $\tilde{p}_2-c$ and it is attained at every price in the interval $[L_N,L_0]$. Toward this end, we first show that if primary $1$ selects a price from the interval $[L_N,L_0]$ it will fetch an expected payoff of $\tilde{p}_2-c$ (Step iv.a). Subsequently, we show that if the primary selects any price from the interval $[L_0,v]$ or $[\tilde{p}_1,L_N]$ it can only get an expected payoff of at most $\tilde{p}_2-c$ (Step iv.b. and Step iv.c. resp.).  
  
  Step iv.a:  Suppose that primary $1$ selects a price $x\in [L_N,L_0]$. When the channel of primary $2$ is available, then primary $2$ selects a price less than or equal to $x$ if one of the following occurs--i) primary $2$ selects $Y$ and estimates the channel state of primary $1$ is $1$, ii) primary $2$ selects $N$ and selects a price less than or equal to $x$. (i) occurs with probability $pq_s$ and (ii) occurs with probability $(1-p)\psi_N(x)$. Now, when primary $1$ selects $N$ it only knows that the channel of primary $2$ is available w.p. $q$. Thus, the expected payoff  of primary $1$ at $x$ is
  \begin{align}
  & (x-c)(1-pqq_s-(1-p)q\psi_{N}(x))=\tilde{p}_2-c\quad \text{from } (\ref{eq:psinerror}).
  \end{align}
  
  Step iv.b: Note that when $q_s=1$, then $L_0=v$. Thus, we consider the case when $q_s<1$. Suppose primary $1$ selects a  price $x$ in the interval $[L_0,v]$. If the channel of primary $2$ is available, then, the primary $2$ will select a price less than or equal to $x$ if one of the following occurs--i) it selects $Y$ and estimates that the channel state of primary $1$ is $1$, ii) primary $2$ selects $N$, iii) primary $2$ selects $Y$, estimates that the channel state of primary $1$ is $0$ and selects a price less than or equal to $x$. (i) occurs with probability $pq_s$ since the channel of primary $1$ is available. (ii) occurs with probability $1-p$.  (iii) occurs with probability $p(1-q_s)\psi_{Y,0}x)$ (since the channel of primary $1$ is available). When primary $1$ selects $N$, it only knows that the channel of primary $2$ is available w.p $q$. Thus, the expected payoff of primary $1$ at $x$ is
  \begin{align}
 &  (x-c)(1-pqq_s-(1-p)q-pq(1-q_s)\psi_{Y,0}(x))\nonumber\\
&=(x-c)(1-pqq_s-(1-p)q)-\nonumber\\
 & (x-c)\dfrac{(1-q)q_s+q(1-q_s)}{1-q_s}(1-\dfrac{\tilde{p}_3-c}{x-c}-\dfrac{pqq_s(1-q_s)+(1-p)q(1-q_s)}{(1-q)q_s+q(1-q_s)})\quad \text{from } (\ref{eq:psi0yerror})\nonumber\\
  & =(x-c)(1-\dfrac{(1-q)q_s+q(1-q_s)}{1-q_s})+(\tilde{p}_3-c)\dfrac{(1-q)q_s+q(1-q_s)}{1-q_s}\nonumber
  \end{align}
  Since $q_s>1/2$, the above is maximized at $x=L_0$. Thus, using (\ref{eq:l0}), the above expression is upper bounded by
  \begin{align}
 & (L_0-c)(1-\dfrac{(1-q)q_s+q(1-q_s)}{1-q_s})+ (L_0-c)\dfrac{pq(1-q_s)^2+q_s(1-q)}{1-q_s}\nonumber\\
 &  =(L_0-c)(1-\dfrac{q(1-q_s)(1-p(1-q_s))}{1-q_s}) =(L_0-c)(1-(1-p)q-pqq_s)=\tilde{p}_2-c.
  \end{align}
  where the last equality follows from (\ref{eq:ln}).  
  
  Step iv.c: Now, suppose that primary $1$ selects a price  $x$ in the interval $ [\tilde{p}_1,L_N]$. Now, if the channel of primary $2$ is available, then it selects a price less than or equal to $x$ if it selects $Y$, estimates that the channel state of primary $1$ is $1$ and selects a price less than or equal to $x$. The above occurs with probability $pq_s\psi_{Y,1}(x)$. The channel of primary $2$ is available w.p. $q$. Thus, at any price $x$ in the interval $[\tilde{p}_1,L_N]$ the expected payoff of primary $1$ is 
  \begin{align}
  & (x-c)(1-pqq_s\psi_{Y,1}(x))\nonumber\\
  & =(x-c)(1-\dfrac{qq_s+(1-q)(1-q_s)}{q_s})+(\tilde{p}_1-c)\dfrac{qq_s+(1-q)(1-q_s)}{q_s}\quad \text{from } (\ref{eq:psi1yerror}). \nonumber
  \end{align}
  Since $q_s>1/2$, the above is maximized at $x=L_N$. Thus, using (\ref{eq:ln}), the maximum expected payoff is
  \begin{align}
  & (L_N-c)(1-\dfrac{qq_s+(1-q)(1-q_s)}{q_s})+(L_N-c)\dfrac{qq_s(1-pq_s)+(1-q)(1-q_s)}{q_s}\nonumber\\
  & =(L_N-c)(1-\dfrac{qq_spq_s}{q_s})=\tilde{p}_2-c.
  \end{align}
  Again the last equality follows from (\ref{eq:ln}).
  
  Hence, the maximum expected payoff attained by primary $1$ under $N$ is $\tilde{p}_2-c$. This is attained at any price in the interval $[L_N,L_0]$ which we have shown in Step iv.a.  
  
  Step v: Thus, either under $Y$ or under $N$, the maximum expected payoff that primary $1$ can attain is $\tilde{p}_2-c$. Hence, any randomization between $Y$ and $N$ will also yield an expected payoff of $\tilde{p}_2-c$. Primary $1$ can attain the payoff of $\tilde{p}_2-c$ following the strategy profile. Hence, primary $1$ does not have any unilateral profitable deviation. Hence, the result follows. \qed
\begin{color}{red}
\begin{example}\label{ex:estimation_error}
We illustrate the NE pricing strategy in Theorem~\ref{thm:mixederror} for an example where $v=50, c=0, s=4, q_s=0.8, q=0.5$. Since $s<(2q_s-1)q(v-c)(1-q)$,  the NE pricing strategy is given in Theorem~\ref{thm:mixederror}. Using the {\em fsolve} function of MATLAB we obtain the unique $p$ satisfying (\ref{eq:alternatep}). Given $p$,  $\psi_{Y,1}(\cdot)$, $\psi_N(\cdot)$ and $\psi_{Y,1}(\cdot)$ are readily obtained and shown in  Fig.~\ref{fig:ex_estimation}. 
\end{example}
\end{color}
\subsection{Numerical Results}
\begin{figure*}
\begin{minipage}{0.31\linewidth}
\includegraphics[width=60mm,height=40mm]{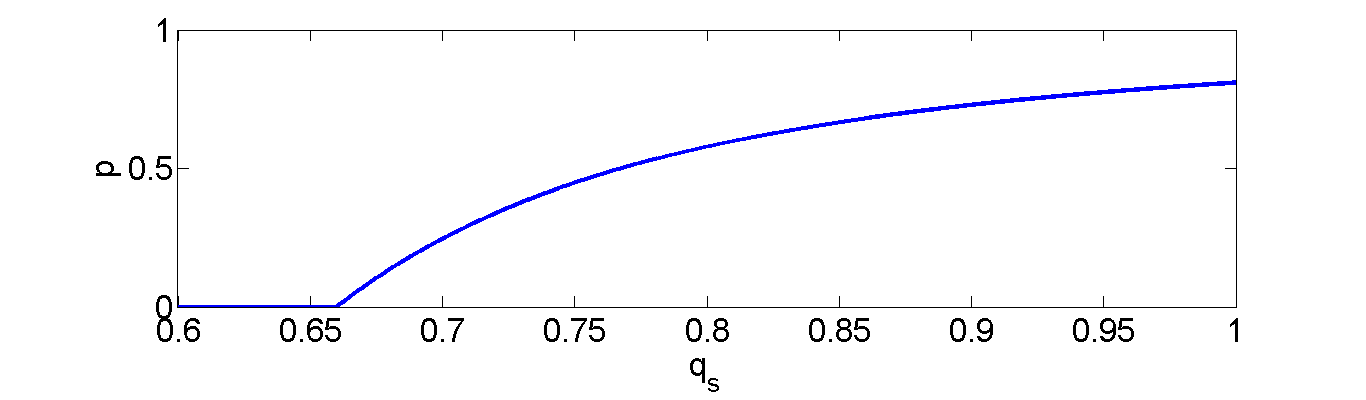}
\caption{\small Variation of $p$ with $q_s$ for an example setting: $v=50, c=0, s=4, q=0.5$. $p$ is $0$ for $q_s\leq 0.67$, as $s$ is above the threshold $q(v-c)(1-q)(2q_s-1)$ for this region.}
\label{fig:pvsqs}
\end{minipage}\hfill
\begin{minipage}{0.31\linewidth}
\includegraphics[width=60mm,height=40mm]{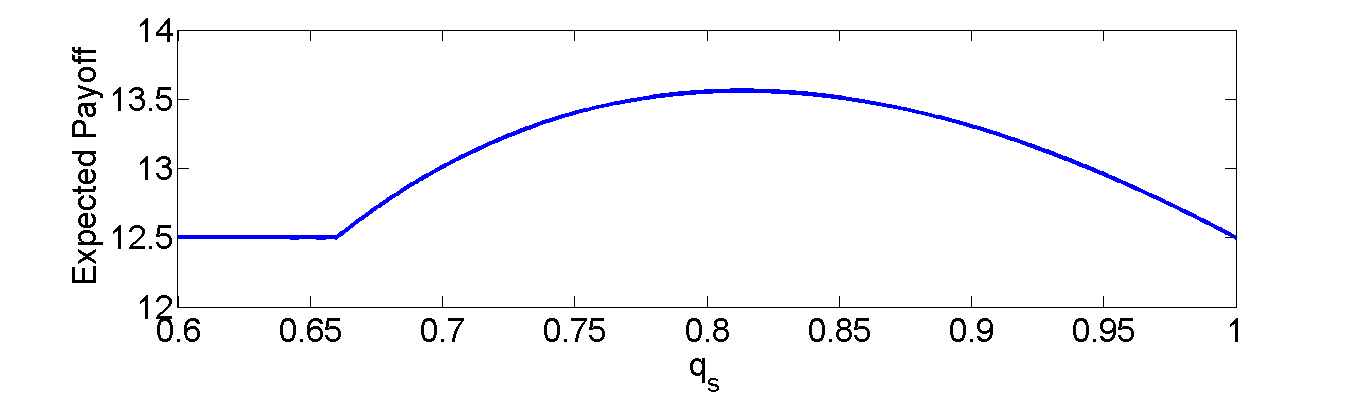}
\caption{\small Variation of the expected payoff of a primary with $q_s$ in the same example setting considered in Fig.~\ref{fig:pvsqs}.}
\label{fig:payoffvsqs}
\end{minipage}\hfill
\begin{minipage}{0.31\linewidth}
\includegraphics[width=60mm,height=40mm]{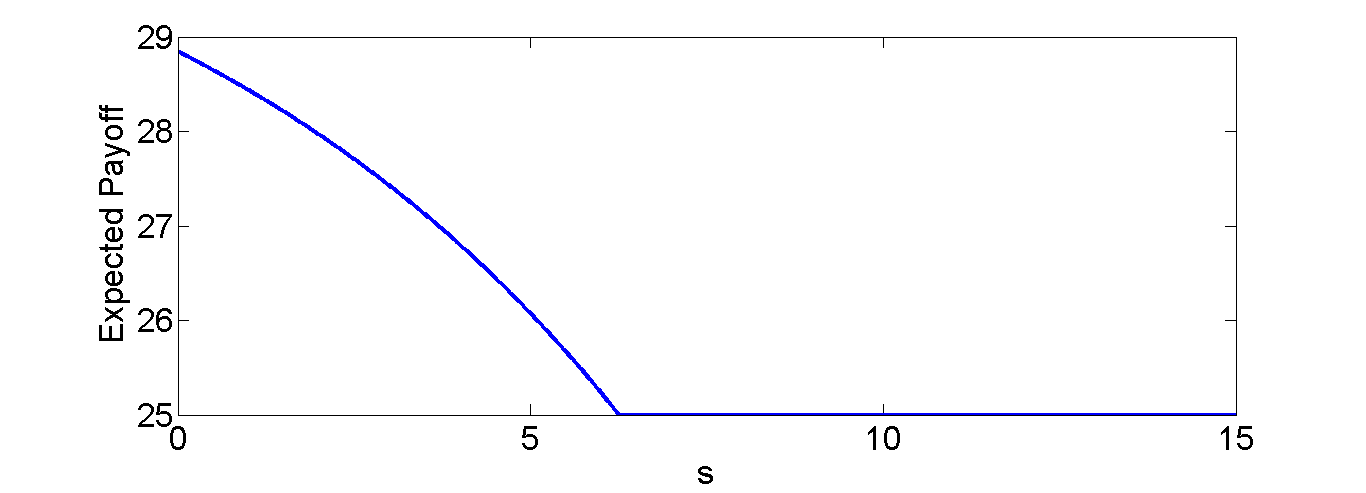}
\caption{\small Variation of the expected payoff of a primary with $s$ in an example setting: $v=50, q=0.5, c=0, q_s=3/4$.}
\label{fig:payoffvsc1}
\end{minipage}
\vspace{-0.3cm}
\end{figure*}
Fig.~\ref{fig:pvsqs} shows that the probability $p$ with which a primary selects $Y$ increases as $q_s$ increases. Intuitively, when $q_s$ increases, the uncertainty of the channel state of the competitor decreases  when a primary selects $Y$, thus, the primary selects $Y$ with a higher probability.  Additionally, Fig.~\ref{fig:pvsqs} shows that the increment of $p$ is sub-linear with $q_s$.

Fig.~\ref{fig:payoffvsqs} shows the variation of the expected payoff of a primary with $q_s$. When $0.5<q_s\leq 0.67$, a primary selects $N$ w.p. $1$, hence, the expected payoff is $(v-c)(1-q)$ for $q_s\leq 0.67$. The expected payoff increases as $q_s$ increases when $0.67<q_s\leq 0.83$. After that the payoff decreases and ultimately the expected payoff again becomes equal to $(v-c)(1-q)$ when $q_s=1$. Thus, the payoff of a primary is higher when there is an error in estimation of the channel state compared to the setting where there is no error in estimation which negates the conventional wisdom that the payoff should increases with the decreases in the error in the estimation. We have already provided the potential reasons behind this behavior in the section~\ref{sec:mainresults_estimationerror}.

Fig.~\ref{fig:payoffvsc1} shows the variation of the expected payoff of a primary with $s$.  Note from Lemma~\ref{lm:result} that the expected payoff of a primary increases as $s$ decreases when a primary selects $Y$ with a positive probability. Fig.~\ref{fig:payoffvsc1} verifies the above result. Specifically, as $s$ increases, the expected payoff decreases when $s<6.5$. Additionally, the expected payoff decreases sub-linearly. When $s\geq 6.5$, a primary only selects $N$ and attains an expected payoff of $(v-c)(1-q)$, thus, the payoff becomes independent of $s$ in this regime.
\begin{figure*}
\begin{minipage}{0.28\linewidth}
\includegraphics[width=0.99\textwidth]{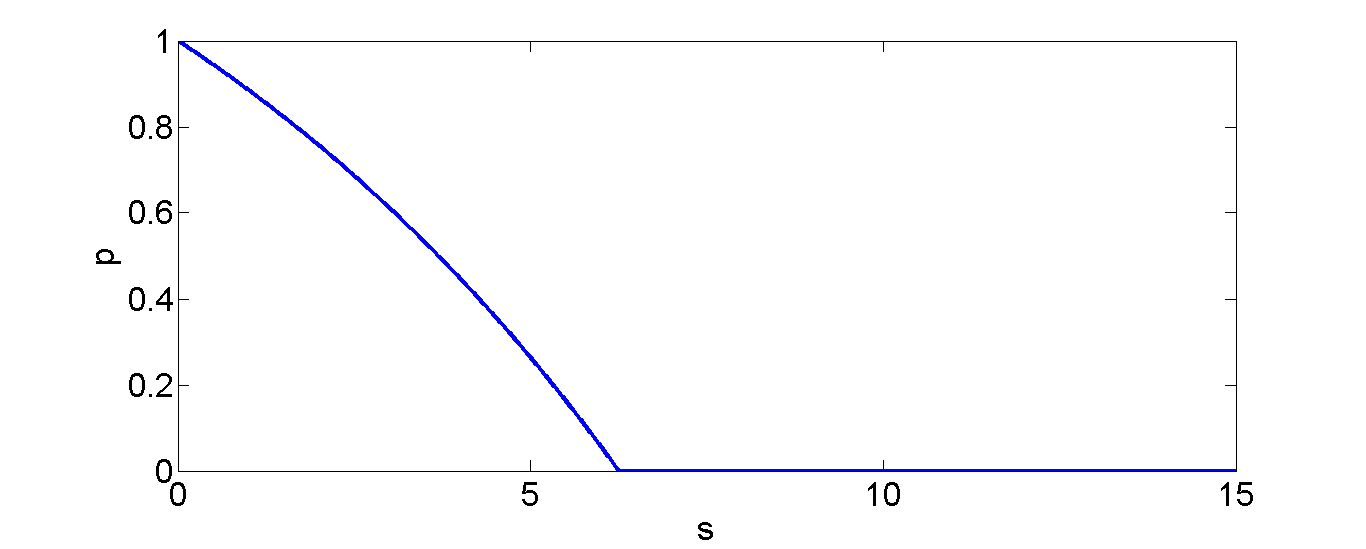}
\caption{\small Variation of $p$ with $s$ in the same example setting as considered in Fig.~\ref{fig:payoffvsc1}.}
\label{fig:pvsc1}
\end{minipage}\hfill
\begin{minipage}{0.28\linewidth}
\includegraphics[width=0.99\textwidth]{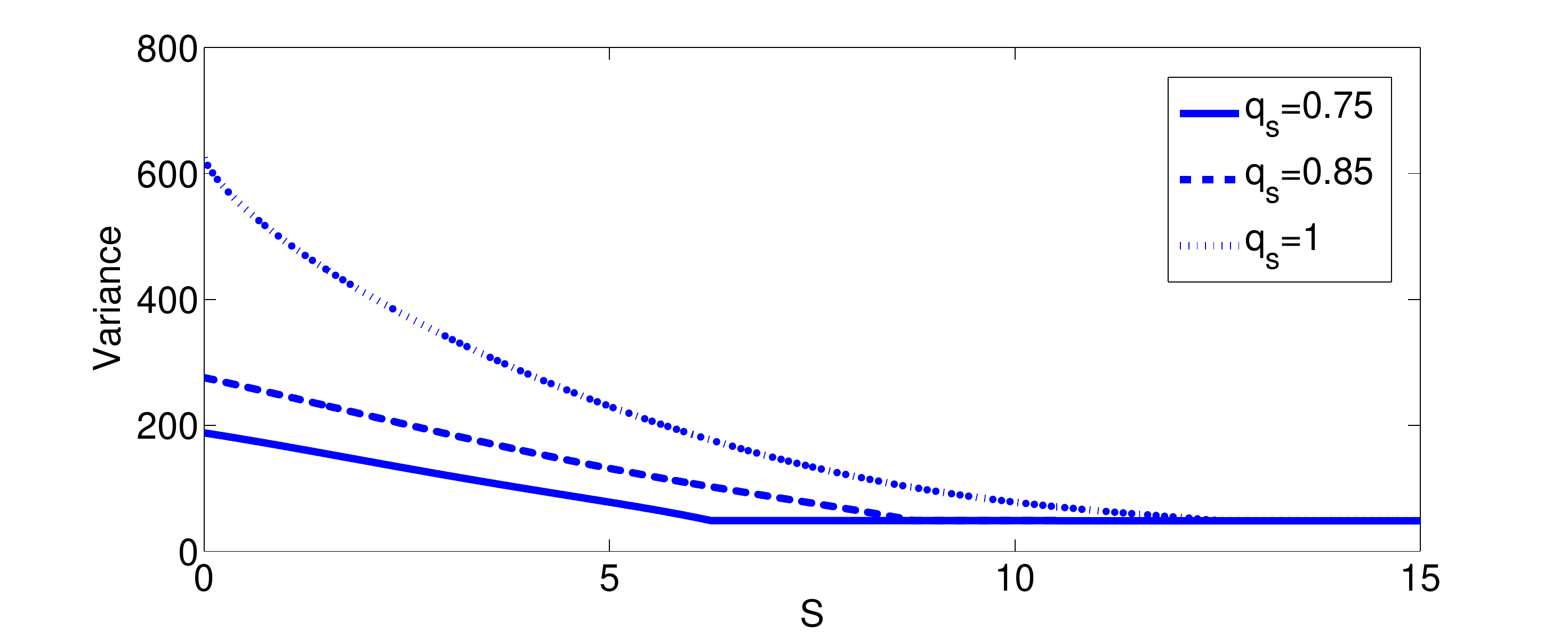}
\caption{\small Variation of the  price variance  for an example setting:$v=51, c=1, q=0.5$}
\label{fig:variancec1}
\end{minipage}\hfill
\begin{minipage}{0.4\linewidth}
\includegraphics[width=0.99\textwidth]{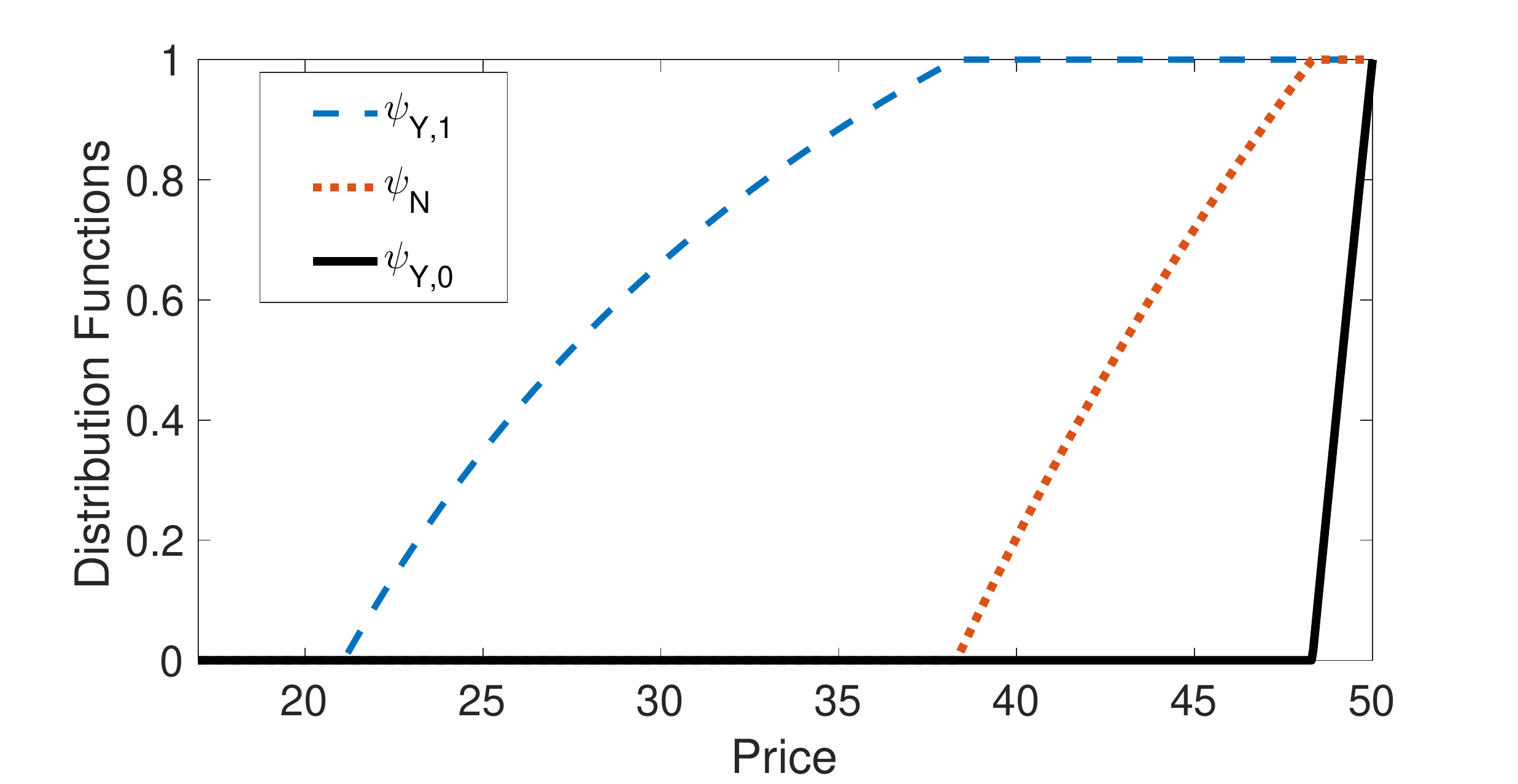}
\caption{\small The price distributions $\psi_{Y,0}(\cdot), \psi_{Y,1}(\cdot)$ and $\psi_N(\cdot)$ in the example sceanrio described in Example~\ref{ex:estimation_error}. Note that they are non-overlapping; the supports of $\psi_{Y,1,}, \psi_N$ and $\psi_{Y,0}$ are $[21.08, 38.3], [38.3, 48.3]$ and $[48.3,50]$ respv. }
\label{fig:ex_estimation}
\end{minipage}
\vspace{-0.5cm}
\end{figure*} 

Fig.~\ref{fig:pvsc1} shows that $p$, the probability with which a primary selects $Y$ increases as $s$ decreases. When $s\geq q(v-c)(1-q)(2q_s-1)=6.25$, then the primary selects $N$ w.p. $1$ i.e. $p=0$. Additionally, Fig.~\ref{fig:pvsc1} shows that $p$ decreases sub-linearly as $s$ increases.

Fig.~\ref{fig:variancec1} shows the variation of the variance of the price selected by a primary with $s$ and $q_s$. Note that the variance decreases as $s$ decreases. Thus, when a primary selects $N$ with a higher probability, the price volatility  is lower. When $s\geq q(v-c)(1-q)(2q_s-1)$, each primary selects $N$ w.p. $1$, thus, the variance becomes independent of $s$.  This is because $\phi(\cdot)$, the price selection strategy from which a primary selects its price when $s\geq q(v-c)(1-q)(2q_s-1)$, is independent of $s$. Note that Fig.~\ref{fig:variancec1} also shows that the variance also decreases as $q_s$ decreases. Intuitively, as $q_s$ decreases, a primary selects $N$ with a higher probability, thus, the variance decreases. Note that buyers in general do not like a market where the prices have higher variances. Thus, when $s$ is low or $q_s$ is high, a buyer may not like the setting.

\begin{figure*}
\begin{minipage}{0.49\linewidth}
\includegraphics[width=80mm,height=40mm]{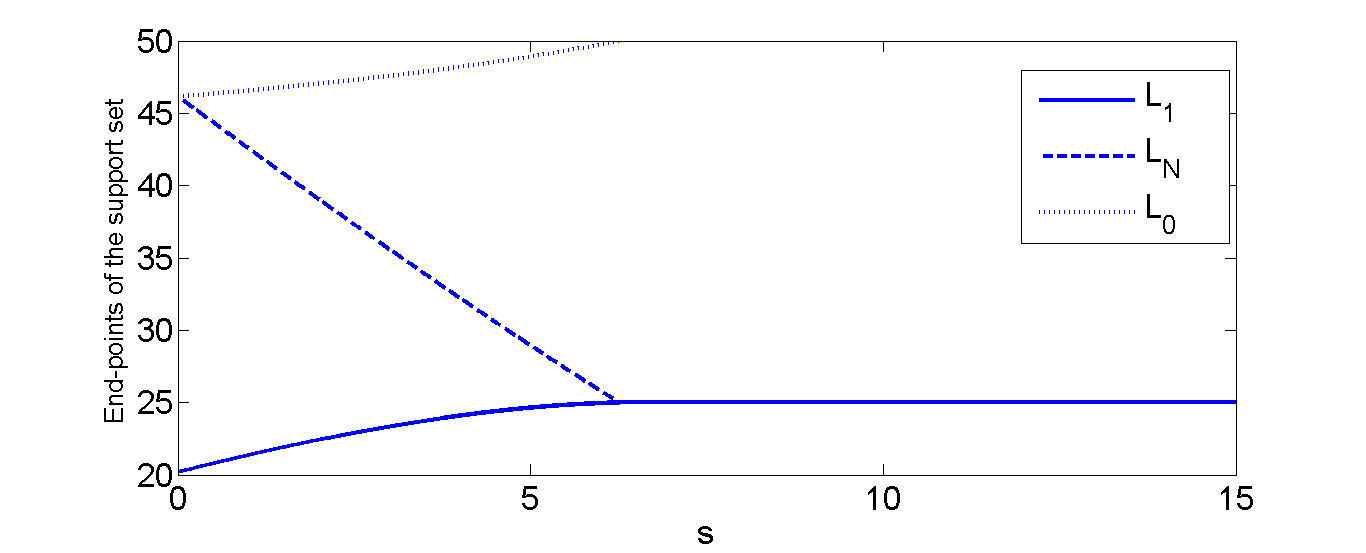}
\caption{\small Variation of upper and lower end-points of the support sets of $\psi_{Y,1}$ ($L_1$  ($=\tilde{p}_1$) and $L_N$, resp.), $\psi_{N}$ ($L_N$, and $L_0$ resp.) and $\psi_{Y,0}$ ($L_0$ and $v$) with $s$ in the same example setting considered in Fig.~\ref{fig:payoffvsc1}. }
\label{fig:endpointsvss}
\end{minipage}\hfil
\begin{minipage}{0.49\linewidth}
\includegraphics[width=80mm,height=40mm]{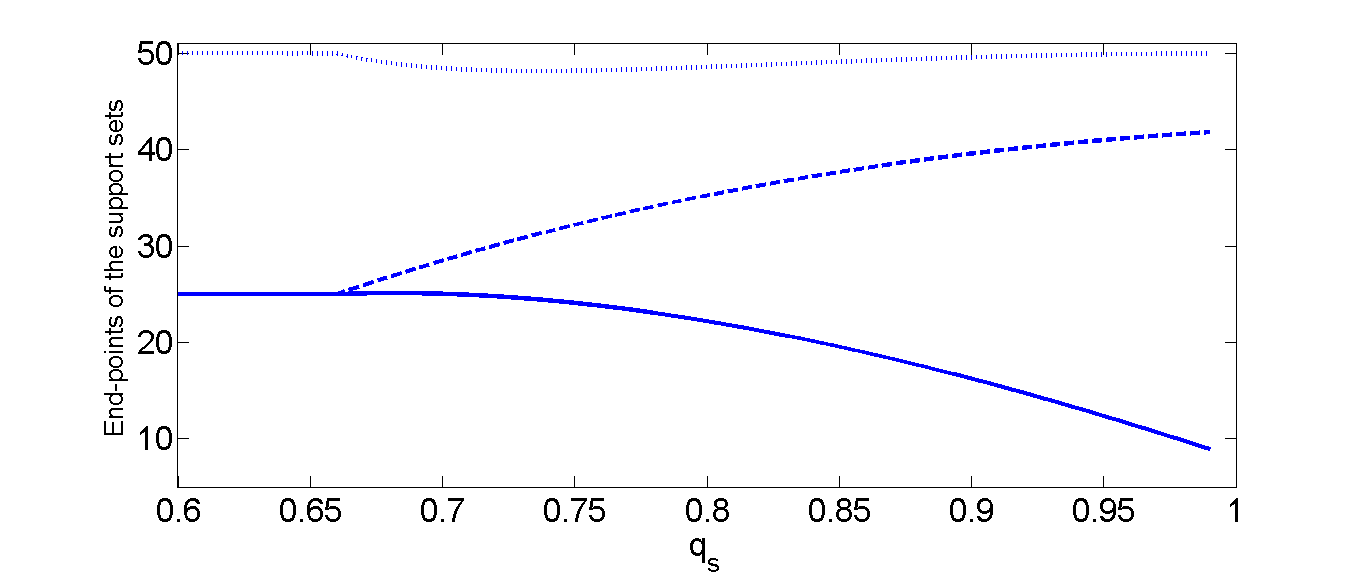}
\caption{\small Variation of upper and lower end-points of the support sets of $\psi_{Y,1}$ ($L_1$ and $L_N$, resp.), $\psi_{N}$ ($L_N$, and $L_0$ resp.) and $\psi_{Y,0}$ ($L_0$ and $v$) with $q_s$ in the same example setting considered in Fig.~\ref{fig:pvsqs}.}
\label{fig:endpointsvsqs}
\end{minipage}
\end{figure*}
Fig.~\ref{fig:endpointsvss} shows the variations of the end-points of the support sets of the price distributions. Note that when $s=0$, $L_N=L_0$ as primaries select $N$ with $0$ probability. As $s$ increases, $L_N$ and $L_0$ increase as primaries select $N$ with positive probability; primaries select prices from a larger interval when it selects $N$ as $s$ increases. Note that the lower end-point of $\psi_{Y,1}(\cdot)$, $\tilde{p}_1$ ($L_1$ in the figure) also increases as $s$ increases. Thus, the price interval from which a primary selects its price $Y$ decreases as $s$ increases. Intuitively, as $s$ increases, a primary selects $Y$ with a lower probability, thus the support also decreases. When $s\geq 6.25$, the primary only selects $N$, thus, $\tilde{p}_1=L_N$ and $L_0=v$. 

Fig.~\ref{fig:endpointsvsqs} shows the variation of the end-points of the support sets of the price distributions. When $q_s\leq 0.67$, primaries only select $N$. Thus, $\tilde{p}_1=L_N$ and $L_0=v$. When $q_s>0.67$, primaries select $Y$ with positive probabilities. $\tilde{p}_1$ decreases as $q_s$ increases. Thus, a primary selects a lower price when it selects $Y$ and estimates that the channel state of the competitor is $1$. Intuitively, as $q_s$ increases, the uncertainty reduces, thus, the competitor's channel is more likely to be available when a primary estimates that the channel state of the primary is $1$. Hence, the primary selects a lower price. Since primary selects $Y$ when $q_s>0.67$,  $L_0$ decreases initially.  However, $L_0$  increases when $q_s$ becomes very high. Note that when $q_s$ is very high, then a primary is aware that the channel of the competitor is more likely to be unavailable, hence it selects a high price. Thus, $L_0$ is close to $v$ when $q_s$ is very high. Note also that $L_N$ increases with $q_s$. Intuitively, when $q_s$ increases, a primary selects $N$ with a lower probability, thus, a primary selects its price from a shorter interval when it selects $N$.

\section{Unequal Costs}\label{sec:unequalcost}
We, now, investigate the generalization of the basic model where each different primaries incur different costs to acquire the CSI of their respective competitors depicted in Section~\ref{sec:differentcost_model}. Primary $i$ incurs the cost $s_i$ to acquire the CSI of its competitor. Without loss of generality we assume that  $s_1< s_2$.  

\subsection{Goals}
The impact of different acquisition costs on the payoff of each primary and the frequency with which each primary selects $Y$  is not apriori clear.  For example, primary $1$ which has a lower acquisition cost of CSI, can gain more compared to primary $2$ by acquiring the CSI of primary $2$ by paying a lower cost. However, primary $2$ also acquires the CSI of primary $1$ and selects a lower price when the channel of primary $1$ is available, thus, primary $1$ also selects a lower price in response, which in turn reduces the payoff of primary $1$.  The pricing decision of each primary also depends on the frequencies with which each primary selects $Y$. We resolve all these quandaries.
\subsection{Results}\label{sec:results_unequalcosts}
We summarize our main findings here--
\begin{itemize}
\item The NE strategy is of the form $[T,p_i]$ for primary $i$ with $T=q(v-c)(1-q)$. Note that $T$ is the same as the basic model, however, since different primaries have different acquisition costs, $p_i$s are different. For example, when $s_1<T\leq s_2$, then primary $1$ selects $Y$ w.p. $p_1$, but primary $2$ does not select $Y$. Even when $s_2<T$, then primary $2$ selects $p_2$ where $p_1>p_2$ as $s_1<s_2$. 
\item The difference in the acquisition costs lead to different payoffs for the primaries. In contrast to the basic model, primary $1$ attains a higher payoff compared to the expected payoff of primary $2$ when primary $1$ selects $Y$ with a positive probability (i.e. $s<T$) (Theorems~\ref{thm:mixedunequalcost}, \ref{thm:bothmixedunequalcost}). The expected payoff of primary $1$ becomes close to the payoff of the primary $2$ as the difference between $s_1$ and $s_2$ decreases. The expected payoff of primary $2$ is in fact independent of $s_2$.  The expected payoff of the primaries are the same when $s_1\geq T$ , as both of them only select $N$. 
\item Primary $i$ selects its price from the interval $[L,\tilde{p}_i]$ ($[\tilde{p}_i,v]$, resp.) , when the primary selects $Y$ ($N$, resp.) and the channel of the competitor is available. However, there are also some differences in the pricing structure compared to the basic model because of different acquisition costs. Primary $2$ selects $v$ with a positive probability when it selects $N$ when  $s_1<q(v-c)(1-q)$ and the probability decreases as the difference between $s_1$ and $s_2$ decreases (Theorems~\ref{thm:bothmixedunequalcost}).   Thus, the primary $2$ has a discontinuity at $v$ in contrast to the basic model where primaries select prices from continuous distribution. Additionally, we show that $\tilde{p}_1>\tilde{p}_2$. Thus, primary $2$ selects lower prices when it selects $Y$ and the channel of primary $1$ is available. On the other hand, when primary $2$ selects $N$, its selects higher prices with higher probabilities. 
\end{itemize}

\subsection{High $s_1, s_2$}
Our first result in this section shows that
\begin{theorem}\label{thm:nandnunequalcost}
When $s_1\geq q(v-c)(1-q)$, then in the unique NE, both the primaries select $N$ w.p. $1$ and select their prices according to $\phi(\cdot)$ (given in (\ref{eq:phi})). Both the primaries attain an expected payoff of $q(v-c)(1-q)$. 
\end{theorem}
The proof is similar to the proof of Theorem~\ref{thm:nandn}, thus, we omit it here.

Note that since $s_2>s_1$, $s_2>q(v-c)(1-q)$. Thus, the above theorem shows that the expected payoff of primaries are identical when $s_i$s are sufficiently high as both the primaries select $N$.  
\subsection{Low $s_1$, high $s_2$}
Now, we consider the setting where $s_1<q(v-c)(1-q)$, but $s_2\geq q(v-c)(1-q)$. We show that there exists a NE where primary $1$ randomizes between $Y$ and $N$, and primary $2$ selects $N$.

We first introduce some pricing distributions which we use throughout this section--
\begin{align}\label{eq:psi1ylowc1highc2}
\psi_{1,Y}(x)=& 
0,\quad \text{if } x<\tilde{p}\nonumber\\
& \dfrac{1}{qp_1}(1-\dfrac{\tilde{p}-c}{x-c}),\quad \text{if } \tilde{p}\leq x\leq \tilde{p}_1\nonumber\\
& 1, \quad \text{if } x>\tilde{p}_1
\end{align}
\begin{align}\label{eq:psi1nlowc1highc2}
\psi_{1,N}(x)=&
0, \quad \text{if } x<\tilde{p}_1\nonumber\\
& \dfrac{1}{q(1-p_1)}(1-\dfrac{\tilde{p}-c}{x-c}-qp_1) \text{if } \tilde{p}_1\leq x\leq v\nonumber\\
& 1, \quad \text{if } x>v
\end{align}
\begin{align}\label{eq:psi2lowc1highc2}
\psi_2(x)=&
0, \quad \text{if } x<\tilde{p}\nonumber\\
& (1-\dfrac{\tilde{p}-c}{x-c}),\quad \text{if } \tilde{p}\leq x< \tilde{p}_1\nonumber\\
& \dfrac{1}{q}(1-\dfrac{\tilde{p}_N-c}{x-c}),\text{if } \tilde{p}_1\leq x<v\nonumber\\
& 1, \quad \text{if } x\geq v.
\end{align}
where 
\begin{align}\label{eq:tildepunequalcost}
\tilde{p}_N& =(v-c)(1-q)+q(v-c)(1-q)-s_1+c\nonumber\\
\tilde{p}& =(v-c)(1-q)+c. \quad \tilde{p}_1-c=\dfrac{(v-c)(1-q)}{1-qp_1}.
\end{align}
and 
\begin{align}\label{eq:p1}
p_1=\dfrac{1}{q}(1-\dfrac{(v-c)(1-q)^2}{(v-c)(1-q)-s_1})
\end{align}
Note from (\ref{eq:tildepunequalcost}) and (\ref{eq:p1}) that
\begin{align}\label{eq:tildep1}
\tilde{p}_1-c& =\dfrac{(v-c)(1-q)[(v-c)(1-q)-s_1]}{(v-c)(1-q)^2}=\dfrac{(v-c)(1-q)-s_1}{1-q}
\end{align}
$\psi_2(\cdot)$ clearly has a jump at $v$ as $s_1<q(v-c)(1-q)$. From the expression of $\psi_2(\cdot)$ one may think that $\psi_2(\cdot)$ has a jump at $\tilde{p}_1$. We first rule out the above possibility.
\begin{obs}
$\psi_2(\cdot)$ does not have any jump except at $v$.
\end{obs}
\begin{proof}
First, note that since $s_1<q(v-c)(1-q)$, thus, $\psi_2(\cdot)$ has a jump at $v$.

Next, we show that $\psi_2(\cdot)$ does not have any jump at $\tilde{p}_1$. The continuity of $\psi_2(\cdot)$ at any other point can easily be observed.

Note from (\ref{eq:tildep1}) the left hand limit is 
\begin{align}
1-\dfrac{(v-c)(1-q)}{\tilde{p}_1-c}=1-\dfrac{(v-c)(1-q)^2}{(v-c)(1-q)-s_1}
\end{align}
Again from (\ref{eq:tildep1}), the right hand limit and the value of $\psi_2(\cdot)$ at $\tilde{p}_1$ is
\begin{align}
& \dfrac{1}{q}(1-\dfrac{(v-c)(1-q)+q(v-c)(1-q)-s_1}{\tilde{p}_1-c}=\dfrac{1}{q}(1-\dfrac{(1-q)[(v-c)(1-q)+q(v-c)(1-q)-s_1]}{(v-c)(1-q)-s_1}\nonumber\\
& =1-\dfrac{(v-c)(1-q)^2}{(v-c)(1-q)-s_1}
\end{align}
Hence, $\psi_2(\cdot)$ does not have any jump at $\tilde{p}_1$.  Hence, the result follows.
\end{proof}
The continuity of $\psi_{1,Y}(\cdot)$ and $\psi_{1,N}(\cdot)$ can be easily concluded.  Note that the variations of $\psi_{1,Y}(\cdot)$, $\psi_{1,N}(\cdot)$ and $\psi_2(\cdot)$ are similar they differ only in the support and the scaling parameters.

Now, we are ready to state  the main result of this section.

\begin{theorem}\label{thm:mixedunequalcost}
Consider the following strategy profile: Primary $1$ selects $Y$ w.p. $p_1$ and $N$ w.p. $1-p_1$ ($p_1$ is given in (\ref{eq:p1})) and primary $2$ selects $N$ w.p. $1$. While selecting $Y$, if the channel of primary $2$ is available, then primary $1$ selects its price according to $\psi_{1,Y}(\cdot)$, otherwise it selects $v$ w.p. $1$. While selecting $N$, primary $1$ selects its price according to $\psi_{1,N}(\cdot)$. Primary $2$ selects its price according to $\psi_2(\cdot)$. 

The above strategy profile is the unique NE when $s_2\geq q(v-c)(1-q)$ and $s_1<q(v-c)(1-q)$. The expected payoff that primary $1$ attains is $(v-c)(1-q)+q(v-c)(1-q)-s_1$ and the expected payoff of primary $2$ is $(v-c)(1-q)$.
\end{theorem}
{\em Discussion}: Note that when $s_1<q(v-c)(1-q)$ and $s_2\geq q(v-c)(1-q)$, the payoff of primary $1$ is higher compared to the primary $2$. Apparently, when $s_1$ is low, then primary $1$ takes advantage of the acquired CSI and gains more compared to primary $2$. Primary $2$ can not do the same as the cost $s_2$ is high. The expected payoff of primary $1$ also increases with the decrease in $s_1$. Note that the threshold above which primary $i$ only selects $N$ is $q(v-c)(1-q)$; the threshold is the same for both the players. 

The probability $p_1$ increases with decrease in $s_1$, hence, primary $1$ is more likely to select $Y$ with the decrease in $s_1$. 

Since $p_1$ increases as $s_1$ decreases. $\tilde{p}_1$ also  increases as $s_1$ decreases (from (\ref{eq:tildep1})) . Thus, $\psi_{1,Y}(\cdot)$ has larger support as $s_1$ decreases. 

Under $Y$, primary $1$ selects a price from the interval $[\tilde{p},\tilde{p}_1]$ when the channel of primary $2$ is available; under $N$, primary $1$ selects a price from the interval $[\tilde{p}_1,v]$. Hence, primary $1$ selects higher price under $N$ as the uncertainty of the CSI of other primary increases. Also note that $\psi_2(\cdot)$ overlaps both with $\psi_{1,Y}(\cdot)$ and $\psi_{1,N}(\cdot)$. 

Also note that $\psi_2(\cdot)$ has a jump at $v$. Thus, primary $2$ selects $v$ with a positive probability. Intuitively, primary $1$ selects $Y$ with a higher probability. Thus, primary $1$ knows the channel state of primary $2$ with a higher probability and thus and selects a lower price. In response, primary $2$ has two options-- i) selects a high price with high probability ( at least it can gain more when the channel of the primary $1$ is not available), ii) selects a low price ( it can increase the probability of winning). Our result shows that primary $2$ selects the first option. 

\subsubsection{Proof of Theorem~\ref{thm:mixedunequalcost}}
First, we show that there is no profitable deviation for primary $1$ when primary $2$ follows the prescribed strategy stated in Theorem~\ref{thm:mixedunequalcost} (Case I), subsequently, we show that there is also no profitable deviation for primary $2$ when primary $1$ follows the prescribed strategy stated in Theorem~\ref{thm:mixedunequalcost} (Case II).

Case I: In the first step (i), we show that primary $1$ can attain a maximum expected payoff of $(v-c)(1-q)+q(v-c)(1-q)-s$ under $Y$. Next in step (ii), we show that primary $1$ can attain a maximum expected payoff of $(v-c)(1-q)+q(v-c)(1-q)-s_1$ under $N$. Finally in step (iii), we show that primary $1$ attains the maximum expected payoff following the strategy which will show that primary $1$ does not have any profitable unilateral deviation. 

Step (i): Primary $1$ selects $Y$.  Suppose that the channel of primary $2$ is available, then primary $1$ will know that w.p. $1$. At any $x$ such that $\tilde{p}\leq x\leq \tilde{p}_1$ the primary $1$ gets under $Y$  is
\begin{align}
(x-c)(1-\psi_2(x))-s_1=(v-c)(1-q)-s_1\quad \text{from } (\ref{eq:psi2lowc1highc2}) \& (\ref{eq:tildepunequalcost}).
\end{align}
If primary $1$ selects a price strictly less than $\tilde{p}$, then, its payoff is strictly less than $\tilde{p}-c-s_1=(v-c)(1-q)-s_1$.

Now, at any $v>x\geq \tilde{p}_1$, the expected payoff of primary $1$ in this setting is 
\begin{align}
& (x-c)(1-\psi_2(x))-s_1\nonumber\\
& =(x-c)(1-\dfrac{1}{q}(1-\dfrac{(v-c)(1-q)+q(v-c)(1-q)-s_1}{x-c}))-s_1\text{from } (\ref{eq:psi2lowc1highc2})\nonumber\\
& =(x-c)(1-\dfrac{1}{q})+\dfrac{(v-c)(1-q)+q(v-c)(1-q)-s_1}{q}-s_1
\end{align}
Since $1/q>1$, thus, the supremum is attained at $x=\tilde{p}_1$. Now from (\ref{eq:tildep1}), the maximum value is
\begin{align}
(v-c)(1-q)-s_1
 \end{align}
 Since $\psi_2(\cdot)$ has a jump at $v$, thus, the expected payoff at $v$ is strictly less than the value at a price close to $v$. 
Hence, when the channel of primary $2$ is available, then the maximum expected payoff that primary $1$ can attain at $Y$ is $(v-c)(1-q)-s_1$ and it is attained at any price in the interval $[\tilde{p},\tilde{p}_1]$. 

Now, when the channel of primary $2$ is unavailable, the expected payoff of primary $1$ is $(v-c)-s_1$. Hence, the maximum expected payoff that primary $1$ attains in $Y$ is 
\begin{align}
& (v-c-s_1)(1-q)+q[(v-c)(1-q)-s_1]\nonumber\\
& =(v-c)(1-q)+q(v-c)(1-q)-s_1
\end{align}

Step (ii) Now suppose  primary $1$ selects $N$ and a price $x$ such that $\tilde{p}_1\leq x<v$. Primary $2$ selects a price less than $x$ if the channel of primary $2$ is available and selects a price less than or equal to $x$ (it occurs w.p. $q\psi_2(x)$). By the continuity of $\psi_2(\cdot)$ in the interval $[\tilde{p}_1,v)$, the expected payoff of primary $1$ at $x$ is
\begin{align}
(x-c)(1-q\psi_2(x))=(v-c)(1-q)+q(v-c)(1-q)-s_1\quad \text{from } (\ref{eq:psi2lowc1highc2}).
\end{align}
Since $\psi_2(\cdot)$ has a jump at $v$, thus, the expected payoff of primary $1$ is strictly less at a price close to $v$ compared to $v$. Thus, the expected payoff of primary $1$ at $v$ is strictly less than $(v-c)(1-q)+q(v-c)(1-q)-s_1$.

Now,  at any $x$ such that $\tilde{p}\leq x\leq \tilde{p}_1$, the expected payoff of primary $1$ under $N$ is
\begin{align}
(x-c)(1-q\psi_2(x))=(x-c)(1-q(1-\dfrac{(v-c)(1-q)}{x-c}))\nonumber\\
=(x-c)(1-q)+q(v-c)(1-q)\quad \text{from } (\ref{eq:psi2lowc1highc2}).
\end{align}
The supremum is attained at $x=\tilde{p}_1$. Putting the value of $\tilde{p}_1$ from (\ref{eq:tildep1}) we obtain
\begin{align}
(v-c)(1-q)-s_1+q(v-c)(1-q)
\end{align}
The expected payoff at a price strictly less than $\tilde{p}$ will fetch a payoff which is strictly less than the payoff attained at $\tilde{p}$. Thus, primary $1$ can attain at most an expected payoff of $(v-c)(1-q)+q(v-c)(1-q)-s_1$ under $N$. The maximum expected payoff is attained at any price in the interval $[\tilde{p}_1,v)$. 

Step (iii): Hence, we show that the primary $1$ can attain an expected payoff of $(v-c)(1-q)+q(v-c)(1-q)-s$ under either $Y$ or $N$. Thus, any randomization between $Y$ and $N$ will also give an expected payoff of $(v-c)(1-q)+q(v-c)(1-q)-s_1$. Now, under the strategy profile the expected payoff is also $(v-c)(1-q)+q(v-c)(1-q)-s_1$, hence, primary $1$ does not have any profitable deviation.

Case II: Now, we show that primary $2$ does not have any profitable deviation when primary $1$ selects the prescribed strategy stated in the theorem. Towards this end, we first show in Step (i) that any price in the interval $[\tilde{p},\tilde{p}_1]$ will give an expected payoff of $(v-c)(1-q)$ to primary $2$ when it selects $N$, subsequently, we show that any price in the interval $[\tilde{p}_1,v]$ will also provide an expected payoff of $(v-c)(1-q)$ to primary $2$ when it selects $N$ . In step (iii), we show that  any price $x<\tilde{p}$ will give a strictly lower payoff compared to $(v-c)(1-q)$ when it selects $N$. Finally in step (iv), we show that if primary $2$ selects $Y$, then it can only get a payoff of at most $(v-c)(1-q)$ when $s_2\geq q(v-c)(1-q)$. This will show that primary $2$ attains the maximum expected payoff of $(v-c)(1-q)$ and it is attained when it selects $N$ and selects a price in the interval $[\tilde{p},v]$. 

Step (i): Note that when $x\in [\tilde{p},\tilde{p}_1]$ primary $1$ can select a price less than $x$ only when the channel of primary $1$ is available and primary $1$ selects $Y$, thus, at any $x$ such that $x\in [\tilde{p},\tilde{p}_1]$, the expected payoff of primary $2$ is
\begin{align}
(x-c)(1-qp_1\psi_{1,Y}(x))=(v-c)(1-q)\quad \text{from } (\ref{eq:psi1ylowc1highc2}).
\end{align}
Step (ii) When $x\in [\tilde{p}_1,v]$, then primary $1$ selects a price lower than $x$ only if the channel of primary $1$ is available and either it selects $Y$ or while selecting $N$ it selects a price less than $x$. Hence, the expected payoff of primary $2$ is
\begin{align}
(x-c)(1-qp_1-q(1-p_1)\psi_{1,N}(x))=(v-c)(1-q) \quad \text{from } (\ref{eq:psi1nlowc1highc2}).
\end{align}
(iii) At any price less than $\tilde{p}$ will fetch a payoff which is strictly less than $\tilde{p}-c$. However, $\tilde{p}-c=(v-c)(1-q)$. Thus, the expected payoff of primary $2$ is strictly less than $(v-c)(1-q)$ at any price less than $\tilde{p}$. Hence, primary $2$ can only attain a maximum expected payoff of $(v-c)(1-q)$ and it is attained at the prices in the interval $[\tilde{p},v]$.

(iv) Now, suppose primary $2$ selects $Y$. If the channel of primary $1$ is available, then at any $x\in [\tilde{p},\tilde{p}_1]$, it will get an expected payoff of 
\begin{align}
(x-c)(1-p_1\psi_{1,Y}(x))-s_2=(x-c)(1-\dfrac{1}{q}(1-\dfrac{(v-c)(1-q)}{x-c})-s_2\nonumber\\
=(x-c)(1-1/q)+(v-c)(1-q)/q-s_2
\end{align}
Since $1/q>1$, thus the above is maximized at $x=\tilde{p}$, and the maximum expected payoff is $(v-c)(1-q)-s_2$ since $\tilde{p}-c=(v-c)(1-q)$.

Now, if primary $2$ selects a price in the interval $[\tilde{p}_1,v]$ when the channel of primary $1$ is available, then the expected payoff of primary $2$ is 
\begin{align}
& (x-c)(1-(1-p_1)\psi_{1,N}(x)-p_1)-s_2=(x-c)(1-p_1-\dfrac{1}{q}(1-\dfrac{(v-c)(1-q)}{x-c}-qp_1))-s_2\quad \text{from } (\ref{eq:psi1nlowc1highc2})\nonumber\\
& =(x-c)(1-1/q)+(v-c)(1-q)/q-s_2\nonumber\\
& <(\tilde{p}-c)(1-1/q)+(v-c)(1-q)/q-s_2=(v-c)(1-q)-s_2\quad \text{since } \tilde{p}-c=(v-c)(1-q).\nonumber
\end{align}
Thus,  primary $2$ attains an expected payoff of at most $(v-c)(1-q)-s_2$ when it selects  $Y$ and the channel of primary $1$ is available. 

Now, when the channel of primary $1$ is unavailable the payoff that primary $2$ earns is $v-c-s_2$. Hence, the maximum expected payoff that primary $2$ can earn by selecting $Y$ is 
\begin{align}
q[(v-c)(1-q)-s_2]+(1-q)(v-c-s_2)=q(v-c)(1-q)+(v-c)(1-q)-s_2
\end{align}
when $s_2\geq q(v-c)(1-q)$, thus, the primary attains at most a payoff of $(v-c)(1-q)$. Hence, primary $2$ also does not have any profitable deviation.\qed

\subsection{Low $s_1,s_2$}
Lastly, we show that if $s_2<q(v-c)(1-q)$ then, there exists an NE where primary $2$ also randomizes between $Y$ and $N$. 

Again, we introduce some price distribution functions
\begin{align}\label{eq:psi1ylowc1lowc2}
\psi_{1,Y}=& 
0, \quad x<L\nonumber\\
& \dfrac{1}{p_1}(1-\dfrac{L-c}{x-c}), \quad L\leq x\leq \tilde{p}_2\nonumber\\
& \dfrac{1}{p_1q}(1-\dfrac{(v-c)(1-q)}{x-c}), \quad \tilde{p}_2< x\leq \tilde{p}_1\nonumber\\
& 1, \quad x>\tilde{p}_1
\end{align}
\begin{align}\label{eq:psi2ylowc1lowc2}
\psi_{2,Y}=&
0,\quad x<L\nonumber\\
& \dfrac{1}{p_2}(1-\dfrac{L-c}{x-c}),\quad L\leq x\leq \tilde{p}_2\nonumber\\
& 1,\quad x>\tilde{p}_2
\end{align}
\begin{align}\label{eq:psi1nlowc1lowc2}
\psi_{1,N}=&
0,\quad x<\tilde{p}_1\nonumber\\
& \dfrac{1}{q(1-p_1)}(1-\dfrac{(v-c)(1-q)}{x-c}-p_1q),\quad \tilde{p}_1\leq x\leq v\nonumber\\
& 1, \quad x>v
\end{align}
and 
\begin{align}\label{eq:psi2nlowc1lowc2}
\psi_{2,N}=&
0,\quad x<\tilde{p}_2\nonumber\\
& \dfrac{1}{1-p_2}(1-\dfrac{L-c}{x-c}-p_2),\quad \tilde{p}_2\leq x<\tilde{p}_1\nonumber\\
& \dfrac{1}{q(1-p_2)}(1-\dfrac{\tilde{p}_{1,N}-c}{x-c}-p_2q)\quad \tilde{p}_1\leq x<v\nonumber\\
& 1,\quad x\geq v
\end{align}
where 
\begin{align}\label{eq:l_unequalcost}
\tilde{p}_{1,N}-c&=(v-c)(1-q)+s_2-s_1,\quad
L-c =\dfrac{s_2}{q},\quad \tilde{p}=(v-c)(1-q)+c.\nonumber\\
\tilde{p}_2-c& =(v-c)(1-q)/(1-p_2q),\quad
\tilde{p}_1-c=\dfrac{(v-c)(1-q)}{1-p_1q}
\end{align}
The values of $p_1$ and $p_2$ are
\begin{align}\label{eq:p_1andp_2}
p_1=\dfrac{q(v-c)(1-q)-s_1}{q(v-c)(1-q)-qs_1}, p_2=\dfrac{q(v-c)(1-q)-s_2}{q(v-c)(1-q)-qs_2}
\end{align}
Since $s_1<s_2$,  $p_1>p_2$. Note that $p_i, i=1,2$ only depends on $s_i$.
Using the values of $p_1$ and $p_2$, we obtain from (\ref{eq:l_unequalcost})
\begin{align}\label{eq:tildep2}
\tilde{p}_2-c=\dfrac{(v-c)(1-q)-s_2}{1-q}=\dfrac{s_2}{q(1-p_2)}
\end{align}
and
\begin{align}\label{eq:tildep1mixed}
\tilde{p}_1-c=\dfrac{(v-c)(1-q)-s_1}{1-q}=\dfrac{s_1}{q(1-p_1)}
\end{align}
We also use the above equalities throughout this section.

It is easy to discern that $\psi_{1,N}(\cdot)$ and $\psi_{2,Y}(\cdot)$ are continuous. Now, we show that $\psi_{1,Y}(\cdot)$ is also continuous.
\begin{obs}
$\psi_{1,Y}(\cdot)$ is a continuous function.
\end{obs}
\begin{proof}
We only show that $\psi_{1,Y}(\cdot)$ is continuous at $\tilde{p}_2$, it is easy to discern that $\psi_{1,Y}(\cdot)$ at other values. 
Note from (\ref{eq:tildep2}) and (\ref{eq:l_unequalcost}), the left hand limit is
\begin{align}
\dfrac{1}{p_1}(1-\dfrac{L-c}{\tilde{p}_2-c})=\dfrac{1}{p_1}p_2\nonumber
\end{align}
Now from (\ref{eq:l_unequalcost}), the right hand limit and the value of $\psi_{1,Y}(\cdot)$ at $\tilde{p}_2$ is
\begin{align}
\dfrac{1}{qp_1}(1-\dfrac{(v-c)(1-q)}{\tilde{p}_2-c})=\dfrac{1}{qp_1}p_2q=p_2/p_1\nonumber
\end{align}
Hence, $\psi_{1,Y}(\cdot)$ does not have a jump at $\tilde{p}_2$.
\end{proof}

Next, we show that $\psi_{2,N}(\cdot)$ is  continuous everywhere but at $v$.
\begin{obs}
$\psi_{2,N}(\cdot)$ is continuous except at $v$.
\end{obs}
\begin{proof}
Since $s_1<s_2$, thus, it is easy to discern that $\psi_{2,N}(\cdot)$ has a jump at $v$. 

Now, we show that $\psi_{2,N}(\cdot)$ does not have a jump at $\tilde{p}_1$. It is easy to discern that $\psi_{2,N}(\cdot)$ can not have a jump at any other point. 

From (\ref{eq:tildep1mixed}), we have
\begin{align}
& \dfrac{(v-c)(1-q)+s_2-s_1}{\tilde{p}_1-c}\nonumber\\
& =[(v-c)(1-q)+s_2-s_1]\dfrac{1-q}{(v-c)(1-q)-s_1}
=1-q+s_2\dfrac{1-q}{(v-c)(1-q)-s_1}
\end{align}
Thus, the left hand limit at $\tilde{p}_1$ is
\begin{align}
\dfrac{1}{q(1-p_2)}(1-\dfrac{(v-c)(1-q)+s_2-s_1}{\tilde{p}_1-c}-p_2q)
=1-\dfrac{s_2(1-q)}{q(1-p_2)[(v-c)(1-q)-s_1]}
\end{align}
Now, the right hand limit and the value of $\psi_{2,N}(\cdot)$ at $\tilde{p}_1$ is
\begin{align}
\dfrac{1}{1-p_2}(1-\dfrac{L-c}{\tilde{p}_1-c}-p_2)
=1-\dfrac{s_2(1-q)}{q(1-p_2)[(v-c)(1-q)-s_1]}\quad \text{from (\ref{eq:tildep1mixed})}.
\end{align}
Hence,  $\psi_{2,N}(\cdot)$ does not have a jump at $\tilde{p}_1$.
\end{proof}
Note that though $\psi_{2,N}(\cdot)$ has a jump at $v$, the variation of $\psi_{2,N}$ with $x$ is similar to the other distributions $\psi_{1,N}(\cdot), \psi_{1,Y}(\cdot)$ and $\psi_{2,Y}(\cdot)$. 

Now, we are ready to state the main result in this section.
\begin{theorem}\label{thm:bothmixedunequalcost}
Consider the following strategy profile: Primary $1$ selects $Y$ w.p. $p_1$ and $N$ w.p. $1-p_1$ and primary $2$ selects $Y$ w.p. $p_2$ and $N$ w.p. $1-p_2$ where $p_1$ and $p_2$ are given in (\ref{eq:p_1andp_2}).  While selecting $Y$, primary $i=1,2$  selects its price according to $\psi_{i,Y}(\cdot)$ when the channel of primary  $j, j\neq i$ is available and will select the price $v$ if the channel of primary $j$ is unavailable; while selecting $N$, primary $i$ selects its price according to $\psi_{i,N}(\cdot)$.

The above strategy profile is the unique NE when $s_2<q(v-c)(1-q)$. The expected payoff that primary $1$ attains is $(v-c)(1-q)+s_2-s_1$ and primary $2$ attains is $(v-c)(1-q)$.
\end{theorem}
{\em Discussion}: Since $s_2>s_1$, thus, the expected payoff of primary $1$ is higher compared to primary $2$. Since $s_2<q(v-c)(1-q)$, thus by Theorem~\ref{thm:mixedunequalcost} the expected payoff of primary $1$ is lower compared to the setting where $s_2\geq q(v-c)(1-q)$. Note also that the payoff of primary $1$ decreases with $s_2$, but increases with $s_1$.  Thus, if $s_2$ decreases it only impacts the payoff of primary $1$, it does not affect the payoff of primary $2$. The payoff of primary $1$ also becomes closer to the payoff of primary $2$ as $s_2$ becomes closer to $s_1$ and ultimately becomes equal when $s_2=s_1$ which we have already seen in Section~\ref{sec:equalcost} where we analyze the scenario when both the primaries have identical cost to acquire the CSI of their respective competitors i.e. $s_2=s_1=s$.

$p_1$ ($p_2$,resp.) increases with the decrease in $s_1$ ($s_2$, resp.). Thus, primaries are more likely to select $Y$ as the cost $s_1, s_2$ decrease.  When $s\rightarrow 0$, then $p_1\rightarrow 1$, and when $s_2\rightarrow 0$, then $p_2\rightarrow 1$. Since $s_1<s_2$, thus, $p_1>p_2$.  Primary $1$ is more likely to select $Y$ compared to primary $2$.

Note that under $Y$, primary $1$ (primary $2$, resp.) selects its price from the interval $[L,\tilde{p}_1]$ ($[L,\tilde{p}_2]$, resp.) when the channel of its competitor is available. Note that $L$ is less than $\tilde{p}$ (given in (\ref{eq:tildepunequalcost})) as $s_2<q(v-c)(1-q)$. Hence, each primary selects its price from a larger interval when both primaries randomize between $Y$ and $N$. Also note that $L$ decreases as $s_2$ decreases. However, $L$ is independent of $s_1$.  Since $\tilde{p}_1>\tilde{p}_2$, hence, under $Y$ primary $1$ selects its price from a wider interval compared to primary $2$. Also note from (\ref{eq:tildep2}) and (\ref{eq:tildep1mixed}) $\tilde{p}_1$ and $\tilde{p}_2$ increase as $s_1$  and $s_2$ decrease respectively. Hence, $\psi_{i,Y}(\cdot)$ has larger supports as $s_i$ decreases. 

$\psi_{2,N}(\cdot)$ has a jump at $v$ similar to the setting when $s_2\geq q(v-c)(1-q)$ but $s_1<q(v-c)(1-q)$.

\subsubsection{Proof of Theorem~\ref{thm:bothmixedunequalcost}}
First, we show that there is no profitable unilateral deviation for primary $1$ when primary $2$ follows the strategy prescribed in Theorem~\ref{thm:bothmixedunequalcost} (Case I). Subsequently, we also show that there is no unilateral profitable deviation for primary $2$ when primary $1$ follows the strategy prescribed in Theorem~\ref{thm:bothmixedunequalcost} (Case II).

Case I: First, we show that under $Y$, when the channel of primary $2$ is available, then the maximum expected payoff that primary $1$ can attain is $L-c-s_1$ and it is attained at prices in the interval $[L,\tilde{p}_1]$ (Step i). Next, we show that when the channel of primary $2$ is not available, then under $Y$ the primary $1$ attains a payoff of $(v-c)-s_1$ (by selecting price $v$). This shows that the maximum expected payoff that primary $1$ attains under $Y$ is $(v-c)(1-q)+s_2-s_1$ (Step ii). Subsequently, we show that under $N$, the maximum expected payoff that a primary can attain is $(v-c)(1-q)+s_2-s_1$ and it is attained only at the prices in the interval $[\tilde{p}_1,v)$ (Step iii).  Finally, we show that the maximum expected payoff attained by primary $1$ is $(v-c)(1-q)+s_2-s_1$ and it is attained when primary $1$ follows the strategy (Step iv). 

(i) Suppose that primary $1$ selects $Y$ and the channel of primary $2$ is available. First, we show that at any price $x\in [L,\tilde{p}_2]$, the expected payoff of primary $1$ in this case is $L-c-s_1$ (Step i.a.). Subsequently, we show that at any price $x\in [\tilde{p}_2,\tilde{p}_1]$ the expected payoff of primary $1$ is $L-c-s_1$ (Step i.b.). Next, we show that at any price price $x\in [\tilde{p}_1,v]$ the expected payoff of primary $1$ is at most $L-c-s_1$ (Step i.c.).  Note that at a price less than $L$ will fetch a payoff of strictly less than the payoff of $L-c-s_1$. Hence, this will show  that when the channel of primary $2$ is available, then under $Y$ the maximum expected payoff attained by primary $1$ is $L-c-s_1$ and it is attained only at prices $[L,\tilde{p}_1]$.

Step i.a.: Suppose that primary $1$ selects a price $x\in [L,\tilde{p}_2]$. Since the primary $2$ selects a price less than or equal to $x$ if it selects $Y$ and then selects a price less than or equal to $x$ (it occurs w.p. $p_2\psi_{2,Y}(x)$).  Thus, the expected payoff of primary $1$ is
\begin{align}
(x-c)(1-p_2\psi_{2,Y}(x))-s_1=L-c-s_1\quad \text{from } (\ref{eq:psi2ylowc1lowc2}).\nonumber
\end{align}
Step i.b.: Now, suppose that primary $1$ selects a price $x$  in the interval $[\tilde{p}_2,\tilde{p}_1]$. Primary $2$ selects a price less than or equal to $x$ if it either selects $Y$ or it selects $N$ and then selects a price less than or equal to $x$. Thus, at any price $x$, the expected payoff of primary $1$ is
\begin{align}
(x-c)(1-(1-p_2)\psi_{2,N}(x)-p_2)-s_1
=L-c-s_1\quad \text{from } (\ref{eq:psi2nlowc1lowc2}).\nonumber
\end{align}
Thus, at any price in the interval $[\tilde{p}_2,\tilde{p}_1]$ fetches the primary an expected payoff of $L-c-s_1$. 

Step i.c: Now, at any price in the interval $[\tilde{p}_1,v)$ the expected payoff of primary $1$ is
\begin{align}\label{eq:maxy}
& (x-c)(1-(1-p_2)\psi_{2,N}(x)-p_2)-s_1&\nonumber\\
& =(x-c)(1-\dfrac{1}{q}(1-\dfrac{(v-c)(1-q)+s_2-s_1}{x-c}-qp_2)-p_2)-s_1\text{from } (\ref{eq:psi2nlowc1lowc2})\nonumber\\
& =(x-c)(1-1/q)+((v-c)(1-q)+s_2-s_1)/q-s_1
\end{align}
Since the co-efficient of $(x-c)$ is negative, the above  is maximized at $x=\tilde{p}_1$. Now, from (\ref{eq:tildep1mixed})
\begin{align}
(\tilde{p}_1-c)(1-1/q)
=-((v-c)(1-q)-s_1)/q\nonumber
\end{align}
Thus, (\ref{eq:maxy}) is upper bounded by $L-c-s_1$. 

Since $\psi_{2,N}(\cdot)$ has a jump at $v$, thus, the expected payoff of primary $1$ at $v$ is strictly less than the expected payoff at a price close to $v$. Thus, the maximum expected payoff attained in the interval $[\tilde{p}_1,v]$ is $s_2/q-s_1=L-c-s_1$ (by(\ref{eq:l_unequalcost})). 

Thus, under $Y$ when the channel of primary $2$ is available, the maximum expected payoff of primary $1$ is $L-c-s_1$ and it is attained in every price in the interval $[L,\tilde{p}_1]$. 

(ii) When the channel of primary $2$ is unavailable, then the payoff of primary $1$ is $(v-c)-s_1$ as primary $1$ selects $v$ and is still capable of selling its channel. Hence, the maximum expected payoff of primary $1$ under $Y$ is
\begin{align}
qs_2/q+(v-c)(1-q)-s_1=(v-c)(1-q)+s_2-s_1
\end{align}

(iii) Now, we show that under $N$, the maximum expected payoff of primary $1$ is at most $(v-c)(1-q)+s_2-s_1$ and it is attained when it follows the strategy $\psi_{1,N}(\cdot)$. Toward this end, we first show that if primary $1$ selects a price in the interval $[\tilde{p}_1,v]$ the expected payoff is $(v-c)(1-q)+s_2-s_1$ and it is attained at any price in the interval $[\tilde{p}_1,v)$  (Step iii.a). Subsequently, we show that if primary $1$ selects a price in the interval $[\tilde{p}_2,\tilde{p}_1]$, then the expected payoff under $N$ is at most $(v-c)(1-q)+s_2-s_1$ (Step iii.b.). Finally, we show that even if primary $1$ selects a price in the interval $[L,\tilde{p}_2)$, then the expected payoff is also at most $(v-c)(1-q)+s_2-s_1$ under $N$ (Step iii.c.). Note that at any price less than $L$ will fetch a payoff which is strictly less than the payoff at $L$. Hence, this will complete the proof. 

Step iii.a: Now, suppose primary $1$ selects a price $x\in[\tilde{p}_1,v)$. Now, primary $2$ selects a price less than or equal to $x$ if the channel of primary $2$ is available, and one of the two things occur--(i) primary $2$ selects $Y$, and (ii) primary $2$ selects $N$ and selects a price less than or equal to $x$. (i) occurs w.p. $p_2$ and (ii) occurs w.p. $(1-p_2)\psi_{2,N}(x)$. The channel of primary $2$ is available w.p. $q$. Thus, at $x$, the expected payoff of primary $1$ is 
\begin{align}
(x-c)(1-(1-p_2)q\psi_{2,N}(x)-p_2q)=(v-c)(1-q)+s_2-s_1\quad \text{from } (\ref{eq:psi2nlowc1lowc2})
\end{align}
Since $\psi_{2,N}(\cdot)$ has a jump at $v$, hence, primary $1$ attains strictly higher payoff at a price just below $v$ compared to the payoff at $v$.  Hence, the expected payoff at $v$ is strictly less than $(v-c)(1-q)+s_2-s_1$.

Step iii.b: Now, if primary $1$ selects any price in the interval $[\tilde{p}_2,\tilde{p}_1]$, then its expected payoff is 
\begin{align}
& (x-c)(1-(1-p_2)q\psi_{2,N}(x)-p_2q)=(x-c)(1-q(1-\dfrac{L-c}{x-c}))\quad \text{from } (\ref{eq:psi2nlowc1lowc2})\nonumber\\
& =(x-c)(1-q)+[L-c]q
\end{align}
which is maximized at $x=\tilde{p}_1$. Now, from (\ref{eq:tildep1mixed}) $(\tilde{p}_1-c)(1-q)=(v-c)(1-q)-s_1$. Since $(L-c)q=s_2$ (by (\ref{eq:l_unequalcost}), hence the maximum expected payoff is 
\begin{align}
(v-c)(1-q)-s_1+s_2
\end{align}
Step iii.c: Suppose that the primary $1$ selects a price $x$ in the interval $[L,\tilde{p}_2)$. Since primary $2$ does not select any price in this interval when the channel of primary $2$ is unavailable or primary $2$ selects $N$. Thus, the expected payoff of primary $1$ at $x$ is
\begin{align}
& (x-c)(1-p_2q\psi_{2,Y}(x))\nonumber\\
& =(x-c)(1-q(1-\dfrac{L-c}{x-c}))\quad(\text{from } (\ref{eq:psi2ylowc1lowc2}))=(x-c)(1-q)+(L-c)q\nonumber\\
& <(\tilde{p}_1-c)(1-q)+(L-c)q
=(v-c)(1-q)-s_1+s_2
\end{align}
Hence, under $N$, the maximum expected payoff that a primary can attain is $(v-c)(1-q)+s_2-s_1$ and this is attained at any price in the interval $[\tilde{p}_1,v)$.
 
 (iv) Under $Y$ or $N$, the maximum expected payoff that primary $1$ can attain is $(v-c)(1-q)+s_2-s_1$. Thus, any randomization of $Y$ and $N$ also yields the same expected payoff. Under the strategy profile, the primary $1$ attains the payoff of $(v-c)(1-q)+s_2-s_1$, hence, primary $1$ does not have any profitable unilateral deviation.
 
 Case II: We now show that primary $2$ also does not have any profitable unilateral deviation. Toward this end we first show that when primary $2$ selects $Y$ and primary $1$ is available, then the maximum expected payoff of primary $2$ is $L-c-s_2$ and it is attained at any price in the interval $[L,\tilde{p}_2]$ (Step i). Subsequently, we show that under $Y$, the maximum expected attained by primary $2$ is $(v-c)(1-q)$ and it is attained when primary $2$ follows the strategy (Step ii). Subsequently, we show that under $N$, the maximum expected payoff that primary $2$ attains is $(v-c)(1-q)$ and it is attained at a price in the interval $[\tilde{p}_2,v]$ (Step iii). Finally, we show that the maximum expected payoff that primary $2$ attains is $(v-c)(1-q)$ and it is attained when primary $2$ follows the strategy (Step iv). 

Step i: Suppose that primary $2$ selects $Y$ and primary $1$ is available. We show that the maximum expected payoff that primary $2$ attains is $L-c-s_2$ and it is attained at any price in the interval $[L,\tilde{p}_2]$. Toward this end, we first show that at any price $[L,\tilde{p}_2]$, the expected payoff is $L-c-s_2$ (Step i.a.). At any price in the interval $[\tilde{p}_2,\tilde{p}_1]$ and $[\tilde{p}_1,v]$ the expected payoff is at most $L-c-s_2$ (Steps i.b. and i.c. respectively). This will complete the proof.

Step i.a.: Suppose $x\in [L,\tilde{p}_2]$. Since primary $2$ selects $Y$, primary $2$ knows that primary $1$ is available. Primary $1$ selects a price in the interval $[L,\tilde{p}_2]$ if primary $1$ selects $Y$ (which occurs w.p. $p_1$). Thus, the expected payoff of primary $2$ at $x$ is
\begin{align}
(x-c)(1-p_1\psi_{1,Y}(x))-s_2=L-c-s_2 \quad \text{from } (\ref{eq:psi1ylowc1lowc2}).
\end{align}
{\em Thus, the expected payoff at any price $x\in [L,\tilde{p}_2]$ is $L-c-s_2$.}

Step i.b.: Suppose $x\in [\tilde{p}_2,\tilde{p}_1]$. The expected payoff of primary $2$ at $x$ is 
\begin{align}
& (x-c)(1-p_1\psi_{1,Y}(x))-s_2=(x-c)(1-\dfrac{1}{q}(1-\dfrac{\tilde{p}-c}{x-c}))-s_2\quad \text{from } (\ref{eq:psi1ylowc1lowc2}).\nonumber\\
& =(x-c)(1-1/q)+(\tilde{p}-c)/q-s_2.
\end{align}
Since the coefficient of $x$ is negative, the above is maximized at $\tilde{p}_2$. Thus, the expected payoff is upper bounded by
\begin{align}\label{eq:upperbound_unequalcost}
& (\tilde{p}_2-c)(1-1/q)+(\tilde{p}-c)/q-s_2=(\tilde{p}_2-c)(1-1/q)+(v-c)(1-q)/q-s_2\quad \text{from } (\ref{eq:l_unequalcost})\nonumber\\
& =s_2/q-(v-c)(1-q)/q+(v-c)(1-q)/q-s_2\quad \text{from } (\ref{eq:tildep2})\nonumber\\
& =L-c-s_2\quad \text{since } L-c=s_2/q\quad  (cf. (\ref{eq:l_unequalcost})).
\end{align}
Thus, at any $x\in [\tilde{p}_2,\tilde{p}_1]$ the maximum expected payoff of primary $2$ is $L-c-s_2$.

Step i.c.: Now, suppose $x\in [\tilde{p}_1,v]$. The expected payoff of primary $2$ at $x$ is 
\begin{align}
& (x-c)(1-(1-p_1)\psi_{1,N}(x)-p_1)-s_2=(x-c)(1-\dfrac{1}{q}(1-\dfrac{\tilde{p}-c}{x-c}-p_1q)-p_1)-s_2\quad \text{from } (\ref{eq:psi1nlowc1lowc2})\nonumber\\
&= (x-c)(1-1/q)+\dfrac{\tilde{p}-c}{q}-s_2\nonumber\\
& <(\tilde{p}_2-c)(1-1/q)+\dfrac{\tilde{p}-c}{q}-s_2\quad \text{since } \tilde{p}_2<\tilde{p}_1,
\quad =L-c-s_2\quad \text{from } (\ref{eq:upperbound_unequalcost}).
\end{align}
Thus, from Steps i.a., i.b. and i.c. the maximum expected payoff attained by primary $2$ in this case is $L-c-s_2$ and it is attained at the prices in the interval $[L,\tilde{p}_2]$.

Step ii: When primary $1$ is unavailable, then primary $2$ attains a payoff of $v-c-s_2$. Hence, under $Y$, the maximum expected payoff of primary $2$ is
\begin{align}
(L-c-s_2)q+(v-c-s_2)(1-q)=q(L-c)+(v-c)(1-q)-s_2=(v-c)(1-q)\quad \text{from } (\ref{eq:l_unequalcost}).
\end{align} It is attained when primary $2$ follows the strategy. 

Step iii: Now, we show that when primary $2$ selects $N$, then, its maximum expected payoff is $(v-c)(1-q)$ and it is attained at any price in the interval $[\tilde{p}_2,v]$. Toward this end, we show that at any price in the intervals $[\tilde{p}_2,\tilde{p}_1]$ and $[\tilde{p}_1,v]$, the maximum expected payoff of primary $2$ is $(v-c)(1-q)$ (Steps iii.a. and iii.b.). Subsequently, we show that the maximum expected payoff attained by the primary at any $x\in [L,\tilde{p}_2]$, the maximum expected payoff attained by primary $2$ is $(v-c)(1-q)$. 

Step iii.a: Suppose primary $2$ selects a price $x\in [\tilde{p}_2,\tilde{p}_1]$. Since primary $2$ selects $N$, thus, it only knows that primary $1$ is available w.p. $q$. Thus, at $x$, the expected payoff of primary $2$ at $x$ is 
\begin{align}
(x-c)(1-p_1q\psi_{1,Y}(x))=(v-c)(1-q)\quad \text{from } (\ref{eq:psi1ylowc1lowc2})
\end{align}

Step iii.b.: Suppose primary $2$ selects a price $x\in [\tilde{p}_1,v]$. Then the expected payoff of primary $2$ at $x$ is
\begin{align}
(x-c)(1-p_1q-(1-p_1)q\psi_{1,N}(x))=(v-c)(1-q) \quad \text{from } (\ref{eq:psi1nlowc1lowc2}).
\end{align}
From Steps iii.a. and iii.b., the expected payoff of primary $2$ at $[\tilde{p}_2,v]$ is $(v-c)(1-q)$.

Step iii.c: Now, suppose primary $2$ selects a price $x\in [L,\tilde{p}_2]$ is
\begin{align}
(x-c)(1-p_1q\psi_{1,Y}(x))=(x-c)(1-q)+(L-c)q \quad \text{from } (\ref{eq:psi1ylowc1lowc2})
\end{align}
The above is maximized at $x=\tilde{p}_2$ as the coefficient of $x$ is positive. Hence, the maximum expected payoff of primary $2$ is upper bounded by
\begin{align}
& (\tilde{p}_2-c)(1-q)+(L-c)q=(v-c)(1-q)-s_2+(L-c)q\quad \text{from } (\ref{eq:tildep2}).\nonumber\\
& =(v-c)(1-q)\quad \text{since } (L-c)q=s_2\quad \text{from } (\ref{eq:l_unequalcost}).
\end{align}
Thus, the maximum expected payoff of primary $2$ is $(v-c)(1-q)$ and it is attained only at prices in the interval $x\in [\tilde{p}_2,v]$ (by Steps iii.a. and iii.b.).

Step iv: By Step (ii), the maximum expected payoff of primary $2$ under $Y$ is $(v-c)(1-q)$. From Step (iii), the maximum expected payoff of primary $2$ under $N$ is $(v-c)(1-q)$. Thus, any randomization between $Y$ and $N$ will yield a maximum expected payoff of $(v-c)(1-q)$. This maximum expected payoff is attained by primary $2$ when it follows the strategy.\qed

  \section{Unequal Channel availability probabilities}\label{sec:unequalavail}
We, now, consider the setting, where different primaries may have different availability probabilities depicted in Section~\ref{sec:differentavail_model}.  Without loss of generality, we assume that the channel of primary $1$ is available w.p. $q_1$ and the channel of primary $2$ is available w.p. $q_2$ where $q_1>q_2$. 

\subsection{Goals}
The impact of different availability probabilities on the frequency with which a primary selects $Y$  can not be readily concluded.  If primary $1$ acquires the CSI of primary $2$, it will more often find that that the channel of primary $2$ is unavailable which may increase its payoff. However, primary $2$ itself may also acquire the CSI of primary $1$ and select a lower price, in response primary $1$  selects a lower price which may reduce the payoff of primary $1$.  Even if the NE strategy is of the form $[T,p]$, the values of the thresholds may be different for different primaries. Additionally, it is not clear whether the threshold will be higher for primary $1$. This is because if the availability probability of primary $2$ is low, primary $1$ may select $Y$ for very small values of $s$, but primary $2$ may still select $Y$ for larger values of $s$ as the channel availability probability of the primary $1$ is higher. 

The impact of different availability probabilities on the payoff of each primary is also not apriori clear.  Conventional wisdom suggests that as $s$ decreases the payoff of a primary should not decrease. However, the conventional wisdom is not definitive because of the following. Since the channel of primary $1$ is available with a higher probability, when primary $2$ acquires the CSI of primary $1$, then, primary $1$ selects a lower price more often which may reduce the payoff of primary $2$.   The pricing strategy also inherently depends on the frequency with which a primary selects $Y$.  We resolve all these quandaries. 
\subsection{Results}\label{sec:results_unequalavail}
We first discuss the main insights provided by our analysis. 
\begin{itemize}
\item The NE strategy for primary $i$ is of the form $[T_i,p_i]$ (Definition \ref{defn:classtp}).   However, $T_i$ is different for different primaries due to different availability probabilities.  Our result shows that $T_1=q_2(v-c)(1-q_2)$ (Theorem~\ref{thm:nq1q2}) and $T_2=q_2(v-c)(1-q_1)/(1-q_1+q_2)$ (Theorem~\ref{thm:ynq1q2}) where $T_1>T_2$.  Note that in the basic model, we have shown that the threshold depends on the variance of the availability of the competitor\rq{}s channel. Thus, $T_1=q_2(v-c)(1-q_2)$ is expected. However, the expression of $T_2$ is surprising.   Additionally, we show that $T_1>T_2$ which is again not completely intuitive. Also note that when $T_2\leq s<T_1$, primary $2$ selects $Y$ w.p. $0$, but primary $1$ selects $Y$ w.p. $p_1$. Even when $s<T_2$, $p_i$s are different with $p_1>p_2$ (Theorem~\ref{thm:yyq1q2}). Thus, primary $1$ selects $Y$ with a higher probability. 
\item Different availability probabilities also lead to different payoffs for the primaries. In contrast to the basic model,   the expected payoff of primary $1$ is higher than that of primary $2$ when primary $1$ selects $Y$ with positive probability (Theorems~\ref{thm:ynq1q2} and \ref{thm:yyq1q2}). Additionally, the expected payoff of primary $2$ decreases as $s$ decreases.  Thus, the expected payoff of a primary decreases with the ease of acquiring the CSI which negates the conventional wisdom. Intuitively, since primary $1$ selects $Y$ with a higher probability as $s$ decreases, it selects a lower price when the channel of primary $2$ is available. In response, primary $2$ either must select a high price (so that, it can get a high payoff in the event when the channel of primary $1$ is unavailable) or select a low price (so that, it can increase its probability of winning). Since the channel of primary $1$ is available with a higher probability, the first option fetches a lower payoff compared to the latter one.  Thus, primary $2$ also selects a lower price. Thus, the expected payoff of primary $2$ decreases as $s$ decreases.  The expected payoff of primary $2$ becomes close to that of primary $1$ as the difference between $q_1$ and $q_2$ decreases.
\item Price strategies also exhibit some similarities with the basic model. Specifically, primary $i$ selects its price from the interval $[L,\tilde{p}_i]$ ($[\tilde{p}_i,v]$, resp.) when it selects $Y$ ($N$, resp.) and the channel of the competitor is available. However, since primaries have different availability probabilities, the price selection strategies also have some differences compared to the basic model. For example, $\tilde{p}_1>\tilde{p}_2$. Thus, primary $2$ selects a lower price when it selects $Y$ and the channel of primary $1$ is available. Additionally,   primary $1$ selects $v$ with a positive probability when it selects $N$ and the probability decreases as $q_2$ becomes close to $q_1$. Thus, primary $1$ selects a price from a distribution function which has a discontinuity whereas in the basic model, each primary selects its price from a continuous distribution function. Intuitively, since primary $1$ has higher channel availability probability,  it selects a higher price when it selects $N$. 
\end{itemize}
\subsection{High $s$}
Our first result shows that 
\begin{theorem}\label{thm:nq1q2}
If $s\geq q_2(v-c)(1-q_2)$ then in the unique NE, both the primaries select $N$ w.p. $1$. The expected payoff of both the primaries is $(v-c)(1-q_2)$
\end{theorem}
Note that when $s\geq q_2(v-c)(1-q_2)$ both the primaries attain identical expected payoff though the availability probabilities are different. 
\begin{proof}
When both players select $N$, then the setting becomes equivalent to the setting where primary $1$ (primary $2$, resp.) only knows that the channel of primary $2$ (primary $1$, resp.) is available w.p. $q_1$ ($q_2$ resp.). The above setting has already been considered in \cite{Gaurav1}. From \cite{Gaurav1},
\begin{lem}\label{lm:nandnunequalavail}
In the unique NE pricing strategy under $N$,  primary $i$ should select its pricing strategy using $\psi_i(\cdot)$, where
\begin{align}
\psi_{1}(x)=\begin{cases}
0\quad x<\bar{p}\nonumber\\
\dfrac{1}{q_1}(1-\dfrac{(v-c)(1-q_2)}{x-c})\quad \bar{p}\leq x< v\nonumber\\
1\quad x>v
\end{cases}
\end{align}
\begin{align}
\psi_{2}(x)=\begin{cases}
0\quad x<\bar{p}\nonumber\\
\dfrac{1}{q_2}(1-\dfrac{(v-c)(1-q_2)}{x-c})\quad \bar{p}\leq x\leq v\nonumber\\
1\quad x>v
\end{cases}
\end{align}
where $\bar{p}-c=(v-c)(1-q_2)$. $\psi_1(\cdot)$ has a jump of $\dfrac{q_1-q2}{q_1}$ at $v$.
\end{lem}
It is easy to show that each primary will attain an expected payoff of $(v-c)(1-q_2)$. Now, we show that each primary can not attain higher payoff by selecting $Y$. First, we show that primary $1$ can not attain more by selecting $Y$ (Step i). Subsequently, we show that primary $2$ can not attain more by selecting $Y$ (Step ii).

Step i: Suppose that primary $1$ deviates and selects $Y$. When the channel of primary $2$ is available, then the expected payoff of primary $1$ at any $[\bar{p},v]$ is
\begin{align}
(x-c)(1-\psi_2(x))-s=(x-c)(1-1/q_2)+(v-c)(1-q_2)/q_2-s
\end{align}
The above is maximized at $\bar{p}$ as the co-efficient of $x$ is negative. Since $\bar{p}-c=(v-c)(1-q_2)$ (from Lemma~\ref{lm:nandnunequalavail}), hence, the above is upper bounded by 
\begin{align}
(v-c)(1-q_2)(1-1/q_2)+(v-c)(1-q_2)/q_2-s=\bar{p}-c-s.
\end{align}
The price at any $x<\bar{p}$ will fetch an expected payoff of strictly less than $\bar{p}-c-s$. Thus, the maximum expected payoff that primary $1$ attains in this setting is $\bar{p}-c-s$.

When the channel of primary $2$ is not available, then the payoff of primary $1$ is $(v-c)-s$. Hence, the maximum expected payoff that primary $1$ can attain by deviating unilaterally is 
\begin{align}
q_2(v-c)(1-q_2)+(v-c)(1-q_2)-s
\end{align}
When $s\geq q_2(v-c)(1-q_2)$, then, primary $1$ will attain an expected payoff of strictly less than $(v-c)(1-q_2)$. Hence, primary $1$ does not have any profitable unilateral deviation.

 Step ii: By applying the similar method we can show that the maximum expected payoff attained by primary $2$ under $Y$ is 
\begin{align}
q_1(v-c)(1-q_2)+(v-c)(1-q_1)-s
\end{align}
However, the above is strictly less than $q_2(v-c)(1-q_2)+(v-c)(1-q_2)-s$ since $q_1>q_2$. If $s\geq q_2(v-c)(1-q_2)$, then the maximum expected payoff that primary $2$ will attain under $Y$ is strictly less than $(v-c)(1-q_2)$. However, primary $2$ attains an expected payoff of $(v-c)(1-q_2)$ following the strategy profile at $N$, hence, primary $2$ does not have any profitable unilateral deviation. 
\end{proof}
\subsection{$s$ is neither too high nor too low}
Now, we show that when $\dfrac{q_2(v-c)(1-q_1)}{1-q_1+q_2}\leq s<q_2(v-c)(1-q_2)$ then there is an NE where primary $1$ randomizes between $Y$ and $N$, however, primary $2$ only selects $N$.
 First, we introduce some price distribution functions.
\begin{align}\label{eq:psiymidc1}
\psi_{Y}(x)=&
0,\quad x<L\nonumber\\
& \dfrac{1}{p_1q_1}(1-\dfrac{L-c}{x-c})\quad L\leq x\leq \tilde{p}\nonumber\\
&1, \quad x>\tilde{p}
\end{align}
and
\begin{align}\label{eq:psi1nmidc1}
\psi_{1,N}(x)=&
0,\quad x<\tilde{p}\nonumber\\
& \dfrac{1}{(1-p_1)q_1}(1-\dfrac{L-c}{x-c}-p_1q_1)\quad \tilde{p}\leq x<v\nonumber\\
& 1,\quad x\geq v
\end{align}
\begin{align}\label{eq:psinmidc1}
\psi_{N}(x)=&
0, \quad x<L\nonumber\\
& (1-\dfrac{L-c}{x-c})\quad L\leq x< \tilde{p} \nonumber\\
& \dfrac{1}{q_2}(1-\dfrac{(v-c)(1-q_2)}{x-c})\quad \tilde{p}\leq x\leq v\nonumber\\
& 1, \quad x>v
\end{align}
where 
\begin{eqnarray}
L-c=\dfrac{s}{q_2}\label{eq:l}\\
\tilde{p}-c=\dfrac{L-c}{1-p_1q_1}\label{eq:tildepunequal}\\
p_1=\dfrac{(v-c)(1-q_2)-s/q_2}{q_1(v-c)(1-q_2)-q_1s}\label{eq:p}
\end{eqnarray}
  Replacing the value of $p_1$ from (\ref{eq:p}) in $\tilde{p}$ we obtain
\begin{align}
\tilde{p}-c& =\dfrac{(s/q_2)(q_1(v-c)(1-q_2)-q_1s)}{q_1(v-c)(1-q_2)-q_1s-q_1(v-c)(1-q_2)+sq_1/q_2}=\dfrac{(v-c)(1-q_2)-s}{(1-q_2)}\label{eq:tildep_unequalavail}
\end{align}
It is easy to discern that $\psi_Y(\cdot)$ is  continuous. We, also, show that $\psi_{N}(\cdot)$ is a continuous function.
\begin{obs}
$\psi_{N}(\cdot)$ is a continuous function.
\end{obs}
\begin{proof}
It is easy to discern that $\psi_N(\cdot)$ is continuous every except $x=\tilde{p}$.  Now, we show that $\psi_{N}(\cdot)$ is also continuous at $\tilde{p}$. The left hand limit of $\psi_N(\cdot)$ is $1-\dfrac{L-c}{\tilde{p}-c}$.  Now, the right hand limit is 
\begin{align}
& \dfrac{1}{q_2}(1-\dfrac{(v-c)(1-q_2)}{\tilde{p}-c})=\dfrac{1}{q_2}(1-\dfrac{(\tilde{p}-c)(1-q_2)+(L-c)q_2}{\tilde{p}-c})\quad \text{from (\ref{eq:tildep_unequalavail}) and (\ref{eq:l})}\nonumber\\
& =1-\dfrac{L-c}{\tilde{p}-c}
\end{align}
Hence, $\psi_N(\cdot)$ does not have any jump at $\tilde{p}$.
\end{proof}
Now, we show that $\psi_{1,N}(\cdot)$ has a jump at $v$.
\begin{obs}
$\psi_{1,N}(\cdot)$ has a jump at $v$. 
\end{obs}
\begin{proof}
Note that $s\geq q_2(v-c)(1-q_1)/(1-q_1+q_2)$, hence, $\dfrac{L-c}{v-c}\geq \dfrac{1-q_1}{1-q_1+q_2}>1-q_1$ as $q_1>q_2$ and $L-c=s/q_2$. Thus,
\begin{align}
1-\dfrac{1}{(1-p_1)q_1}(1-\dfrac{L-c}{v-c}-p_1q_1)
>1-\dfrac{1}{(1-p_1)q_1}(1-1+q_1-p_1q_1)
=0
\end{align}
Hence, $\psi_{1,N}(\cdot)$ has a jump at $v$. 
\end{proof}
Now, we are ready to state the main result.
\begin{theorem}\label{thm:ynq1q2}
Consider the following strategy profile: Primary $1$ selects $Y$ w.p. $p_1$ (given in (\ref{eq:p})) and $N$ w.p. $1-p_1$ and primary $2$ selects $N$ w.p. $1$. While selecting $Y$, primary $1$ selects its price according to $\psi_{Y}(\cdot)$ when the channel of primary $2$ is available and selects $v$ when the channel of primary $2$ is unavailable. While selecting $N$, primary $1$ selects its price according to $\psi_{1,N}(\cdot)$. Primary $2$ selects its price according to $\psi_{N}(\cdot)$. 

The above strategy profile is the unique NE when $q_2(v-c)(1-q_2)>s\geq q_2(v-c)(1-q_1)/(1-q_1+q_2)$\footnote{Note that $(1-q_2)(1-q_1)+q_2(1-q_2)>(1-q_1)$ as $q_2<q_1$, thus, $(1-q_2)>\dfrac{1-q_1}{1-q_1+q_2}$}. The expected payoff of primary $1$ is $(v-c)(1-q_2)$ and the expected payoff of primary $2$ is $s/q_2$.
\end{theorem}
{\em Discussion}: Since $s<q_2(v-c)(1-q_2)$, hence, $s/q_2<(v-c)(1-q_2)$. Thus, the expected payoff of primary $2$ is lower compared to the expected payoff of primary $1$. The expected payoff of primary $2$ decreases as $s$ decreases.  {\em This negates conventional wisdom which suggests that the expected payoff of a primary should increase when $s$ decreases. }

Note that the support of $\psi_{Y}$ is $[L,\tilde{p}]$ and the support of $\psi_{N}$ is $[L,v]$. Thus, the support of $\psi_{Y}$ and $\psi_{N}$ overlap with each other. Also note that $\psi_{1,N}(\cdot)$ has a jump at $v$, where as $\psi_N(\cdot)$ does not have any jump. Intuitively, since primary $1$ has higher availability probability, primary $1$ selects higher prices with higher probabilities.

$p_1$ increases with decrease in $s$. $L$ decreases when $s$ decreases (by (\ref{eq:l})).  Thus, primaries select their prices from a larger interval as $s$ decreases. Also note that $L$ only depends on $q_2$, it is independent of $q_1$. 

 Also note from (\ref{eq:tildep_unequalavail}) that $\tilde{p}$ increases as $s$ decreases. Thus, $\psi_Y(\cdot)$ has a larger support as $s$ decreases and primary $1$ selects its price from a larger interval   when it selects $Y$, and the channel of primary $2$ is available.
 
Now, we prove the above theorem.
\begin{proof}
First, we show that primary $1$ does not have any profitable deviation when primary $2$ follows the strategy prescribed in Theorem~\ref{thm:ynq1q2} (Case I). Subsequently, we show that primary $2$ also does not have any profitable unilateral deviation when primary $2$ follows the strategy prescribed in Theorem~\ref{thm:ynq1q2} (Case II). 

Case I: First, we show that under $Y$, if the channel of primary $2$ is available, then primary $1$ can attain a maximum expected payoff of $L-c-s$  (Step i). When the channel of primary $2$ is unavailable primary $1$ will attain the payoff $v-c-s$. Thus, it shows that under $Y$, the expected payoff that primary $1$ attains is $(v-c)(1-q_2)$ (Step ii). Subsequently, we show that under $N$, the maximum expected payoff that primary $1$ attains is $(v-c)(1-q_2)$ (Step iii). Finally, we show that primary $1$ achieves the maximum expected payoff under the prescribed strategy (Step iv).

(i): Suppose primary $1$ selects $Y$ and the channel of primary $2$ is available. The at any price $x\in [L,\tilde{p}]$, the expected payoff of primary $1$ under $Y$ is
\begin{align}
(x-c)(1-\psi_N(x))-s=L-c-s\quad \text{from } (\ref{eq:psinmidc1}).
\end{align}
At any price less than or equal to $L$ will fetch a payoff which is strictly less than $L-c-s$. 

At any price $x$ in the interval $[\tilde{p},v]$ the expected payoff of primary $1$ is
\begin{align}
& (x-c)(1-\psi_N(x))-s=(x-c)(1-\dfrac{1}{q_2}(1-\dfrac{(v-c)(1-q_2)}{x-c}))-s\quad \text{from } (\ref{eq:psinmidc1})\nonumber\\
& =(x-c)(1-1/q_2)+(v-c)(1-q_2)/q_2-s\nonumber
\end{align}
The above is maximized at $\tilde{p}$. Putting $x=\tilde{p}$, and from (\ref{eq:tildep_unequalavail}) we obtain
\begin{align}
& \dfrac{(v-c)(1-q_2)-s}{1-q_2}(1-1/q_2)+(v-c)(1-q_2)/q_2-s=s/q_2-s\nonumber\\
& =L-c-s\quad \text{by (\ref{eq:l})}
\end{align}
Hence, the maximum expected payoff attained by primary $1$ is $L-c-s$ and it is attained at any price in the interval $[L,\tilde{p}]$. 

(ii): Now, when the channel of primary $2$ is not available, then the payoff that primary $1$ achieves under $Y$ is $(v-c)-s$. Hence, the maximum expected payoff that primary $1$ can achieve under $Y$ is
\begin{align}
q_2(L-c-s)+(v-c-s)(1-q_2)=(v-c)(1-q_2)
\end{align}
By following the strategy profile, primary $1$ achieves the above payoff under $Y$.

(iii): When primary $1$ selects $N$, then it only knows that the channel of the primary $2$  is available w.p. $q_2$. Thus, under $N$, at any price $x$ in the interval $[\tilde{p},v]$ the expected payoff of primary $1$ is
\begin{align}
(x-c)(1-q_2\psi_{N}(x))=(v-c)(1-q_2)\quad 
\end{align}
Similarly, at any price $x$ in the interval $[L,\tilde{p}]$, the expected payoff of primary $1$ is
\begin{align}
& (x-c)(1-q_2\psi_N(x))=(x-c)(1-q_2(1-\dfrac{L-c}{x-c}))\quad \text{from } (\ref{eq:psinmidc1})\nonumber\\
& =(x-c)(1-q_2)+(L-c)q_2
\end{align}
The above is maximized at $x=\tilde{p}$. Putting the value of $\tilde{p}$ we obtain
\begin{align}
& (\tilde{p}-c)(1-q_2)+(L-c)q_2 =(\tilde{p}-c)(1-q_2)+s\quad \text{by (\ref{eq:l})}\nonumber\\
& =(v-c)(1-q_2)\quad \text{by (\ref{eq:tildep_unequalavail})}.
\end{align}
At any price less than $L$ fetches a payoff of strictly less than $L$ which is less than $(v-c)(1-q_2)$. Hence, the maximum expected payoff that primary $1$ attains under $N$ is $(v-c)(1-q_2)$. This is achieved at any price in the interval $[\tilde{p},v]$.

(iv): We have shown that under $Y$ or under $N$, the maximum expected payoff that primary $1$ can attain is $(v-c)(1-q_2)$. Thus, any randomization between $Y$ and $N$ also yields at most an expected payoff of $(v-c)(1-q_2)$. Primary $1$ attains the above payoff when it follows the prescribed strategy. Hence, primary $1$ does not have any profitable deviation.

Case II: Now, we show that primary $2$ does not have any profitable unilateral deviation. Toward this end we first show that under $N$, the maximum expected payoff that primary $2$ attains is $L-c$ (Step i). Subsequently, we show that under $N$, the primary $2$ attains the maximum expected payoff $L-c$ when it selects price in the interval $[L,v)$ (Step ii). Subsequently, we show that if primary $2$ deviates and selects $Y$, then it can only attain a payoff of at most $L-c$ when $s\geq \dfrac{q_2(v-c)(1-q_1)}{1-q_1+q_2}$ (Step iii).

Step (i): Suppose that primary $2$ selects $N$. Suppose that primary $2$ selects a  price $x$ in the interval $[L,\tilde{p}]$. If the channel of primary $1$ is available, then it selects a price less than or equal to $x$ where $x\in [L,\tilde{p}]$ if primary $1$ selects $Y$ and then selects a price less than or equal to $x$. The above occurs w.p. $p_1\psi_{Y}(x)$. The channel of primary $1$ is available w.p. $q_1$. Hence, by the continuity of $\psi_{1,Y}(\cdot)$ at $x$ the expected payoff of primary $2$ under $N$ is
\begin{align}
(x-c)(1-p_1q_1\psi_{Y}(x))=L-c\quad \text{from } (\ref{eq:psiymidc1}).
\end{align}
Now, suppose that primary $2$ selects a price $x$ from the interval $[\tilde{p},v)$. If the channel of primary $1$ is available, then primary $1$ selects a price less than or equal to $x$ when $x\in [\tilde{p},v)$ if--i) primary $1$ selects $Y$   or ii) primary $1$ selects $N$ and selects a price less than or equal to $x$. (i) occurs with probability $p_1$ and (ii) occurs with probability $(1-p_1)\psi_{1,N}(x)$. The channel of primary $1$ is available w.p. $q_1$. Since $\psi_{1,N}$ is continuous in $[\tilde{p},v)$, the expected payoff of primary $2$ at $x$ is
\begin{align}
(x-c)(1-(1-p_1)q_1\psi_{1,N}(x)-p_1q_1)=L-c\quad \text{from } (\ref{eq:psi1nmidc1}).
\end{align}
Since $\psi_{1,N}(\cdot)$ has a jump at $v$,  the expected payoff at $v$ is strictly less than the expected payoff just below $v$. On the other hand a price less than $L$ will fetch a payoff strictly less than $L-c$. Hence, the maximum expected payoff that primary $2$ can attain under $N$ is $L-c$.

Step ii: Primary $2$ attains a payoff of $L-c$ under $N$ and it is attained only at prices in the interval $[L,v)$. 

Step (iii): Now, we show that if primary $2$ selects $Y$, then it will not attain a payoff higher than $s/q_2$. Towards this end, we show when the channel of primary $1$ is available, then the maximum expected payoff attained by primary $2$ is $L-c-s$ (Step iii.a). When the channel of primary $1$ is unavailable, then the payoff attained by primary $2$ is $v-c-s$. Subsequently, we show that the maximum expected payoff attained under $Y$ is at most $L-c$ when $s\geq \dfrac{q_2(v-c)(1-q_1)}{1-q_1+q_2}$ (Step iii.b.). This will complete the proof.  

Step iii.a: When the channel of primary $1$ is available then the expected payoff of primary $2$ at any price $x$ in the interval $[L,\tilde{p}]$ is 
\begin{align}
(x-c)(1-p_1\psi_{Y}(x))-s=(x-c)(1-1/q_1)+(L-c)/q_1-s\quad \text{from } (\ref{eq:psiymidc1}).
\end{align}
The above is maximized at $x=L$ since the co-efficient of $x$ is negative, hence, the maximum value is $L-c-s$.

Similarly,   the expected payoff of primary $2$ at any price $x$ in the interval $[\tilde{p},v)$ is
\begin{align}
& (x-c)(1-(1-p_1)\psi_{1,N}(x)-p_1)-s=(x-c)(1-1/q_1)+(L-c)/q_1-s\quad \text{from } (\ref{eq:psi1nmidc1})\nonumber\\
& <(L-c)-s\quad \text{since } \tilde{p}>L.
\end{align}
The payoff at a price less than $L$ fetches a payoff which is strictly less than $L-c-s$. 

Hence, the maximum expected payoff attained by primary $2$ when the channel of primary $1$ is available is $L-c-s$.

Step iii.b: When the channel of primary $1$ is unavailable, then the payoff that primary $2$ attains is $(v-c)-s$. Hence, the expected payoff of primary $2$ under $Y$ is
\begin{align}
& q_1(L-c-s)+(v-c-s)(1-q_1) =\dfrac{q_1s}{q_2}+(v-c)(1-q_1)-s\nonumber\\
& =(v-c)(1-q_1)+\dfrac{s(q_1-q_2)}{q_2}\nonumber\\
& \leq s/q_2\quad \text{as }  q_2(v-c)(1-q_1)/(1-q_1+q_2)\leq s\nonumber\\
& = L-c
\end{align}
But primary $2$ attains $L-c$ under $N$ following the strategy $\psi_N(\cdot)$. Thus, primary $2$ does not have any profitable unilateral deviation. Hence, the result follows.
\end{proof}
\subsection{Low $s$}
Now, we show that when $s<\dfrac{q_2(v-c)(1-q_1)}{1-q_1+q_2}$ then there exists an NE where both the primaries randomize between $Y$ and $N$. Again we first introduce some pricing distributions. 
\begin{align}\label{eq:psi1ylowc1}
\psi_{1,Y}(x)=& 
 0, \quad x<L\nonumber\\
& \dfrac{1}{p_1}(1-\dfrac{L-c}{x-c})\quad L\leq x<\tilde{p}_2\nonumber\\
& \dfrac{1}{p_1q_1}(1-\dfrac{\bar{p}-c}{x-c})\quad \tilde{p}_2\leq x\leq \tilde{p}_1\nonumber\\
& 1, \quad x>\tilde{p}_1
\end{align}
\begin{align}\label{eq:psi2ylowc1}
\psi_{2,Y}(x)=& 
0,\quad x<L\nonumber\\
& \dfrac{1}{p_2}(1-\dfrac{L-c}{x-c})\quad L\leq x\leq \tilde{p}_2\nonumber\\
& 1,\quad x>\tilde{p}_2
\end{align}
\begin{align}\label{eq:psi1nlowc1}
\psi_{1,N}(x)=& 0,\quad x<\tilde{p}_1\nonumber\\
& \dfrac{1}{(1-p_1)q_1}(1-\dfrac{\bar{p}-c}{x-c}-p_1q_1)\quad \tilde{p}_1\leq x<v\nonumber\\
& 1,\quad x\geq v
\end{align}
\begin{align}\label{eq:psi2nlowc1}
\psi_{2,N}(x)=&
0,\quad x<\tilde{p}_2\nonumber\\
& \dfrac{1}{1-p_2}(1-\dfrac{L-c}{x-c}-p_2),\quad \tilde{p}_2\leq x<\tilde{p}_1\nonumber\\
& \dfrac{1}{(1-p_2)q_2}(1-\dfrac{(v-c)(1-q_2)}{x-c}-p_2q_2)\quad \tilde{p}_1\leq x\leq v\nonumber\\
& 1,\quad x>v
\end{align}
where 
\begin{eqnarray}
  \bar{p}-c=(v-c)(1-q_1)+s(q_1-q_2)/q_2\label{eq:barp_avail}\\
 p_1=\dfrac{q_1(v-c)(1-q_2)-s(q_1/q_2-q_1+q_2)}{q_1(v-c)(1-q_2)-q_1s}\label{eq:p1q1q2}\\
p_2=\dfrac{q_2(v-c)(1-q_1)-s(1-q_1+q_2)}{q_2(v-c)(1-q_1)-q_2s}\label{eq:pq1q2}\\
L-c=s/q_2\nonumber\\
 \tilde{p}_2-c=\dfrac{(v-c)(1-q_1)+s(q_1-q_2)/q_2}{1-p_2q_1},\quad 
 \tilde{p}_1-c=\dfrac{(v-c)(1-q_1)+s(q_1-q_2)/q_2}{1-p_1q_1}\label{eq:lunequalavail}
\end{eqnarray}

First, we show some results which we use throughout this section. Replacing the value of $p_2$  in $\tilde{p}_2$ we have
\begin{align}\label{eq:tildep2vss}
\tilde{p}_2-c& =\dfrac{[q_2(v-c)(1-q_1)-q_2s][(v-c)(1-q_1)+s(q_1-q_2)/q_2]}{q_2(v-c)(1-q_1)^2+s(q_1-q_1^2+q_1q_2-q_2)}\nonumber\\
& =\dfrac{[q_2(v-c)(1-q_1)-q_2s][(v-c)(1-q_1)+s(q_1-q_2)/q_2]}{q_2(v-c)(1-q_1)^2+sq_2(1-q_1)(q_1-q_2)/q_2}=\dfrac{(v-c)(1-q_1)-s}{(1-q_1)}
\end{align}
\begin{align}\label{eq:tildep2avail}
\dfrac{L-c}{1-p_2}& =\dfrac{s}{q_2(1-p_2)}=\dfrac{((v-c)(1-q_1)-s)s}{s(1-q_1)}\nonumber\\
& =\dfrac{(v-c)(1-q_1)-s}{1-q_1}
=\tilde{p}_2-c
\end{align}
Also note from (\ref{eq:lunequalavail}) and (\ref{eq:p1q1q2}) that
\begin{align}\label{eq:lmixed}
& \tilde{p}_1-c=\dfrac{[(v-c)(1-q_1)+s(q_1-q_2)/q_2][q_1(v-c)(1-q_2)-q_1s]}{q_1(v-c)(1-q_2)-sq_1-q_1^2(v-c)(1-q_2)+sq_1^2/q_2-sq_1^2+sq_1q_2}\nonumber\\
& =\dfrac{[(v-c)(1-q_1)+s(q_1-q_2)/q_2][q_1(v-c)(1-q_2)-q_1s]}{q_1(1-q_2)[(v-c)(1-q_1)+s(q_1-q_2)/q_2]}
 =\dfrac{(v-c)(1-q_2)-s}{1-q_2}
\end{align}
$\psi_{1,Y}(\cdot), \psi_{2,Y}(\cdot)$ and $\psi_{2,N}(\cdot)$ are continuous. However, $\psi_{1,N}(\cdot)$ is not continuous.
\begin{obs}
$\psi_{1,N}(\cdot)$ is continuous except at $v$. 
\end{obs}
\begin{proof}
It is easy to discern the continuity at every other point except $v$. Note from (\ref{eq:psi1nlowc1}) $\psi_{1,N}(\cdot)$ has a jump of $ \dfrac{s(q_1-q_2)}{(v-c)(1-p_1)q_1q_2}$ at $v$. 
\end{proof}
\begin{obs}
$\psi_{1,Y}(\cdot)$ is continuous.
\end{obs}
\begin{proof}
It is easy to verify that $\psi_{1,Y}(\cdot)$ (cf. (\ref{eq:psi1ylowc1})) is continuous everywhere except at $\tilde{p}_2$. We now show that it is also continuous at $\tilde{p}_2$. The left hand limit at $\tilde{p}_2$ is 
\begin{align}
\dfrac{1}{p_1}(1-\dfrac{L-c}{\tilde{p}_2-c})& =\dfrac{1}{p_1}(1-(1-p_2))\quad \text{from }(\ref{eq:tildep2avail})\nonumber\\
& =\dfrac{p_2}{p_1}
\end{align}
The right hand limit (cf.(\ref{eq:psi1ylowc1})) is 
\begin{align}
\dfrac{1}{p_1q_1}(1-\dfrac{\bar{p}-c}{\tilde{p}_2-c})& =\dfrac{1}{p_1q_1}(1-1+p_2q_1)\quad \text{from } (\ref{eq:lunequalavail})\nonumber\\
& =\dfrac{p_2}{p_1}
\end{align}
which is equal to the left hand limit. 
\end{proof}
\begin{obs}
$\psi_{2,N}(\cdot)$ is continuous.
\end{obs}
\begin{proof}
It is easy to verify that $\psi_{2,N}(\cdot)$ (cf. (\ref{eq:psi2nlowc1})) is continuous everywhere except at $\tilde{p}_1$. Now, we also show that $\psi_{2,N}(\cdot)$ is continuous at $\tilde{p}_1$. 
First note from (\ref{eq:lmixed}) that 
\begin{align}\label{eq:ide}
\dfrac{(v-c)(1-q_2)}{(1-q_2)(\tilde{p}_1-c)}=1+\dfrac{s}{(1-q_2)(\tilde{p}_1-c)}
\end{align}
The right hand limit at $\tilde{p}_1$ is 
\begin{align}
& \dfrac{1}{(1-p_2)q_2}(1-\dfrac{(v-c)(1-q_2)}{\tilde{p}_1-c}-p_2q_2)=\dfrac{1}{(1-p_2)q_2}(1-(1-q_2)-\dfrac{s}{\tilde{p}_1-c}-p_2q_2)\quad \text{from } (\ref{eq:ide})\nonumber\\
& =1-\dfrac{s}{(1-p_2)q_2(\tilde{p}_1-c)}
=1-\dfrac{L-c}{(1-p_2)(\tilde{p}_1-c)}\quad \text{since } L-c=s/q_2\nonumber\\
& =\dfrac{1}{(1-p_2)}(1-\dfrac{L-c}{\tilde{p}_1-c}-p_2)
\end{align}
which is the left hand limit (cf.(\ref{eq:psi2nlowc1})). Hence, the result follows. 
\end{proof}
Note from (\ref{eq:p1q1q2}) and (\ref{eq:pq1q2}) that $p_i$ $i=1,2$ both depend on $q_1$ and $q_2$.
Next, we show that $p_1>p_2$.
\begin{lem}\label{lm:p1p2}
$p_1>p_2$ when $q_1>q_2$.
\end{lem}
\begin{proof}
 From (\ref{eq:p1q1q2}) and (\ref{eq:pq1q2}), we need to show that
\begin{align}
\dfrac{q_1(v-c)(1-q_2)-s(q_1/q_2-q_1+q_2)}{q_1(v-c)(1-q_2)-q_1s}>\dfrac{q_2(v-c)(1-q_1)-s(1-q_1+q_2)}{q_2(v-c)(1-q_1)-q_2s}
\end{align}
By cross multiplication it is sufficient to show that
\begin{align}
& q_1q_2(v-c)^2(1-q_2)(1-q_1)-q_1q_2s(v-c)(1-q_2)-s(q_1/q_2-q_1+q_2)q_2(v-c)(1-q_1)+\nonumber\\& q_2s^2(q_1/q_2-q_1+q_2)>q_1q_2(v-c)^2(1-q_1)(1-q_2)-q_1q_2(v-c)(1-q_1)s-\nonumber\\& s(1-q_1+q_2)q_1(v-c)(1-q_2)+ s^2(1-q_1+q_2)q_1.\nonumber\\
& \text{Or,} (v-c)sq_1q_2(q_2-q_1)-s(v-c)[(q_1-q_1q_2+q_2^2)(1-q_1)-(q_1-q_1^2+q_1q_2)(1-q_2)]\nonumber\\
& +s^2(q_2^2-q_1q_2-q_1q_2+q_1^2)>0\nonumber
\end{align}
The last expression is $s^2(q_1-q_2)^2$ which is always positive when $q_1>q_2$. Thus, it is sufficient to show that
\begin{align}
& (v-c)sq_1q_2(q_2-q_1)-\nonumber\\&  s(v-c)[q_1-q_1q_2+q_2^2-q_1^2+q_1^2q_2-q_1q_2^2-q_1+q_1^2-q_1q_2+q_1q_2-q_1^2q_2+q_1q_2^2]>0\nonumber\\
& \text{Or, }(v-c)sq_1q_2(q_2-q_1)-s(v-c)[q_2^2-q_1q_2]>0\nonumber\\
& (v-c)s(q_1-q_2)(q_2-q_1q_2)>0
\end{align}
as $q_1>q_2$, the above expression is indeed positive. Hence, the result follows.
\end{proof}
Now, we are ready to state the main result of this section. 
\begin{theorem}\label{thm:yyq1q2}
Consider the following strategy profile: Primary $i$ selects $Y$ w.p. $p_i$ (cf. (\ref{eq:p1q1q2})\& (\ref{eq:pq1q2}))and $N$ w.p. $1-p_i$. While selecting $Y$, primary $i=1,2$ selects its price according to $\psi_{i,Y}(\cdot)$ when the channel of primary $j\neq i$ is available and selects $v$ when the channel of primary $j$ is unavailable. While selecting $N$, primary $i$ selects its price according to $\psi_{i,N}(\cdot)$. 

The above strategy profile is an NE when $s< q_2(v-c)(1-q_1)/(1-q_1+q_2)$. The expected payoff of primary $1$ is $(v-c)(1-q_2)$ and the expected payoff of primary $2$ is $(v-c)(1-q_1)+s(q_1-q_2)/q_2$.
\end{theorem}
\emph{Discussion}: Since $p_1>p_2$ (by Lemma~\ref{lm:p1p2}), primary $2$ selects $Y$ with a lower probability compared to primary $1$. Both $p_1$ and $p_2$ increase as $s$ decreases. Both $p_1$ and $p_2$ go to $1$ as $s\rightarrow 0$. Note that threshold $T_i$ above which primary $i$ selects only $N$ is higher for primary $1$ i.e. $T_1>T_2$. Hence, primary $1$ selects $Y$ for a wider value of $s$. 

Note that the expected payoff of primary $2$ decreases with the cost of acquiring the CSI $s$. This negates {\em conventional wisdom which suggests that the expected payoff of a  primary should increase as $s$ decreases}. The expected payoff of primary $1$ is independent of $s$. The expected payoff of primary $2$ is lower than that of primary $1$.  The expected payoff of primary $2$ becomes equal to that of the primary $1$ when $q_2$ becomes equal to $q_1$. 

Note that $\psi_{1,N}(\cdot)$ (see (\ref{eq:psi1nlowc1})) has a jump at $v$ since $q_1>q_2$. The jump decreases as the difference between $q_1$ and $q_2$ decreases. Since primary $1$ has a higher availability probability, thus, it selects higher prices when it selects $N$. $\psi_{1,N}(\cdot)$ is continuous elsewhere. It is easy to show that $\psi_{1,Y}, \psi_{2,Y}$ and $\psi_{2,N}$ are continuous everywhere. Note that $L$ decreases as $s$ decreases. Thus, a primary selects its price from a larger interval as $s$ decreases. $\tilde{p}_1$ and $\tilde{p}_2$ both decrease with $s$ (from (\ref{eq:tildep2vss}) and (\ref{eq:lmixed})). Hence, $\psi_{1,Y}(\cdot)$ and $\psi_{2,Y}(\cdot)$ have larger supports when $s$ decreases.

\subsubsection{Proof of Theorem~\ref{thm:yyq1q2}}
First, we show that primary $1$ does not have any profitable unilateral deviation  when primary $2$ follows the strategy prescribed in the theorem (Case I). Subsequently, we show that primary $2$ also does not have any profitable unilateral deviation when primary $1$ follows the strategy prescribed in the theorem (Case II).

Case I: First, we show that under $Y$, the maximum expected payoff that primary $1$ can attain is $(v-c)(1-q_2)$ (Step i). Toward this end, we first show that when primary $1$ selects $Y$ and the channel of the primary $2$ is available, then the expected payoff that primary $1$ will attain is at most $L-c-s$ (Step i.a.). When primary $1$ selects $Y$ and the channel of the primary $2$ is unavailable, then the payoff of the primary $1$ is $v-c-s$ which will  in turn show that the maximum payoff attained by primary $1$ under $Y$ is $(v-c)(1-q_2)$ (Step i.b). Subsequently, we show that under $N$, the maximum expected payoff that primary $1$ can attain is $(v-c)(1-q_2)$  (Step ii). Finally, we show that the maximum expected payoff is attained by primary $1$ when it follows the strategy profile (Step iii).

Step i.a: Suppose that primary $1$ selects $Y$ and the channel of primary $2$ is available.  Then, at any price $x$ in the interval $[L,\tilde{p}_2]$, the expected payoff of primary $1$ is
\begin{align}
(x-c)(1-p_2\psi_{2,Y}(x))-s=L-c-s\quad \text{from } (\ref{eq:psi2ylowc1})
\end{align}
Now suppose that primary $1$ selects a price $x$ in the interval $[\tilde{p}_2,\tilde{p}_1]$. Primary $2$ selects a price less than or equal to $x$ if (i) primary $2$ selects $Y$ which occurs w.p. $p_2$ and (ii) primary $2$ selects $N$ and then selects a price less than or equal to $x$ which occurs w.p. $(1-p_2)\psi_{2,N}(x)$. Thus, by the continuity of $\psi_{2,N}(\cdot)$ the expected payoff of primary $1$ at $x$ is
\begin{align}
(x-c)(1-p_2-(1-p_2)\psi_{2,N}(x))-s=L-c-s\quad \text{from } (\ref{eq:psi2nlowc1}).
\end{align}
Similarly, when primary $1$ selects a price $x$ from the interval $[\tilde{p}_1,v]$, then its expected payoff is 
\begin{align}
(x-c)(1-p_2-(1-p_2)\psi_{2,N}(x))-s=(x-c)(1-\dfrac{1}{q_2})+(v-c)(1-q_2)/q_2-s\quad \text{from } (\ref{eq:psi2nlowc1}). 
\end{align}
Thus, the above is maximized at $\tilde{p}_1$ since the coefficient of $x$ is negative. Hence, the maximum value is
\begin{align}
& (\tilde{p}_1-c)(1-1/q_2)+(v-c)(1-q_2)/q_2-s\nonumber\\
& =(\tilde{p}_1-c)(1-1/q_2)+(\tilde{p}_1-c)(1-q_2)/q_2+s/q_2-s\quad (\text{from (\ref{eq:lmixed})}) =L-c-s\quad \text{from  (\ref{eq:lunequalavail})}.\nonumber
\end{align}
Any price which is strictly less than $L$ will fetch a payoff of less than $L-c-s$. Hence, the maximum expected payoff that primary $1$ can attain is $L-c-s$ when the channel of primary $2$ is available and it is achieved at any price in the interval $[L,\tilde{p}_1]$.

Step i.b.: Note that the payoff that primary $1$ attains when the channel of primary $2$ is unavailable is $(v-c)-s$. Hence, the maximum expected payoff that primary $1$ attains under $Y$ is 
\begin{align}
& (L-c-s)q_2+(v-c-s)(1-q_2)=(v-c)(1-q_2)+(L-c)q_2-s\nonumber\\
& =(v-c)(1-q_2)\quad \text{from (\ref{eq:lunequalavail})}.
\end{align}

Step ii: Now, we show that if primary $1$ selects $N$, then, it will attain a maximum expected payoff of $(v-c)(1-q_2)$ and it is attained when it selects a price from the interval $[\tilde{p}_1,v]$. Towards this end, we first show that when primary $1$ selects a price in the interval $[\tilde{p}_1,v]$, then its expected payoff is $(v-c)(1-q_2)$ (Step ii.a.). Subsequently, we show that when primary $1$ selects a price from the interval $[\tilde{p}_2,\tilde{p}_1]$, then its expected payoff is at most $(v-c)(1-q_2)$ (Step ii.b.). Finally, we show that if primary $1$ selects a price in the interval $[L,\tilde{p}_2]$, then its expected payoff is less than $(v-c)(1-q_2)$ (Step ii.c.). Note that a price which is strictly less than $L$ will fetch a payoff which is strictly less than the payoff at $L$, hence, this will show that under $N$ the expected payoff of primary $1$ is $(v-c)(1-q_2)$ and it is attained when at prices in the interval $[\tilde{p}_1,v]$.

Step ii.a: Suppose that primary $1$  selects a price $x\in [\tilde{p}_1,v]$. Primary $2$ selects a price less than or equal to $x$ if the channel of primary $2$ is available and either primary $2$ selects $Y$ or it selects $N$ and then selects a price less than or equal to $x$. Thus, the probability that primary $2$ selects a price less than or equal to $x$ is $p_2q_2+(1-p_2)q_2\psi_{2,N}(x)$.  Thus, by the continuity of $\psi_{2,N}(\cdot)$, the expected payoff at $x$ is
\begin{align}
(x-c)(1-(1-p_2)q_2\psi_{2,N}(x)-p_2q_2)=(v-c)(1-q_2)\quad \text{from } (\ref{eq:psi2nlowc1}).
\end{align}
Step ii. b.: Similarly, at  price $x$ in the interval $[\tilde{p}_2,\tilde{p}_1]$, the expected payoff of primary $1$ is   
\begin{align}
(x-c)(1-p_2q_2-(1-p_2)q_2\psi_{2,N}(x))=(x-c)(1-q_2)+(L-c)q_2
\end{align}
The above is maximized at $\tilde{p}_1.$ From (\ref{eq:lunequalavail}) $L-c=s/q_2$, thus, the maximum value is
\begin{align}\label{eq:50}
(\tilde{p}_1-c)(1-q_2)+s
=(v-c)(1-q_2)\quad \text{from (\ref{eq:lmixed})}
\end{align}
Step ii.c.: Now, suppose that primary $1$ selects a price $x\in[L,\tilde{p}_2]$. Primary $2$ does not select a price in this interval if it selects $N$. Hence, at $x$, the expected payoff of primary $1$ is 
\begin{align}
& (x-c)(1-p_2q_2\psi_{2,Y}(x))=(x-c)(1-q_2)+(L-c)q_2\quad \text{from } (\ref{eq:psi2ylowc1})\nonumber\\
& <(\tilde{p}_1-c)(1-q_2)+(L-c)q_2\quad \text{as } \tilde{p}_1>\tilde{p}_2\nonumber\\
& =(v-c)(1-q_2)\quad \text{from (\ref{eq:lmixed}) and (\ref{eq:lunequalavail})}.
\end{align}
Hence, the maximum expected payoff that primary $1$ can attain under $N$ is $(v-c)(1-q_2)$ and this is attained at every price in the interval $[\tilde{p}_1,v]$. 

Step iii: The maximum expected payoff that primary $1$ can attain is $(v-c)(1-q_2)$ either under $Y$ or $N$. Hence, any randomization between $Y$ and $N$ will also yield the same expected payoff. The maximum expected payoff is attained by primary $1$ when it follows the strategy profile. Hence, primary $1$ does not have any profitable unilateral deviation. 

Case II: We,now, show that primary $2$ also does not have any profitable unilateral deviation. Toward this end, we first show that when primary $2$ selects $Y$ and the channel of primary $1$ is available, then the maximum expected payoff that it can get is $L-c-s$ (Step i). When primary $2$ selects $Y$ and the channel of the primary $1$ is unavailable, then the payoff of primary $2$ is $(v-c-s)$. Subsequently, we show that under $Y$, the maximum expected payoff attained by primary $2$ is $(v-c)(1-q_1)+s(q_1-q_2)/q_2$ (Step ii). Subsequently, we show that when primary $2$ selects $N$ then, the maximum expected payoff that primary $2$ can get is also $(v-c)(1-q_1)+s(q_1-q_2)/q_2$ (Step iii). Finally, we show that primary $2$ can attain the maximum expected payoff when it follows the strategy profile (Step iii).

Step i: Suppose primary $2$ selects $Y$ and the channel of primary $1$ is available.

Primary $1$ does not select a price from the interval $[L,\tilde{p}_2]$ when it selects $N$. Thus, at  price $x$   in the interval $[L,\tilde{p}_2]$, the expected payoff of primary $2$ is
\begin{align}
(x-c)(1-p_1\psi_{1,Y}(x))-s=L-c-s\quad \text{from } (\ref{eq:psi1ylowc1})
\end{align}
At price $x\in [\tilde{p}_2,\tilde{p}_1]$, the expected payoff of primary $2$ is 
\begin{align}
& (x-c)(1-p_1\psi_{1,Y}(x))-s=\nonumber\\
& (x-c)(1-1/q_1)+\dfrac{(v-c)(1-q_1)+s(q_1-q_2)/q_2}{q_1}-s \quad \text{from } (\ref{eq:psi1ylowc1}) \& (\ref{eq:barp_avail}).
\end{align}
Since the co-efficient of $x$ is negative,  the above is maximized at $x=\tilde{p}_2$. From (\ref{eq:lunequalavail}) note that $(v-c)(1-q_1)+s(q_1-q_2)/q_2=(\tilde{p}_2-c)(1-p_2q_1)$. Thus, the expected payoff of primary $2$ is upper bounded by 
\begin{align}\label{eq:51}
& (\tilde{p}_2-c)(1-1/q_1)+(\tilde{p}_2-c)(1-p_2q_1)/q_1-s =(\tilde{p}_2-c)(1-p_2)-s=L-c-s\quad \text{from (\ref{eq:tildep2avail})}
\end{align}
Now, suppose that primary $2$ selects a price $x\in [\tilde{p}_1,v)$. At $x$, the expected payoff of primary $2$ is
\begin{align}
& (x-c)(1-p_1-(1-p_1)\psi_{1,N}(x))-s\nonumber\\&=(x-c)(1-1/q_1)+\dfrac{(v-c)(1-q_1)+s(q_1-q_2)/q_2}{q_1}-s\quad \text{from } (\ref{eq:psi1nlowc1}) \& (\ref{eq:barp_avail})\nonumber\\
& <(\tilde{p}_2-c)(1-1/q_1)+\dfrac{(v-c)(1-q_1)+s(q_1-q_2)/q_2}{q_1}-s \quad \text{since } \tilde{p}_2<\tilde{p}_1.
\end{align}
Note from (\ref{eq:lunequalavail}) that $(v-c)(1-q_1)+s(q_1-q_2)/q_2=(\tilde{p}_2-c)(1-p_2q_1)$, thus, the above can be written as
\begin{align}
& (\tilde{p}_2-c)(1-1/q_1)+(\tilde{p}_2-c)(1-p_2q_1)/q_1-s =L-c-s\quad \text{from (\ref{eq:51})}.
\end{align}
Since $\psi_{1,N}(\cdot)$ has a jump at $v$, thus, the expected payoff at $v$ is strictly lower compared to a price close to $v$. Thus, the expected payoff of primary $2$ at $v$ is strictly less than $L-c-s$. Similarly, a price which is strictly less than $L$ fetches a payoff of at most $L-c-s$ under $Y$.

Hence, when the channel of primary $1$ is available, then, under $Y$ the maximum expected payoff that primary $2$ can attain is $L-c-s$. It is attained at any price in the interval $[L,\tilde{p}_2]$.

Step ii: When the channel of primary $1$ is unavailable, then the payoff that primary $2$ attains is $(v-c-s)$. Hence, the maximum expected payoff of primary $2$ under $Y$ is 
\begin{align}
q_1(L-c-s)+(1-q_1)(v-c-s)=(v-c)(1-q_1)+(L-c)q_1-s\nonumber\\
=(v-c)(1-q_1)+q_1s/q_2-s=(v-c)(1-q_1)+s(q_1-q_2)/q_2\quad \text{since } L-c=s/q_2.\nonumber
\end{align}

Step iii: Now, we show that if primary $2$ selects $N$, then the maximum expected payoff attained by primary $2$ is $(v-c)(1-q_1)+s(q_1-q_2)/q_2$. Toward this end, we first show that the maximum expected payoff attained by primary $2$ at any price in the interval $[\tilde{p}_2,\tilde{p}_1]$ and $[\tilde{p}_1,v]$ is $(v-c)(1-q_1)+s(q_1-q_2)/q_2$ and the maximum expected payoff is attained at any price in the interval $[\tilde{p}_2,v)$ (Step iii.a. and Step iii.b. resp.). Subsequently, we show that if primary $2$ selects any price less than $\tilde{p}_2$, then the maximum expected payoff is at most $(v-c)(1-q_1)+s(q_1-q_2)/q_2$ (Step iii.c.). 

Step iii.a: Suppose that primary $2$ selects a price $x$ in the interval $[\tilde{p}_2,\tilde{p}_1]$. Primary $1$ selects a price less than or equal to $x\in [\tilde{p}_2,\tilde{p}_1]$ if the channel of primary $1$ is available, it selects $Y$ and a price which is less than or equal to $x$. Thus, primary $1$ selects a price less than or equal to $x$ w.p. $p_1q_1\psi_{1,Y}(x)$. Since $\psi_{1,Y}(\cdot)$ is continuous in $[\tilde{p}_2,\tilde{p}_1]$,  the expected payoff of primary $2$ at $x$ is
\begin{align}
(x-c)(1-p_1q_1\psi_{1,Y}(x))=(v-c)(1-q_1)+s(q_1-q_2)/q_2\quad \text{from } (\ref{eq:psi1ylowc1})\& (\ref{eq:barp_avail}).
\end{align}
Step iii.b: Now suppose primary $1$ selects a price $x$ from the interval $[\tilde{p}_1,v)$. Primary $1$ selects a price less than or equal to $x$ w.p. $p_1q_1+(1-p_1)q_1\psi_{1,N}(x)$. Thus, the expected payoff of primary $2$ at $x$ is 
\begin{align}
(x-c)(1-p_1q_1-(1-p_1)q_1\psi_{1,N}(x))=(v-c)(1-q_1)+s(q_1-q_2)/q_2\quad \text{from } (\ref{eq:psi1nlowc1}) \& (\ref{eq:barp_avail}).
\end{align} to primary $2$. 
The expected payoff at $v$ is strictly less than the expected payoff at a price just below $v$ since $\psi_{1,N}(\cdot)$ has a jump at $v$. Thus, the expected payoff at $v$ is strictly less than $(v-c)(1-q_1)+s(q_1-q_2)/q_2$. 

Step iii.c: Now suppose that primary $2$ selects a price  $x$ from the interval $[L,\tilde{p}_2]$. Primary $1$ selects a price in the interval only when it selects $Y$. Thus, the expected payoff of primary $2$ at $x$ is 
\begin{align}
(x-c)(1-p_1q_1\psi_{1,Y}(x))=(x-c)(1-q_1)+(L-c)q_1\quad \text{from } (\ref{eq:psi1ylowc1})
\end{align}
Since the co-efficient of $x$ is positive, the above is maximized at $x=\tilde{p}_2$. By (\ref{eq:tildep2avail}) $L-c=(1-p_2)(\tilde{p}_2-c)$. Hence, the maximum value is
\begin{align}
& (\tilde{p}_2-c)(1-q_1)+(\tilde{p}_2-c)(1-p_2)q_1=(\tilde{p}_2-c)(1-p_2q_1)\nonumber\\
&=(v-c)(1-q_1)+s(q_1-q_2)/q_2\quad \text{from (\ref{eq:lunequalavail})}.
\end{align}
On the other hand a price which is strictly less than $L$ fetches a payoff which is strictly less than the payoff at $L$ as $L$ is the lowest end-point of the support of primary $1$. Thus, under $N$, the maximum expected payoff attained by primary $2$ is $(v-c)(1-q_1)+s(q_1-q_2)/q_2$ and it is attained at any price in the interval $[\tilde{p}_2,v)$. 

Step iv: Hence, the maximum expected payoff attained by primary $2$ is $(v-c)(1-q_1)+s(q_1-q_2)/q_2$ and it is attained if primary $2$ follows the strategy profile. Thus, primary $2$ also does not have any profitable unilateral deviation. Hence, the result follows.\qed

\subsection{Numerical Results}
\begin{figure*}
\begin{minipage}{0.31\linewidth}
\includegraphics[width=50 mm,height=30mm]{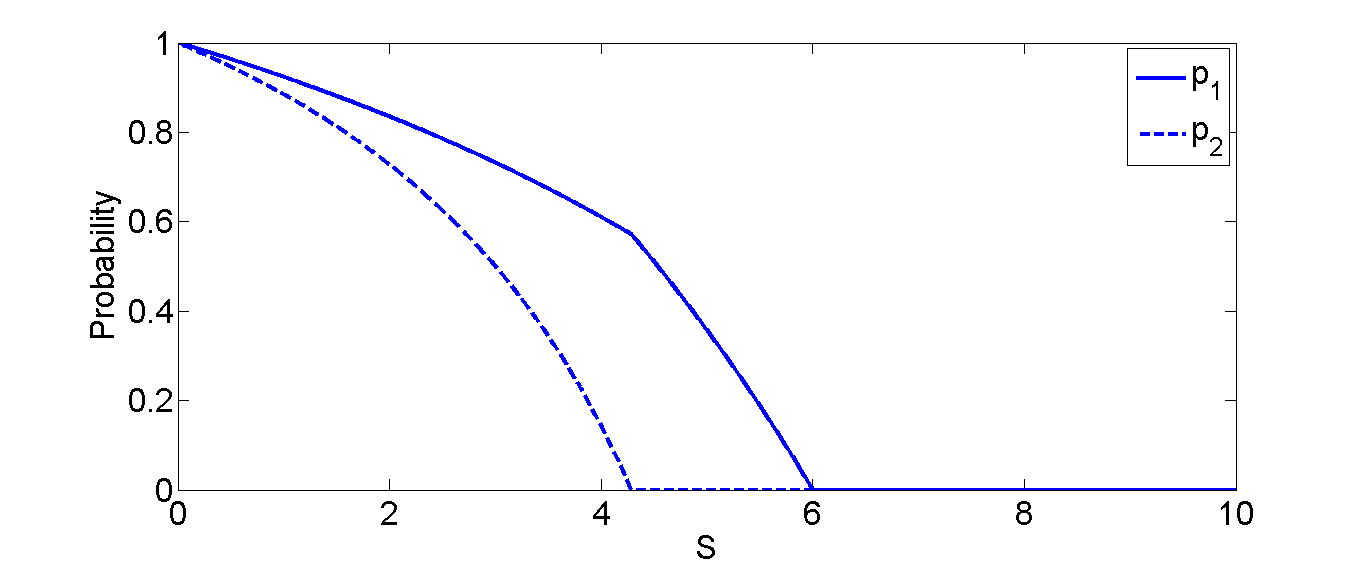}
\caption{\small Variation of $p_i$, $i=1,2$ with $S$ for an example setting: $v=25$, $c=0$, $q_1=0.7$, $q_2=0.4$.}
\label{fig:prob_unequal}
\end{minipage}\hfill
\begin{minipage}{0.31\linewidth}
\includegraphics[width=50mm,height=30mm]{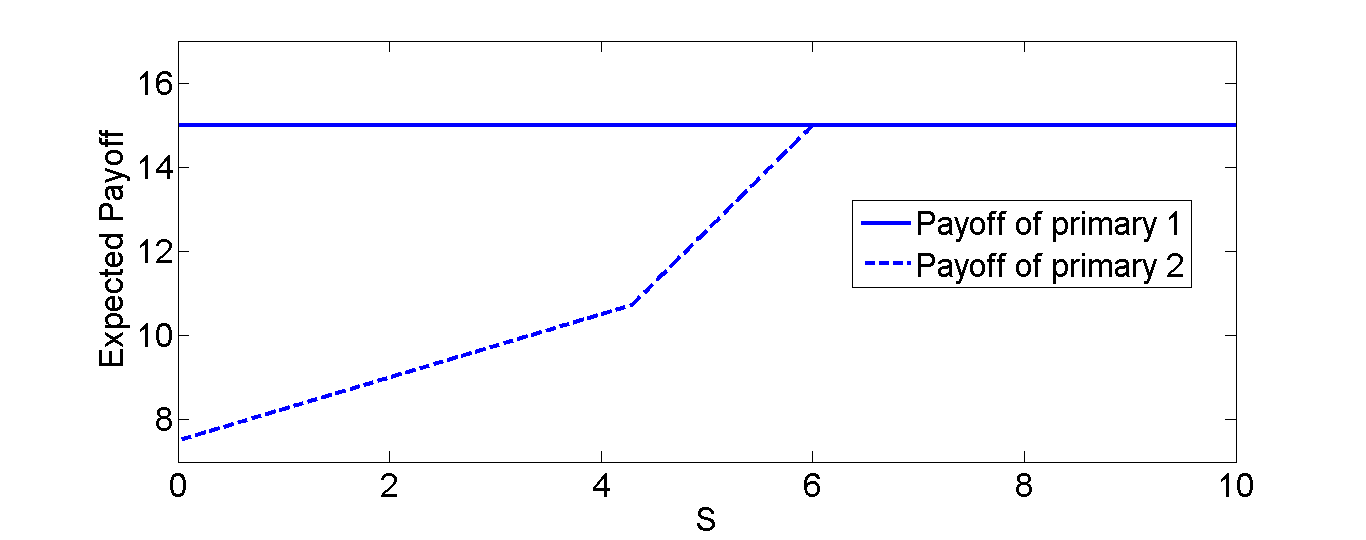}
\caption{\small Variation of the expected payoffs of primaries with $s$ in the same example setting considered in Fig.~\ref{fig:prob_unequal}. }
\label{fig:payoffunequalavail}
\end{minipage}\hfill
\begin{minipage}{0.31\linewidth}
\includegraphics[width=50mm,height=30mm]{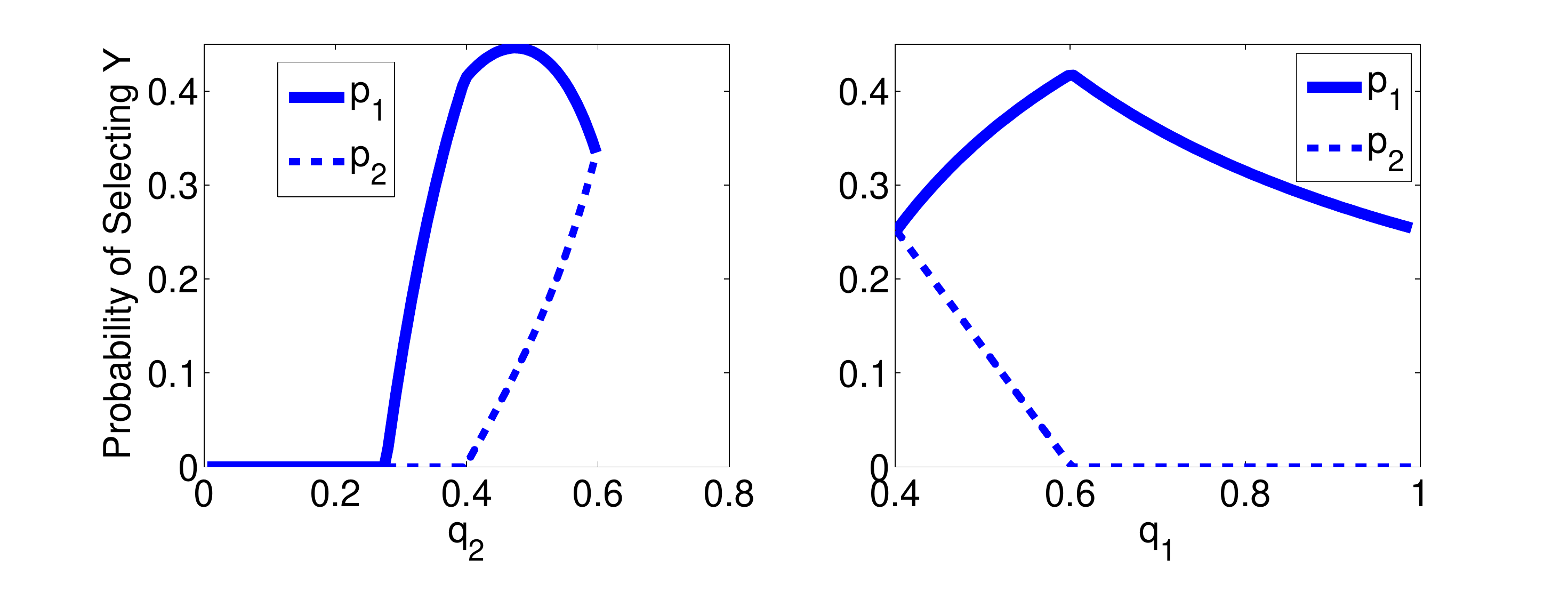}
\caption{Variation of $p_1$ and $p_2$ with $q_1$ and $q_2$. In the left hand figure we use $q_1=0.6, c=0, v=50, s=4$ and for the right hand figure we use $q_2=0.4,v=50, c=0, s=4$.}
\label{fig:p1p2vsq1q2}
\end{minipage}
\vspace{-0.3cm}
\end{figure*}
Fig.~\ref{fig:prob_unequal} shows the variation of $p_i$, the probability with which primary $i$ selects $Y$. As $S$ increases, $p_i$s decrease. Additionally, $p_1>p_2$ when $0<s<q_2(v-c)(1-q_2)$, when $s\geq q_2(v-c)(1-q_2)$ both the primaries select $N$ w.p. $1$ and thus, $p_i=0$. When $s=0$, then $p_i=1$. When $s\geq q_2(v-c)(1-q_1)/(1-q_1+q_2)$ $p_2$ is also $0$ but $p_1$ is positive. $p_1$ decreases at a slower rate compared to the $p_2$. The difference between $p_1$ and $p_2$ is maximum at $s=q_2(v-c)(1-q_1)/(1-q_1+q_2)$. When $s\geq q_2(v-c)(1-q_1)/(1-q_1+q_2)$, $p_1$ decreases at a faster rate. 

Fig.~\ref{fig:payoffunequalavail} shows the variation of the expected payoffs of the primaries with $s$, the cost of acquiring the CSI. The expected payoff of primary $1$ is independent of $s$. However, the expected payoff of primary $2$ decreases as $s$ when $s<q_2(v-c)(1-q_2)$. Thus, as $s$ decreases the payoff of primary $2$ decreases which contradicts the conventional wisdom which suggests that the payoff of a primary {\em should increase} as $s$ decreases. 

Fig.~\ref{fig:p1p2vsq1q2} shows the variations of $p_i, i=1,2$ with $q_i, i=1,2$. Note from the left hand figure of Fig.~\ref{fig:p1p2vsq1q2} that when $q_2$ is low, both the primaries select $N$ w.p. $1$, thus, $p_i=0$. When $q_2>0.25$, $p_1$ becomes positive, but $p_2$ is $0$. Due to high $q_1$, primary $1$ selects $Y$ with a higher probability and gains more compared to primary $2$.  When $q_2>0.4$, $p_2$ becomes positive and the difference between $p_1$ and $p_2$ decreases. As $q_2$ becomes close to $q_1$, primary $2$ also selects $Y$ with a higher probability.  Eventually, when $q_2\rightarrow q_1$, $p_2\rightarrow p_1$.

Note from the right hand figure of Fig.~\ref{fig:p1p2vsq1q2} that when $q_1=q_2$, $p_1=p_2$. As $q_1$ increases $p_2$ decreases and eventually it becomes $0$.  Note that $p_1$ initially increases with $q_1$. In this regime $p_2$ decreases, however, $q_1$ is not so high, thus, primary $1$ can gain more by selecting $Y$, hence, $p_1$ increases. However, eventually $p_2$ becomes $0$ and $q_1$ is high, thus, $p_1$ decreases. 

\section{Generalization}
\subsection{More than two primaries}
Consider that there are $n>2$ primaries. The availabilities are assumed to be the same. 
We briefly study two settings--
\subsubsection{Primary acquires the C-CSI of only one of the competitors}
In this scenario, we consider that the primary can acquire  the C-CSI of at most one of its competitors. If the primary decides to acquire the C-CSI, then it will acquire the C-CSI of a {\em randomly} selected competitor. First, we will describe a strategy and later {\em numerically} we show that such a strategy is asymptotically a NE as the number of primaries becomes large. The cost of acquiring the C-CSI(s) is assumed to be the same ($s$) for each primary. 

Strategy Profile ($SP_n$):
{\em When $s\geq T$, each primary selects $N$ w.p. $1$ and randomizes its price according to a distribution $\psi(\cdot)$ which is of the form $\phi_1(\cdot)$  (cf. (\ref{eq:class1})) from an interval $[L,v]$. When $s<T$, each primary selects $Y$ w.p. $p$ (and thus, $N$ w.p. $1-p$).  When a primary selects $Y$ and the channel state of the selected primary is $1$, then, it selects its price from the interval $[L_1,L_N]$ using $\psi_{1,Y}(\cdot)$; when the primary selects $N$, then it selects its price from the interval $[L_N,L_0]$ using $\psi_N(\cdot)$ and when the primary selects $Y$ and the channel state of the selected primary turns out to be $0$, then it selects its price from the interval $[L_0,v]$ using $\psi_{0,Y}(\cdot)$. $\psi_{1,Y}, \psi_{0,Y}$ and $\psi_N$ are of the form $\phi_1(\cdot)$ (cf.(\ref{eq:class1})).}

The full characterization of the strategy profiles, the threshold $T$ and the supports are given in \cite{jsac-tech16}. Note that the strategy $SP_n$ belongs to the class $[T,p]$ (Definition~\ref{defn:classtp}).  Note also that the primary selects lower (higher,resp.) prices when it finds that the channel state of the selected primary is $1$  ($0$, resp.) in the strategy profile $SP_n$.

\begin{figure*}
\begin{minipage}{0.49\linewidth}
\includegraphics[width=0.99\textwidth]{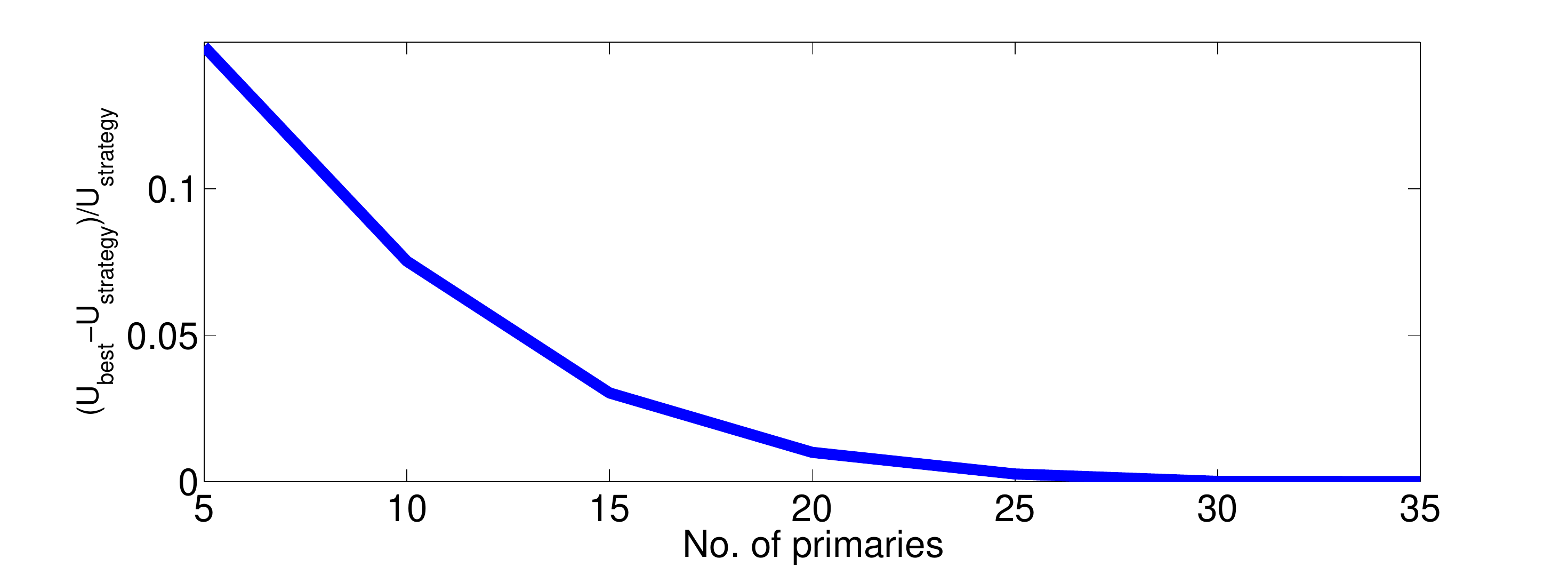}
\caption{\small The variation of relative gain with $n$ for $v=50, s=1.5, q=0.6$, no.of secondaries$=q(n-2)$.}
\label{fig:relativegain}
\end{minipage}\hfill
\begin{minipage}{0.49\linewidth}
\includegraphics[width=0.99\textwidth]{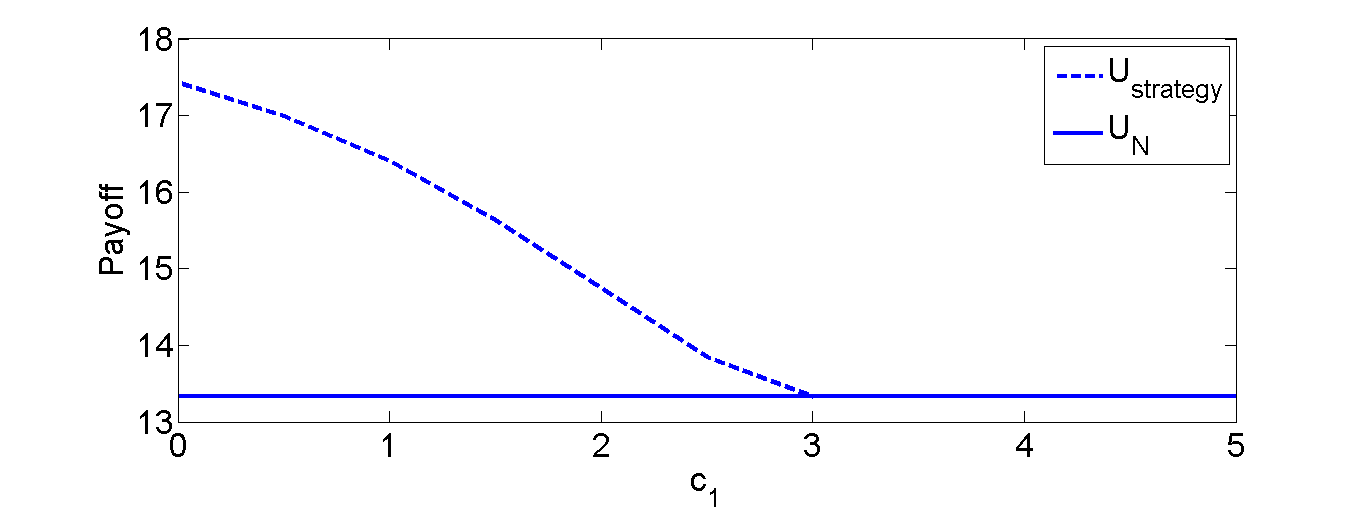}
\caption{\small Variation of the expected payoff under the strategy profile $SP_n$ as a function of $s$ when $v=50, c=0, n=10$, no. of secondaries$=6$. $U_N$ is the expected payoff in the setting where the primaries can not acquire the C-CSI. }
\label{fig:payoff}
\end{minipage}
\vspace{-0.2in}
\end{figure*}
We show the variation of relative gain i.e. $\dfrac{U_{best}-U_{strategy}}{U_{strategy}}$ as $n$ increases in Fig.~\ref{fig:relativegain}. $U_{best}$ is the best possible expected payoff that a primary can attain. $U_{strategy}$ is the expected payoff that a primary attains when it follows the strategy $SP_n$. Fig.~\ref{fig:relativegain} shows that as $n$ increases, the relative gain decreases exponentially and becomes close to $0$ when $n\geq 25$. The relative gain is small ($0.1$) even when $n=10$, thus, the strategy profile $SP_n$ closely approximates the NE strategy even for small values of $n$  

 Fig.~\ref{fig:payoff} shows the variation of expected payoff as $s$ increases under the strategy profile $SP_n$.  When $s$ is high, the primaries do not acquire the C-CSI, thus their expected payoffs are equal to the payoffs they would get in the setting where the primaries can not acquire the C-CSI ($U_N$). The expected payoff of the primary increases as $s$ decreases. {\em Note that we observe the similar behavior when there is an estimation error} (Theorems 5 and 6).  However, such a similarity is not surprising. When a primary only acquires the C-CSI of a randomly selected competitor, it is still not aware of the channel states of {\em all the competitors}. It only has an estimate of how many competitors will be there. Simlarly, when there is an error in estimating the C-CSI, the primary is still not certain whether the competitor is available when there are two primaries. Thus, the expected payoff shows similarity with the setting with two primaries and a non-zero estimation error. 

\subsubsection{Primary acquires the C-CSI of all the competitors}
Now, we consider other extreme where if a primary decides to acquire the C-CSI, it will aquire the C-CSI of all of its competitors. If a primary selects $N$, then it will not acquire the C-CSI of any of its competitors. The C-CSI acquisition costs and the availability probabilities are the same. We show in Appendix~\ref{sec:n-1all} that
 \begin{theorem}\label{thm:payoffconstant}
 In a symmetric NE, the expected payoff of primaries are independent of the C-CSI acquisition costs.
 \end{theorem}
Note that, we have seen the similar results when there are two primaries (Theorems~\ref{thm:nandn} and ~\ref{thm:mixedstrategy}). Hence, the result where the expected payoff of each primary is independent of the C-CSI acquisition cost extends to a more general sceanrio where there are more than two primaries. 
\subsection{Relaxation of $0-1$ assumption of the channel state}
Suppose that the available channel is one of the two states high qualiy (state $2$) and low quality state (state $1$). Thus, the states of the channel of each primary is $0$ (unavailable), $1$ (available, but low quality) and $2$ (available and high quality). We assume that there are two primaries and the C-CSI acquisition costs are the same. The channel of each primary is in state $i$, $i=0,1,2$ w.p. $q_i$ independent of the channel states of other primaries. All these are common information. 

Since the available channel state can be in more than one state,  we define a penalty function to capture the preference order of the secondaries among available channels. The secondary will prefer a channel with a lower penalty. Higher prices and lower qualities induce higher penalties. Hence, we consider at state $i$ the penalty function $g_i$ for price $x$ is $g_i(x)=x-h(i)$ where $h(\cdot)$ is a strictly increasing function. Note that negative of $g_i(\cdot)$ can be considered to be a utility function, and quasi-linear utility functions similar to our approach is a standard assumption in literature.  We have already characterized the NE strategy for such penalty functions when the primary can not acquire the C-CSI of its competitors \cite{isit, arnob_ton}.  Our result shows that--
\begin{lem}
In the complete information game, i.e.,
 when the primaries know the C-CSI of each others, then the expected payoff is exactly same as the expected payoff attained by the primary in the scneario where primaries can not acquire the C-CSI of their competitotrs. 
\end{lem}
\textit{Outline of the Proof}
First, we determine the expected payoff in the complete information game. 

If the channel states of both the primaries are the same, it becomes the standard{\em Bertrand Competition}. Hence, the payoffs of both the players will be $0$. 

Now, suppose that one primary's channel state is $2$ and the other hass $1$. Without loss of generality, assume that the channel state of primary $1$ is $2$ and that of primary $2$ is $1$. This is similar to the {\em Bertrand-Edgworth competition} and the payoff of the primary $1$  (which has channel state $2$ ) is $h(2)$ and that of primary $2$ (which has channel state $1$) is $0$.

When only one primary is available it will select the price $v+h(i)$ at channel state $i$( $i\geq1$) and attains the monopoly profit $v-c+h(i)$. The probability that a primary is in state $0$ is $1-q_1-q_2$.

Thus, if the primary's channel state is $2$, then its expected payoff is
\begin{align}
(v+h(2)-c)(1-q_1-q_2)+(h(2)-h(1))q_1=(v-c)(1-q_1-q_2)+h(2)(1-q_2)-h(1)q_1.
\end{align}
When the primary is in state $1$, it can only gain positive payoff only when the competitor's channel is in state $0$. Hence, its expected payoff is
\begin{align}
(v+h(1)-c)(1-q_1-q_2).
\end{align}

Now, in the incomplete information game, from \cite{arnob_ton,isit,archivedreport-part2} at channel state $1$, the expected payoff of a primary is $(v+h(1)-c)(1-q_1-q_2)$. The expected payoff of primary $2$ is $(L-c)(1-q_2)$, where $L$ is
\begin{align}
(L-c+h(1))(1-q_2)=(v-c+h(1))(1-q_1-q_2).
\end{align}
Thus,
\begin{align}
(L-c+h(2))(1-q_2)=(v-c)(1-q_1-q_2)+h(2)(1-q_2)-h_1q_1
\end{align}which is also the expected payoff of a primay whose channel state is $2$ in the complete information game. Hence, the result follows.\qed

The above proof also shows that a primary never selects $Y$ w.p. $1$ for any positive cost $s>0$,the result that we observe throughout the paper in each setting. Though the above result states that in the two extremes the expected payoff of a primary is the same, we conjecture that the expected payoff of a primary would be independent of the C-CSI acquisition similar to the basic model which we have studied.


\appendix
\subsection{Uniqueness Proofs of Theorems~\ref{thm:nandn} and \ref{thm:mixedstrategy}}
Here, we show that  there can not be any other NE strategy profile apart from those described in Theorems~\ref{thm:nandn} and \ref{thm:mixedstrategy} in the basic model. Note that when $s\geq q(v-c)(1-q)$, then the NE strategy profile is the one described in Theorem~\ref{thm:nandn} and when $s<q(v-c)(1-q)$, the NE strategy profile is the one described in Theorem~\ref{thm:mixedstrategy}.

\subsubsection{Structure of the Pricing strategies}
We first investigate the key  structure of the NE pricing strategies (if it exists). 

Note that under $Y$, if a primary knows that its competitor\rq{}s channel is not available then it will choose $v$ w.p. $1$. We thus, investigate the structure of $F_1(\cdot)$ and $F(\cdot)$ in an NE strategy. Recall that $F_1(\cdot)$ is the pricing distribution that a primary chooses when it selects $Y$ and knows that the channel of its competitor is available, while $F(\cdot)$ is the pricing distribution that a primary chooses when it selects $N$. 
\begin{theorem}\label{thm:jump}
In an NE strategy profile, neither $F(\cdot)$ nor $F_1(\cdot)$ can have a jump at any price which is less than $v$. Additionally, $F_1(\cdot)$ can not have a jump at $v$. 
\end{theorem}
\begin{proof}
First, we show that neither $F(\cdot)$ nor $F_1(\cdot)$ can have a jump at any price which is less than $v$. Subsequently, we show that $F_1(\cdot)$ can not have a jump at $v$.

Note that a primary can only have a jump at a price if it is a best response. First, note that $F(\cdot)$ can not have a jump at a price less than or equal to $c$. This is because at a price less than or equal to $c$ will fetch a negative profit, however, if the primary selects $v$, then it will get an expected payoff of $(v-c)(1-q)$. 

Similarly, if $F_1(\cdot)$ has a jump at a price less than or equal to $c$, then its payoff under $F_1(\cdot)$ is at most $(c-c)-s=-s$. Note that when the channel of the competitor is unavailable, then the primary will attain the payoff of $(v-c)-s$. Hence, the expected payoff under $Y$ is thus, $(v-c-s)(1-q)-sq=(v-c)(1-q)-s$. However, if the primary selects $N$ and the price $v$ which will fetch an expected profit of $(v-c)(1-q)$.

Now if either $F_1(\cdot)$ or $F(\cdot)$ has a jump at $c<x<v$, then the other primary can select a price $x-\epsilon$ and still can gain higher payoff compared to $x$. Thus, the other primary will not select any price in the interval $(x-\epsilon,x+\epsilon)$ as it will get a strictly higher payoff at $x-\epsilon$ compared to any price in the interval. Hence, the primary itself can gain strictly  higher payoff by selecting a price at $y\in (x,x+\epsilon)$ compared to $x$. It contradicts the fact that either $F_1(\cdot)$ or $F(\cdot)$ will have a jump at $x<v$. 

Next, we show that $F_1(\cdot)$ can not have a jump at $v$. Suppose $F_1(\cdot)$ has a jump at $v$, then the other primary will never select $v$ with positive probability when its channel is available as it can get strictly higher payoff by selecting a price slightly less than $v$. Thus, at $v$, the primary is never going to sell its channel when the channel of other primary is available. Thus, the expected payoff that the primary will get under $F_1(\cdot)$ is $-s$. Thus, under $Y$, the expected payoff that the primary will attain is $(v-c)(1-q)-s$. Again, the primary will have an incentive to deviate to select $N$ and select the price $v$ which will fetch a payoff of at least $(v-c)(1-q)$.
\end{proof}
The above theorem shows that if the channel of a primary is available then it can not have a jump at any price other than $v$. 

Now, we show an important property of $F_1(\cdot)$ and $F(\cdot)$ when a primary randomizes between $Y$ and $N$ in an NE strategy. 
\begin{theorem}\label{thm:disjoint}
Suppose that primary $1$ selects $Y$ w.p. $p$ and $N$ w.p. $1-p$ in an NE.  Then, the upper end point of the support set of $F_1(\cdot)$ must be lower than or equal to to the lower end-point of the support set of $F(\cdot)$. 
\end{theorem}
\begin{proof}
Note from Theorem~\ref{thm:jump} that $F_1(\cdot)$ can not have a jump at $v$. Thus, the lower end point of $F_1(\cdot)$ can never be $v$. If the lower end-point of the support set of $F(\cdot)$ is $v$, then the statement is trivially true. So, we consider the setting where the lower end-point of the support set of  $F(\cdot)$ is less than $v$. Suppose the statement is false. Thus, there must exist a $x<y<v$ such that $x$ is in the support set of $F(\cdot)$ and $y$ is in the support set of $F_1(\cdot)$. Now, suppose that the maximum expected payoff of  primary $1$ when it selects $F_1(\cdot)$ under $Y$ is $\bar{p}_1$.  Also let $\bar{p}_2$ be the maximum expected payoff primary $1$ gets  when it selects $F(\cdot)$ under $N$. 

Since $x<v$, thus, if the channel of competitor is available, it can not have any jump at $x$. Hence, while choosing $N$, the probability of winning at $x$ is $(1-q\phi_2(x))$ where $\phi_2(\cdot)$ is the probability that the primary $2$ will select a price less than or equal to $x$ when its channel is available. Since $x<v$ and primary $2$ does not have a jump at $x$, thus, $x$  is a best response to primary $1$ under $N$. Thus,
\begin{align}\label{eq:payoffn}
(x-c)(1-q\phi_2(x))=\bar{p}_2
\end{align}
Since $\bar{p}_1$ is the maximum expected payoff that primary $1$ gets under $F_1(\cdot)$, thus, if primary $1$ selects $x$ under $F_1(\cdot)$, then its payoff would be
\begin{align}\label{eq:ratio}
 (x-c)(1-\phi_2(x))\leq\bar{p}_1\nonumber\\
 \dfrac{1-\phi_2(x)}{1-q\phi_2(x)}\leq \dfrac{\bar{p}_1}{\bar{p}_2}\quad\text{from (\ref{eq:payoffn}) }
\end{align}
Similarly, since $y<v$, thus, primary $2$ will not have a jump at $y$ when its channel is available. Thus, primary $1$\rq{}s expected payoff under $F_1(\cdot)$ at the price $y$ is  
\begin{align}\label{eq:disjoint1}
(y-c)(1-\phi_2(y))=\bar{p}_1
\end{align}
If primary $1$ selects $N$ and the price $y$, then its expected payoff is
\begin{align}\label{eq:disjoint}
(y-c)(1-q\phi_2(y))& =\bar{p}_1\dfrac{1-q\phi_2(y)}{1-\phi_2(y)}\quad \text{from (\ref{eq:disjoint1}) }\nonumber\\
& \geq \dfrac{(1-q\phi_2(y))(1-\phi_2(x))}{(1-\phi_2(y))(1-q\phi_2(x))} \bar{p}_2\text{from (\ref{eq:ratio}) }
\end{align}
Now, note that $\phi_2(y)\geq \phi_2(x)$ as $y>x$. If $\phi_2(y)=\phi_2(x)$, then the expected payoff at $y$ must be greater than the expected payoff at $x$, hence, $x$ can not be a best response at $N$ for primary $1$. However,  if $\phi_2(y)>\phi_2(x)$, then the expected payoff at $y$ at $N$ is strictly higher than $\bar{p}_2$ by (\ref{eq:disjoint}).  Thus, this leads to a contradiction since $\bar{p}_2$ is the maximum expected payoff at $N$. Hence, the result follows.
\end{proof}
Now, we show that both $F(\cdot)$ and $F_1(\cdot)$ are contiguous. Additionally, if a primary randomizes between $Y$ and $N$, then there is no \lq\lq{}gap\rq\rq{} between $F(\cdot)$ and $F_1(\cdot)$. 
\begin{theorem}\label{thm:no_gap}
(i) In a NE strategy if a primary selects $Y$ w.p. $1$, and it selects $F_1(\cdot)$ when it knows that the channel of other primary is available, then $F_1(\cdot)$ must be contiguous and the upper end-point of $F_1(\cdot)$ must be $v$.\\
(ii) In a NE strategy if a primary selects $N$ w.p. $1$, and if it selects $F(\cdot)$ when it knows that channel of other primary is available, then $F(\cdot)$ must be contiguous and the upper end-point of $F(\cdot)$ must be $v$.\\
(iii) In a NE strategy if the primary randomizes between $Y$ and $N$, both $F_1(\cdot)$ and $F(\cdot)$ must be contiguous, there must not be any gap between the support sets of $F_1(\cdot)$ and $F(\cdot)$. Moreover, the upper-end point of $F(\cdot)$ must be $v$. 
\end{theorem}
\begin{proof}
We only show the proof of part (i). The proof of the other parts will be similar. 

{\em Part (i)}:
Suppose that primary $1$ selects $F_1(\cdot)$ such that $F_1(x)=F_1(y)$ for some $v\geq y>x$ such that both $y,x$ are under the support set of $F_1(\cdot)$.  Since  $x<v$  thus, primary $2$ does not have a jump at $x$ when its channel is available. Hence, $x$ is a best response for primary $1$ under $F_1(\cdot)$. By Theorem~\ref{thm:disjoint} if a primary randomizes between $Y$ and $N$, then the lower end-point of $F(\cdot)$ must be greater than or equal to the lower end-point of $F_1(\cdot)$. Thus, $F(x)=F(y)=0$. Thus, primary $2$ will attain a strictly higher payoff at any value $z\in (x,y)$ compared to at $x$. Thus, there is an $\epsilon>0$ where primary $2$ will never select any price in the interval $[x,x+\epsilon]$, hence, $x$ itself is not a best response for primary $1$. But the above contradicts the fact that $x$ is in the support set of $F_1(\cdot)$. Hence, the result follows.
\end{proof}
\subsubsection{Special Property where primaries randomize between $Y$ and $N$}
 Next theorem shows that in an NE if both the primaries randomize between $Y$ and $N$. Then both of them should put the same probability mass on $Y$ (and $N$, resp.). 
\begin{theorem}
Suppose primary $1$ selects $Y$ w.p. $1>p_1>0$ and $N$ w.p. $1-p_1$. Primary $2$ selects $Y$ w.p. $1>p_2>1$ and $N$ w.p. $1-p_2$. Then, $p_1=p_2$ in an NE strategy profile.
\end{theorem}
\begin{proof}
Suppose that at $Y$, primary $1$ ($2$, resp.) selects a price using the distribution $F_1(\cdot)$ ($\bar{F}_1(\cdot)$, resp.)  when it knows that the channel of primary $2$ ($1$, resp.)  is available for sale. At $N$, suppose that primary $1$ ($2$, resp.) selects a price using the distribution $F(\cdot)$ ($\bar{F}(\cdot)$, resp.). 

Let $L_1$ ($\bar{L}_1$, resp.) and $U_1$ ($\bar{U}_1$, resp.) be respectively the lower and upper end-points of the support of $F_1$ ($\bar{F}_1$, resp.). Let $L$ ($\bar{L}$, resp.) and $U$ ($\bar{U}$, resp.) be the lower and upper end-point of the support of $F(\cdot)$ ($\bar{F}$, resp.) respectively. By Theorem~\ref{thm:disjoint} $L_1<L$ and $\bar{L}_1<\bar{L}$. Note also from Theorem~\ref{thm:no_gap} that $U_1=L$ and $\bar{U}_1=\bar{L}$. 

 First, we show that $L_1=\bar{L}_1$. Suppose not. Without loss of generality assume that $L_1<\bar{L}_1$. Thus, primary $2$ does not select any price in the interval $(L_1,\bar{L}_1)$. Thus, the primary $1$ will get a strictly higher payoff at $\bar{L}_1-\epsilon$ for some $\epsilon>0$ compared to $L_1$. Hence,  primary $1$ must select prices close to $L_1$ with probability $0$ which contradicts that  $L_1$ is the lower end-point of $F_1$. Thus, $L_1=\bar{L}_1$.  

 By Theorem~\ref{thm:jump} $L_1$ can not be equal to $v$. Thus, $L_1=\bar{L}_1<v$. Thus,  both $L_1$ and $\bar{L}_1$ are best responses to primary $1$ and primary $2$ respectively at $Y$. Since $L_1=\bar{L}_1$, thus, the expected payoff at $Y$ must be the same for both players. Also note that since primaries randomize between $Y$ and $N$, thus, the payoffs at $Y$ and $N$ must be the same. Hence, the expected payoff of the primaries at $N$ also must be the same. Thus, no primary can have a jump at $v$ under $N$. This is because if a primary has a jump at $N$, then the other primary would get a strictly higher payoff at a price just below $v$ which contradicts that both the primaries must have the same payoff under $N$. Thus, $L, \bar{L}<v$. 

Now, we show that $L=\bar{L}$, towards this end, we introduce few more notations. Let $\bar{p}_1-c$ be the maximum expected payoff of primary $1$ ($2$, resp.) under $F_1(\cdot)$ ($\bar{F}_1(\cdot)$, resp.)   and $\bar{p}_2-c$ be the expected payoff of primary $1$ ($2$, resp.) under $F(\cdot)$ ($\bar{F}(\cdot)$, resp.). 

Suppose $L\neq \bar{L}$. Without loss of generality assume that $L>\bar{L}$.  Thus, $\bar{L}<v$. Since $\bar{L}$ is the upper end-point of $\bar{F}_1(\cdot)$ and $\bar{L}<v$, thus, the expected payoff of primary $2$ at $\bar{L}$ under $\bar{F}_1(\cdot)$ is $\bar{p}_1-c$.   Thus, 
\begin{align}
(\bar{L}-c)(1-p_1F_1(\bar{L}))=\bar{p}_1-c\label{eq:pay21}
\end{align}
$\bar{L}$ is also a best response of primary $2$ at $N$, thus,
\begin{align}\label{eq:pay22barl}
(\bar{L}-c)(1-qp_1F_1(\bar{L}))=\bar{p}_2-c
\end{align}
Since $v>L>\bar{L}$ and $L$ is the upper end-point of $F_1(\cdot)$, thus, $L$ is also a best response of primary $1$ under $Y$. 
\begin{align}\label{eq:pr1}
(L-c)(1-p_2-(1-p_2)\bar{F}(L))=\bar{p}_1-c
 \end{align}
 Since $L$ is the lower end point of $F(\cdot)$, thus, under $N$, the expected payoff of primary $1$ at $L$ is
 \begin{align}\label{eq:pay12}
 (L-c)(1-qp_2-q(1-p_2)\bar{F}(L))=\bar{p}_2-c
 \end{align} 
 Also note that  since $L>\bar{L}$, thus, $L$ is in the support of $\bar{F}(\cdot)$, thus, under $N$, the expected payoff to primary $2$ at $L$ is 
 \begin{align}\label{eq:pay22}
 (L-c)(1-qp_1)=\bar{p}_2-c
 \end{align}
 as $F_1(L)=1$ and $F(L)=0$.
 
Thus, from  (\ref{eq:pay22}) and (\ref{eq:pay12}) $p_1=p_2+(1-p_2)\bar{F}(L)$. Now, the expected payoff of primary $2$ at $L$ when it selects $Y$ and the channel of primary $1$ is available, is
\begin{align}\label{eq:pay21l}
(L-c)(1-p_1)& =(L-c)(1-p_2-(1-p_2)\bar{F}(L))\nonumber\\
& =\bar{p}_1-c \quad from (\ref{eq:pr1})
\end{align}
  Hence, from (\ref{eq:pay21}), (\ref{eq:pay21l}), (\ref{eq:pay22barl}) and (\ref{eq:pay22}) that
\begin{align}
\dfrac{1-p_1F_1(\bar{L})}{1-p_1}=\dfrac{1-qp_1F_1(\bar{L})}{1-qp_1}
\end{align}
which leads to a contradiction as neither $q$ is not equal to $1$ nor $F_1(\bar{L})= 1$. Hence, we must have $L=\bar{L}$. 

Now, at $L$, the expected payoff of primary $2$ at $Y$ is $(L-c)(1-p_1)=\bar{p}_1-c$. Similarly, at $\bar{L}$, the expected payoff of primary $1$ at $Y$  is $(\bar{L}-c)(1-p_2)=\bar{p}_1-c$. Since $L=\bar{L}$, thus, we must have $p_1=p_2$. Hence, the result follows. 
\end{proof}
Next, we determine the probability with which the primaries must randomize between $Y$ and $N$ in an NE strategy. 
\begin{obs}
If both the primaries randomize between $Y$ and $N$, they should do it w.p. $p$ where $p=\dfrac{q(v-c)(1-q)-s}{q(v-c)(1-q)-sq}$.
\end{obs}
\begin{proof}
Suppose that a primary selects its price from $F_1(\cdot)$ under $Y$ and when it knows that the channel of other primary is available. Suppose that under $F_1(\cdot)$ the expected payoff is $\tilde{p}_1-c$. Thus, the expected payoff of primary $1$ under $Y$ is
\begin{align}
(v-c)(1-q)+q(\tilde{p}_1-c)-s
\end{align}
Suppose that the primary selects its price from $F(\cdot)$ under $N$. Since no primary has any jump at $v$ when both the primaries randomize between $Y$ and $N$ and $v$ is the upper end-point of $F(\cdot)$ by Theorem~\ref{thm:no_gap}, thus, the expected payoff under $N$ is $(v-c)(1-q)$. Since the primary randomizes between $Y$ and $N$, thus, the expected payoff under $Y$ and under $N$ must be the same.  Hence, we must have $s=q(\tilde{p}_1-c)$. 

Suppose $L$ be the upper end point of the support of $F_1(\cdot)$ (and thus, also the lower endpoint of $F(\cdot)$). Hence, the expected payoff at $L$ is 
\begin{align}
(L-c)(1-qp)=(v-c)(1-q)
\end{align}
Thus, $L=(v-c)(1-q)/(1-qp)+c$. Also note that $L$ is also a best response at $F_1(\cdot)$. Thus, 
\begin{align}
(L-c)(1-p)& =\dfrac{s}{q}\nonumber\\
\dfrac{(v-c)(1-q)(1-p)}{1-qp}& =\dfrac{s}{q}
\end{align}
Obtaining $p$ from the above expression will give the desired result.
\end{proof}
\subsubsection{Does there exists an NE where one player selects $Y$ w.p. $1$?}
\begin{theorem}
There is no NE where a primary selects $Y$ w.p. $1$ and the other primary selects $Y$ w.p. $p$ and $N$ w.p. $1-p$.
\end{theorem}
\begin{proof}
Without loss of generality, assume that primary $1$ selects $Y$ w.p. $1$ and the  primary $2$ selects $Y$ w.p. $p$ and $N$ w.p. $1-p$. 

Now suppose that primary $1$ selects a price using the distribution function $F_1(\cdot)$ when it knows that the channel of its competitor is available for sale. Let at $Y$, primary $2$ selects a price using distribution function $F_2(\cdot)$ when it knows that the channel of its competitor is available for sale, and at $N$, it selects a price using distribution function $\bar{F}_2(\cdot)$. 

Let $L_1$ be the lower end-point of the support of $F_1(\cdot)$ and $L_2$ ($\bar{L}_2$, resp.) be the lower end-point of $F_2(\cdot)$ ($\bar{F}_2$, resp.).

Note from Theorem~\ref{thm:disjoint}  that $\bar{L}_2>L_2$. Now, we show that $L_1=L_2$. Suppose that $L_1>L_2$, then, primary $2$ can attain strictly higher payoff at any price close to $L_1$ compared to at $L_2$ which shows that $L_2$ can not be a lower end-point of $F_2$. By symmetry, it also follows that $L_1$ can not be less than $L_2$, hence $L_1=L_2$. Thus, the expected payoff at $Y$ must be equal for both the primaries.

Now, note that $F_1(\cdot)$ can not have a jump at $v$ by Theorem~\ref{thm:jump}. Note that the upper end-point of $\bar{F}_2(\cdot)$ is $v$ by Theorem~\ref{thm:no_gap}. Since $F_1(\cdot)$ does not have a jump at $v$, thus, $v$ is a best response of primary $2$ under $N$. Thus, the expected payoff of primary $2$ under $N$ is $(v-c)(1-q)$.  Since primary $2$ randomizes between $N$ and $Y$, thus, the expected payoff of primary $2$ is $(v-c)(1-q)$ under $Y$. Thus, the expected payoff of primary $1$ is also $(v-c)(1-q)$. 

At any $x\in [\bar{L}_2,v)$ primary $2$ does not have any jump, thus, $x$ is a best response for primary $1$. Thus, at any $x\in [\bar{L}_2,v)$ the expected payoff of primary $1$ is
\begin{align}
(x-c)(1-p-(1-p)\bar{F}_2(x))=\tilde{p}_1-c\nonumber\\
\bar{F}_2(x)=\dfrac{1}{1-p}(1-p-\dfrac{1}{1-p}\dfrac{\tilde{p}_1-c}{x-c})
\end{align}
Note that under $Y$, the expected payoff of primary $1$ is $(v-c)(1-q)+q(\tilde{p}_1-c)-s$.

Now, if primary $1$ selects $N$ and a price $x\in [\bar{L}_2,v)$, then its expected payoff is 
\begin{align}
& (x-c)(1-qp-q(1-p)\bar{F}_2(x)) \nonumber\\
& =(x-c)(1-qp-q(1-p-\dfrac{\tilde{p}_1-c}{x-c}))\nonumber\\
&=(x-c)(1-q)+(\tilde{p}_1-c)q
\end{align}
Thus, for any small enough $\epsilon>0$, we have $(v-c-\epsilon)(1-q)+(\tilde{p}_1-c)q>(v-c)(1-q)(1-q)+q(\tilde{p}_1-c)-s$> Hence, primary $1$ has profitable unilateral deviation. Hence, such a strategy profile can never be an NE. 
\end{proof}
Note that we have already ruled out the possibility of the NE strategy profile where a primary selects $Y$ w.p. $1$ and the other selects either $N$ or $Y$ w.p. $1$. Hence, there is no NE where a primary selects $Y$ w.p. $1$.
\subsubsection{Does there exist a NE where one player selects $N$ w.p. $1$?}
\begin{theorem}
There is no NE where a primary selects $N$ w.p. $1$ and the other primary randomizes between $Y$ and $N$. 
\end{theorem}
\begin{proof}
Without loss of generality assume that primary $1$ selects $N$ w.p. $1$ and primary $2$ randomizes between $Y$ and $N$. 

Suppose that primary 1 selects its price using $F(\cdot)$. Let $L$ be the lower end-point of the support of $F(\cdot)$. Let $\tilde{p}_1-c$ be the expected payoff of primary $1$. Let primary $2$ selects $F_2(\cdot)$ when it selects $Y$ and it knows that the channel of primary $1$ is available. Let  $L_2$ be the lower end-point of $F_2(\cdot)$. 
First, note that $L_1$ must be equal to the $L_2$.  Since $L_2<v$ by Theorem~\ref{thm:no_gap} and $L_1=L_2$, thus, $L_2$ is a best response for both primary $1$ and $2$. The expected payoff of primary $2$ under $Y$ when the channel of primary $1$ is  $L_2-c-s$. Similarly, the expected payoff of primary $1$ is $L_2-c$.  Thus, $\tilde{p}_1-c=L_2-c$. Expected payoff of primary $2$ under $Y$ is, $q(L_2-c)+(v-c)(1-q)-s$.

Also let $L$ be the lower end point of $\bar{F}_2$ where $\bar{F}_2$ be the pricing strategy that primary $2$ uses when it selects $N$. From Theorem~\ref{thm:no_gap} the upper end-point of the support of $F_2(\cdot)$ is also $L$. From Theorem\ref{thm:no_gap} also note that the upper end-point of $\bar{F}_2(\cdot)$ is $v$. 

First, note that under $N$ the expected payoff of primary $2$ must be at least $(v-c)(1-q)$ as this is the payoff that primary $2$ can at least get when it selects $v$. Now, we show that under $N$, the expected payoff of primary $2$ must be  equal to $(v-c)(1-q)$. Suppose not, i.e. primary $2$ attains an expected payoff of larger than $(v-c)(1-q)$. Since the upper end-point of $\bar{F}_2$ is $v$, thus, primary $1$ must have a jump at  $v$. Since primary $1$ has a jump at $v$, thus, $v$ is a best response for primary $1$. Thus, primary $1$ attains an expected payoff of $(v-c)(1-q)$ under $N$. Thus, $\tilde{p}_1-c=(v-c)(1-q)$. Since primary $2$ is randomizing between $Y$ and $N$, thus, the primary $2$\rq{}s expected payoff is also greater than $(v-c)(1-q)$ when it selects $Y$. Thus,  if the primary $1$ select $Y$ and price $L_2$ w.p. $1$ when the channel of primary $2$ is available and $v$ w.p. $1$ otherwise; then it will also get an expected payoff of $q(L_2-c)+(v-c)(1-q)-s$ which is higher compared to $(v-c)(1-q)$. Hence, this is not possible.

Thus, the expected payoff of primary $2$ must be equal to $(v-c)(1-q)$. Since primary $1$ gets an expected payoff of at least  of $(v-c)(1-q)$, thus, $L_2-c\geq (v-c)(1-q)$. Since $L$ is the upper end-point of the support of  $F_2(\cdot)$ and $L$ is also the lower end-point of the support of $\bar{F}_2$, thus,
\begin{align}
(L-c)(1-F_1(L))-s\geq  (v-c)(1-q)-s\nonumber\\
(L-c)(1-qF_1(L))= (v-c)(1-q)
\end{align}
both can not be true simultaneously since $q\neq 1$. Hence, the result follows.
\end{proof}
\subsection{Proof of Theorem~\ref{thm:payoffconstant}}\label{sec:n-1all}
First, we show that there is no NE where each primary selects $Y$ w.p. $1$ (Theorem~\ref{thm:noyyn>2}). Subsequently, we show Theorem~\ref{thm:payoffconstant}.

First, we introduce some notations:
  \begin{definition}\label{defn:w}
  Let $w_{m,n}(x)$ denote the probability that at least $m$ ($0<m<n$) success out of $(n-1)$ events where each event has probability of success of $x$. Thus,
  \begin{align}\label{eq:w}
  w_{m,n}(x)=\sum_{i=m}^{n-1}{n-1 \choose i}x^i(1-x)^{n-1-i}.
  \end{align}
  We also denote
  \begin{align}\label{defn:W}
  W_{m,n}(x)=1-w_{m,n}(x).
  \end{align}
  \end{definition}
  $w_{m,n}(\cdot)$ is strictly increasing and continuous function in the interval $[0,1]$. Thus, $W_{m,n}(\cdot)$ is strictly decreasing and continuous. Note that inverse of $w_{m,n}$ exists and is also increasing and continuous. Also, note that $w_{m,K-1}<w_{m,K}$, thus, $W_{m,K-1}>W_{m,K}$.

\begin{theorem}\label{thm:noyyn>2}
There is no NE where each primary selects $Y$ w.p. $1$. 
\end{theorem}
Assume  all the players select $Y$, so that they know each other\rq{}s channel states. Thus, the competition becomes similar to  {\em Bertrand Competition}\cite{mwg}, i.e. if the number of available channels is less than or equal to $m$, then each primary will set its price at $v$ since the channel of primary will always be sold. Otherwise, each primary will set its price at $c$. Now, the probability that the number of available channels is at most $m$ is $1-w_{m,n}(q)=W_{m,n}(q)$. Thus, the expected payoff of a player is 
\begin{align}\label{eq:payoff_strategyn}
& (v-c)W_{m,n}(q)+(c-c)(1-W_{m,n}(q))-s=(v-c)W_{m,n}(q)-s.
 \end{align} 
 Now consider the following unilateral deviation for a primary: Primary $1$ selects $N$ and sets its price at $v$ w.p. $1$. The channel of primary $1$ will be bought at least when the number of available channels is at most $m$.  Since primary $1$  decides not to incur the cost $s$, thus, its expected payoff is 
 \begin{align}
 (v-c)W_{m,n}(q)
 \end{align}
 This is strictly higher than (\ref{eq:payoff_strategyn}). Hence, the strategy profile can not be an NE.\qed
 
 Note that when $n=2$, we obtain  a similar result (Theorem~\ref{thm:yandy}). The above result shows that there is no NE where each primary selects $Y$ even when $n>2$ and each primary acquires the CSI of all other primaries.
 
 The above theorem states that in a symmetric NE, each primary must select $Y$ w.p. $p$ and $N$ w.p. $1-p$ where $0\leq p<1$. 

Note that when each primary selects $N$ w.p. $1$, then, the expected payoff is $(v-c)W_{m,n}(q)$ which has been shown in \cite{Gaurav1}. Now, we show that even when each primary randomizes between $Y$ and $N$ in an NE with a positive probability, then, the expected payoff is $(v-c)W_{m,n}(q)$. We assume that each primary selects $Y$ w.p. $p$ $0\leq p<1$ and $N$  w.p. $1-p$. Before proving the result, we introduce some notations and state and show some results. 

Note that when $k<m$, then in a NE a primary must select $v$ w.p. $1$ when it selects $Y$ as its channel will always be sold.  Now, we state some structural properties of $F_k(\cdot)$, $k\geq m$ and $F_N$. Let $U_k$ be the upper end-point of $F_k$, for  $k=m,\ldots,n$ and $U_N$ be the upper end-point of $F_N(\cdot)$.
\begin{obs}
$U_N>c$
\end{obs}
\begin{proof}
At a price less than or equal to $c$, a primary gets an expected payoff of at most $0$. However, at $v$, a primary can achieve an expected payoff of $(v-c)W_{m,n}(q)$ which is strictly positive. Hence, a primary will select any price less than or equal to $c$ with $0$ probability. Hence, $U_N>c$. 
\end{proof}
Next result shows that $F_j(\cdot)$, $j\geq m$ and $F_N$ can not have any jump except at $v$. 
\begin{obs}
$F_j(\cdot)$ where $j\geq m$, and $F_N(\cdot)$ can not have a jump at any point except $c$. 
\end{obs}
\begin{proof}
We show the above result for $F_j(\cdot)$ where $j\in \{m,\ldots,n\}$. The proof for $F_N(\cdot)$ will be the same. 

Suppose not i.e. $F_j(\cdot)$ has a jump at $x>c$. However, $j\geq m$, hence, the channel of all primaries will not be bought when everyone selects $x$. Thus, a primary can select a price slightly below $x$ (as $x>c$) which can greatly increase the probability of winning. Hence, a primary can  attain a strictly higher payoff by selecting a price $x-\epsilon$ for small enough $\epsilon>0$ which contradicts the fact that $x$ is a best response. Hence, the result follows. 
\end{proof}
Suppose that a primary selects $Y$ and $j$ number of channels are available among the rest where $j\geq m$. Then, the above result shows that if the primary selects $x>c$, then its expected payoff is 
\begin{align}\label{eq:payoffj}
(x-c)W_{m,j}(pF_j(x)+(1-p)F_N(x))
\end{align}
Let $U_{max}=\max U_k$ for $k\geq m$. 
\begin{lem}
$U_N> U_{max}$.
\end{lem}
\begin{proof}
Note that if $U_{max}=c$, then the result is trivially true.

Thus, we consider the case where $U_{max}>c$. Suppose that it is not true. Without loss of generality, we assume that $U_{max}=U_j$, for some $j\geq m$. Now, when $j$ number of channels are available, then the expected payoff at $U_j$ (from (\ref{eq:payoffj})  is thus,
\begin{align}
(U_j-c)W_{m,j}(pF_j(U_j)+(1-p)F_N(U_j))=0
\end{align}
Since $F_j(U_j)=1$ and $F_N(U_j)=1$ as the upper end-point of $F_N(\cdot)$ is less than that of $F_j(\cdot)$. Since $U_j>c$, thus,  no primary has a jump at $U_j$. Hence, $U_j$ is a best response for the primary when $j$ number of competitors is available. 

  Now, suppose that primary selects any other price $y\in (c,U_N]$, the expected payoff of the primary is 
  \begin{align}
  (y-c)W_{m,j}(pF_j(y)+(1-p)F_N(y))
  \end{align}
  Since $y>c$ and $y\leq U_N$, thus, at least $F_j(y)<1$. Hence, the above expression is strictly positive. Hence, $U_j$ can not be a best response for the primary which contradicts that $U_j$ is the upper end-point of $F_j(\cdot)$.  Thus, we reach to a contradiction. Hence, the result follows.
\end{proof}
\begin{lem}\label{lem:un=v}
$U_N=v$.
\end{lem}
\begin{proof}
Suppose not i.e. $U_N<v$. Since $U_{max}<U_N$, thus, the expected payoff at $U_N$ is
\begin{align}
(U_N-c)W_{m,n}(q)
\end{align}
as if a primary selects the price $U_N$, it will be only sold if less than $m$ number of competitors\rq{} channels are available. Note that $U_N>c$, thus, no primary does not have any jump at $U_N$. Hence, $U_N$ is a best response under $N$ for a primary. 

However, at any price in the interval $x\in (U_N,v)$, the expected payoff is
\begin{align}
(y-c)W_{m,n}(q)
\end{align}
Thus, the expected payoff is strictly higher at $y$. Hence, it contradicts the fact that $U_N$ is a best response when a primary selects $N$. Hence, the result follows.
\end{proof}
Now we are ready to show Theorem~\ref{thm:payoffconstant}.

\textit{Proof}:
Lemma~\ref{lem:un=v} shows that $U_N=v$. On the other hand no primary  has a jump at $v$ as $v>c$. Thus, $v$ is a best response when a primary selects $N$. However, the expected payoff of the primary when it selects $N$ at $v$ is $(v-c)W_{m,n}(q)$. Now, when a primary randomizes between $Y$ and $N$, the expected payoff under $N$ must be equal to the expected payoff under $Y$. Hence, the expected payoff of a primary is $(v-c)W_{m,n}(q)$ when a primary randomizes between $Y$ and $N$.  Note that if both the players select $N$, then the expected payoff of each primary is also $(v-c)W_{m,n}(q)$. 
  \end{document}